\documentclass[letterpaper,11pt]{article}

\usepackage{amsfonts,amsmath,amssymb,amsthm,graphicx}
\usepackage{mathtools}
\usepackage{mathrsfs}
\usepackage{ifthen}
\usepackage{changepage}
\usepackage{enumerate}
\usepackage{scalerel}
\usepackage{accents}
\usepackage[margin=1in]{geometry}
\usepackage{enumitem}
\usepackage{natbib}
\usepackage{subcaption} \usepackage{float}      

\usepackage[utf8]{inputenc} \usepackage[T1]{fontenc} 
\usepackage{csquotes}
\usepackage{comment}
\usepackage{bbm}
\usepackage{lmodern}
\usepackage{array}
\usepackage{makecell}
\usepackage{csquotes}
\usepackage{booktabs}
\usepackage{xfrac}

\usepackage{tcolorbox}

\usepackage[colorlinks=true, linkcolor= blue!70!black, citecolor=blue!70!black,urlcolor=blue,breaklinks=true]{hyperref}

\usepackage{thm-restate}

\usepackage{cleveref}
\crefname{claim}{claim}{claims}

\usepackage{tikz}
\usepackage{pgfplots}
\pgfplotsset{compat=1.18}
\usetikzlibrary{patterns}
\usepgfplotslibrary{fillbetween}
\usetikzlibrary{intersections}
\usepackage{fix-cm}
\usepackage[ruled,vlined,linesnumbered]{algorithm2e}
\Crefname{algocf}{Algorithm}{Algorithms}
\allowdisplaybreaks
\usetikzlibrary{arrows.meta}

\allowdisplaybreaks
\theoremstyle{plain}
\newtheorem{theorem}{Theorem}[section]
\newtheorem{lemma}[theorem]{Lemma}

\newtheorem{proposition}[theorem]{Proposition}
\newtheorem{corollary}[theorem]{Corollary}

\theoremstyle{plain}
\newtheorem{definition}{Definition}[section] \newtheorem{example}[definition]{Example}

\theoremstyle{plain}
\newtheorem{assumption}{Assumption}

\newcounter{relctr} \everydisplay\expandafter{\the\everydisplay\setcounter{relctr}{0}}  
\AtBeginDocument{}

\newcommand{\squishlist}{
\begin{list}{{{\small{$\bullet$}}}}
{\setlength{\itemsep}{3pt}      \setlength{\parsep}{1pt}
\setlength{\topsep}{1pt}       \setlength{\partopsep}{0pt}
\setlength{\leftmargin}{1em} \setlength{\labelwidth}{1em}
\setlength{\labelsep}{0.5em} } }
\newcommand{\squishend}{  \end{list}}

\newcommand{\quality}{\omega}

\newcommand{\prior}{F}
\newcommand{\priorPDF}{f}
\newcommand{\distOfMean}{G}

\newcommand{\type}{\theta}

\newcommand{\typeCDF}{D}

\newcommand{\price}{p}
\newcommand{\buyerUtility}{v}

\newcommand{\posterior}{\mu}

\newcommand{\signalSpace}{\Sigma}
\newcommand{\signalscheme}{\phi}
\newcommand{\pricingscheme}{t}
\newcommand{\action}{a}

\newcommand{\signal}{\sigma}

\newcommand{\Rev}{\textsc{Rev}}

\newcommand{\inverseVal}{\kappa}

\newcommand{\NI}{\textsc{NI}}
\newcommand{\FI}{\textsc{FI}}
\newcommand{\noinfor}{\signalscheme^{\textsc{\NI}}}
\newcommand{\noinforPrice}{p^{\NI}}
\newcommand{\fullinfor}{\signalscheme^{\textsc{\FI}}}

\newcommand{\OPT}{\textsc{Opt}}
\newcommand{\voa}{\Gamma}

\newcommand{\typefn}{u}
\newcommand{\qualityfnone}{g_1}
\newcommand{\qualityfntwo}{g_2}

\newcommand{\Myer}{\textsc{Myer}}
\newcommand{\MyerPrice}{p^{\Myer}}

\newcommand{\priormean}{\bar{\quality}}

\newcommand{\signalnum}{K}
\newcommand{\quantileidx}{r}
\newcommand{\disclosurefn}{\Phi}

\newcommand{\disclosurefnClass}{\mathcal{X}}
\newcommand{\buyerUtilityClass}{\mathcal{V}}
\newcommand{\pooledmean}{x}
\newcommand{\kink}{c}
\newcommand{\cvxmeasure}{\nu}

\newcommand{\binrange}{I}

\newcommand{\auxratio}{s}

\newcommand{\cvxspace}{\mathcal{H}}
\newcommand{\cvxspacesandwich}{\mathcal{H}_{\textsc{sand}}}

\newcommand{\cvx}{h}
\newcommand{\newcvx}{\cvx\primed}

\newcommand{\RDD}{\textsc{Rdd}}

\newcommand{\basefn}{\lambda}
\newcommand{\recfn}{\Lambda}
\newcommand{\ext}{\textsc{ext}}

\newcommand{\KQuantilePartitionDisclosure}{{\sf $K$-Quantile Partition}}
\newcommand{\KQualityPartitionDisclosure}{{\sf $K$-Quality Partition}}
\newcommand{\naturals}{\mathbb{N}}
\newcommand{\quant}{Q}
\newcommand{\quants}{\mathbf{\quant}}
\newcommand{\qualities}{\boldsymbol{\quality}}
\newcommand{\partitionprobs}{\boldsymbol{\xi}}

\newcommand{\quantDisClass}{\disclosurefnClass_\signalnum}

\newcommand{\RQDQuant}{\text{\hyperref[eq:RQD]{$\textsc{Rqd}[\quantDisClass]$}}}

\newcommand{\RDDQuant}{\text{\hyperref[eq:rdd]{$\textsc{Rdd}[\quantDisClass,\cvxspace]$}}}

\newcommand{\RQDQuantTitle}{\text{\hyperref[eq:RQD]{$\textsc{RQD}[\quantDisClass]$}}}

\newcommand{\RDDQuantTitle}{\text{\hyperref[eq:rdd]{$\textsc{RDD}[\quantDisClass,\cvxspace]$}}}

\newcommand{\reals}{\mathbb{R}}
\newcommand{\RQD}{\textsc{Rqd}}

\newcommand{\RDDQuantsand}{\text{\hyperref[eq:rdd]{$\textsc{Rdd}[\disclosurefnClass_\signalnum,\cvxspacesandwich]$}}}

\newcommand{\RQDQuality}{\text{\hyperref[eq:RQD]{$\textsc{Rqd}[\disclosurefnClass_\signalnum\primed]$}}}

\newcommand{\cmrs}{\textsc{cmrs}}
\newcommand{\cmrsfn}{\tau}
\renewcommand{\qualityfnone}{b}
\renewcommand{\qualityfntwo}{a}
\newcommand{\newtypefn}{\nu} 

\newcommand{\xhdr}[1]{\vspace{6pt}\noindent{\bf {#1.}}}

\newcommand{\omt}[1]{}

\newcommand{\cc}[1]{\ensuremath{\mathsf{#1}}}

\newcommand{\prob}[2][]{\mathbb{P}\ifthenelse{\not\equal{}{#1}}{_{#1}}{}\!\left[{\def\givenn{\middle|}#2}\right]}
\newcommand{\expect}[2][]{\mathbb{E}\ifthenelse{\not\equal{}{#1}}{_{#1}}{}\!\left[{\def\givenn{\middle|}#2}\right]}
\newcommand{\indicator}[2][]{\mathbf 1\ifthenelse{\not\equal{}{#1}}{_{#1}}{}\!\left[{\def\givenn{\middle|}#2}\right]}

\newcommand{\R}{\mathbb{R}}
\newcommand{\E}{\mathbb{E}}

\newcommand{\supp}{\cc{supp}}

\newcommand{\primed}{^{\dagger}}
\newcommand{\doubleprimed}{^{\ddagger}}

\newcommand{\dd}{{\mathrm d}}

 \DeclareMathOperator*{\argmax}{arg\,max}

\usepackage{natbib}
\bibpunct{(}{)}{;}{a}{,}{,}

\begin{document}

\title{Simple and Robust Quality Disclosure: \\
The Power of Quantile Partition}
\author{
Shipra Agrawal\thanks{Columbia University. 
Email: {\tt sa3305@columbia.edu}}
\and
Yiding Feng\thanks{Hong Kong University of Science and Technology. Email: {\tt ydfeng@ust.hk}}
\and
Wei Tang\thanks{Chinese University of Hong Kong. Email: {\tt weitang@cuhk.edu.hk}}}
\date{}

\maketitle

\begin{abstract}
    
Quality information on online platforms is often conveyed through simple, percentile-based badges and tiers that remain stable across different market environments. 
Motivated by this empirical evidence, 
we study robust quality disclosure in a market where a platform commits to a public disclosure policy mapping the seller's product quality into a signal, and the seller subsequently sets a downstream monopoly price.
Buyers have heterogeneous private types and valuations that are linear in quality. We evaluate a disclosure policy via a minimax competitive ratio: its worst-case revenue relative to the Bayesian-optimal disclosure-and-pricing benchmark, uniformly over all prior quality distributions, type distributions, and admissible valuations.

Our main results provide a sharp theoretical justification for quantile-partition disclosure. 
For $\signalnum$-quantile partition policies, we fully characterize the robust optimum: the optimal worst-case ratio is pinned down by a one-dimensional fixed-point equation and the optimal thresholds follow a backward recursion. 
We also give an explicit formula for the robust ratio of any quantile partition as a simple ``max-over-bins'' expression, which explains why the robust-optimal partition allocates finer resolution to upper quantiles and yields tight guarantees such as $1 + \frac{1}{\signalnum}$ for uniform percentile buckets. 
In contrast, we show a robustness limit for finite-signal monotone (quality-threshold) partitions, which cannot beat a factor-$2$ approximation. 
Technically, our analysis reduces the robust quality disclosure to a robust disclosure design program by establishing a tight functional characterization of all feasible indirect revenue functions.
 \end{abstract}

\newpage
\section{Introduction}

Information about product quality is central to consumer purchase decisions. 
BrightLocal's 2025 Local Consumer Review Survey finds that only 4\% of consumers say they never read online business reviews.\footnote{See the BrightLocal's 2025 Local Consumer Review Survey \href{https://www.brightlocal.com/research/local-consumer-review-survey/}{here}.}
In the modern online markets, consumers are inundated with quality-related signals, ranging from platform-generated badges such as ``verified purchase,'' ``top rated,'' or ``recommended,'' to star ratings, written reviews, expert and influencer endorsements, and third-party consumer reports. 
These disclosures have become so pervasive that many online platforms now actively mediate quality information as an integral part of the consumers' shopping experience.

Yet, despite their importance, the signals used in practice are often deliberately simple and remarkably robust across diverse market environments.
Upwork, for instance, assigns percentile-based talent badges to freelancers, such as ``Top Rated'' (top 10\%), ``Top Rated Plus'' (top 3\%), and ``Expert-Vetted'' (top 1\%), and deploys the same badge system across a highly heterogeneous marketplace spanning many skill categories (e.g., software development, graphic design, writing) \citep{Upwork}.
Airbnb similarly highlights ``Guest Favourites'' by marking eligible homes as being among the top 1\%, 5\%,  10\% and bottom 10\% of available listings within each local market where the same coarse ranking signal is used across different cities, neighborhoods, and property types \citep{AirbnbTopHomesHighlight}.
Such simple and robust disclosure policies have clear several practical advantages: they are easy to implement at scale, readily interpretable to consumers, and require limited coordination between platforms and sellers.

Because quality disclosure directly shapes consumer beliefs, it also influences demand, matching efficiency, and ultimately trade outcomes.
Although the study of quality disclosure dates back to the seminal contributions by \cite{A-78,G-81,M-81}, and has since been extensively explored in conjunction with other market and mechanism-design primitives (see, e.g., \citealp{ES-07,BS-12,BKP-12,AHLS-22,GL-24}), 
much of this literature evaluates disclosure rules under fairly detailed knowledge about the market environment. 
In practice, however, platforms often operate across many categories, regions, and submarkets, and cannot reasonably rely on fine-grained knowledge of either the quality distribution or the demand side (e.g., buyers' preferences). 
Motivated by this perspective, we focus on the following two fundamental practical questions: 
\begin{displayquote}
    \emph{When simplicity and robustness are imposed as constraints, what is the best disclosure rule?
    How close can such simple and robust disclosure rules come to the revenue performance of fully optimal disclosure that is tailored to the precise market environment?}
\end{displayquote}

To formalize this robustness agenda, we introduce a robust quality disclosure framework in which the platform commits to a simple disclosure policy and we evaluate its performance by its worst-case approximation to the optimal revenue across admissible environments.
In our setting, an online platform (or ``market designer'') mediates information about a seller’s product quality. 
The product has an intrinsic quality level, and the platform commits to a quality-disclosure policy that maps the realized quality into a public signal observed by all buyers. After a signal is realized, the seller posts a signal-contingent price, potentially a different posted price for different signals, which is chosen by maximizing the expected revenue given the buyers' type distribution and the posterior induced by the signal. 
Buyers are heterogeneous: each buyer has a private type capturing their willingness to pay, and their valuation depends on both their type and the (unknown) product quality. 
Throughout, we focus on \emph{linear} valuation functions, where a buyer's value is linear in product quality.\footnote{In \Cref{apx:extensions general val}, we discuss how our results extend beyond linear valuations to a broader class of valuation functions.}
Upon observing the signal, a buyer forms a posterior belief about quality and purchases the product if and only if the posted price does not exceed their (posterior) expected valuation. 
The seller's expected revenue is then the posted price times the induced demand under that signal, aggregated over signals and quality realizations.
Our goal is to understand how effectively simple disclosure policies can approximate the performance of optimal ones, uniformly across admissible market environments.

\subsection{Our Contributions and Techniques}

\xhdr{A robust quality disclosure framework}
As our first contribution, 
we introduce and study the robust quality disclosure problem:
the platform seeks a simple disclosure policy from a prescribed class $\disclosurefnClass$
that performs well uniformly across all admissible market primitives, namely, all prior quality distributions, all consumer type distributions, and all admissible valuation functions. The class $\disclosurefnClass$ serves as our modeling device for ``simplicity'' along two dimensions:
\begin{itemize}
    \item Detail-freeness -- 
We let $\disclosurefnClass$ to encode requirements that the platform's disclosure rule not rely on fine details of the demand side.\footnote{Our detail-freeness requirement on the policy class $\disclosurefnClass$ echoes the ``Wilson doctrine'' in mechanism design \citep{W-85}, which advocates mechanism-design prescriptions that are as detail-free as possible.}
\item 
    Structural simplicity -- We restrict $\disclosurefnClass$ to simple policies that are widely used in practice, such as quantile or monotone partitions.
\end{itemize}
Performance is measured by the {\em minimax competitive ratio}, defined as the worst-case ratio between the optimal Bayesian revenue (achievable when the market designer can tailor disclosure policy using full knowledge of the quality distribution, the consumers' valuation function, and their type distribution) and the revenue generated by the chosen disclosure policy coupled with optimal downstream pricing (optimized for that same valuation function and type distribution) (see \Cref{def:robust disclosure framework}). 
This formulation captures the practical reality that platforms must design disclosure mechanisms under limited knowledge about both the supply side (quality) and the demand side (preferences), yet still hope to ensure strong performance guarantees in the worst case.

\xhdr{Optimal robust quantile-partition policy}
Our first set of results identifies a natural and practically relevant family of ``simple'' disclosure rules, i.e., \emph{$\signalnum$-quantile
partition policies} where $\signalnum$ is the number of partitional intervals.
Here, a disclosure policy is a rule that takes a prior quality distribution as input and outputs a signaling scheme, a mapping from product quality realizations to probabilistic signals observed by buyers.
A $\signalnum$-quantile partition disclosure policy
fixes quantile thresholds $0 = \quant_0 \le \quant_1 \le \cdots \le \quant_\signalnum= 1$ and, for every prior quality distribution, pools qualities whose \emph{prior quantiles} fall in the same interval $[\quant_{r-1},\quant_r]$ into one public signal.
Within this class, we fully characterize the \emph{robustly optimal} policy: the optimal worst-case
competitive ratio $\voa_\signalnum^*$ is uniquely determined as the solution to a scalar fixed-point equation involving a $\signalnum$-fold composition of a specific function, and the corresponding optimal thresholds are derived by a simple backward recursion (see \Cref{thm:opt quantile partition:optimal policy}).

Our characterization yields a sharp
quantitative message: this ratio $\voa_\signalnum^*$
decreases with $\signalnum$, starting at 2 for $\signalnum = 1$ (i.e., no information) and converging to $1$ at a rate of $\voa_\signalnum^* = 1 + \Theta(1/\signalnum)$ (see \Cref{prop:optimal robust CR monotonicity}). 
Notably, our results show that even with just $\signalnum = 5$ signals, the optimal policy guarantees at least $91.71\%$ of the Bayesian optimal revenue uniformly across all prior quality distributions, consumer type distributions, and linear valuation functions.
Taken together, these results offer a formal theoretical justification for quantile-partition disclosure rules (e.g., ``top 10\%, 5\%, 1\%'' labels) that are widely adopted on real-world platforms: despite their simplicity, they provably approximate the Bayesian-optimal revenue uniformly over market environments.

\xhdr{Robustness of general $\signalnum$-quantile partition}
Our second set of results provides an explicit, easy-to-use performance formula for any $\signalnum$-quantile partition policy. 
In particular, the robust competitive ratio depends only on the threshold profile $\quants = (\quant_0, \quant_1, \ldots, \quant_\signalnum)$ via a remarkably simple ``max-over-bins'' expression of the form (see \Cref{thm:quant partition:general policy}):
\begin{align*}
    \max_{\quantileidx \in [\signalnum]}~
    \left(
    1 + \frac{\quant_\quantileidx - \quant_{\quantileidx - 1}}{\left(1+\sqrt{1-\quant_\quantileidx}\right)^2}
    \right)~,
\end{align*}
which identifies exactly where the worst case comes from: an adversarial instance can concentrate all the revenue loss on a single quantile partition bin, and the overall guarantee is therefore governed by the interval that maximizes the derived expression.
The formula also clarifies why the optimal policy allocates quantile partitions unevenly: to minimize the worst-case ratio, it must equalize the above expression across all bins, which necessarily makes the quantile intervals progressively narrower (i.e., finer discrimination) in the upper quantiles and coarser intervals in the lower quantiles (see \Cref{prop:optimal robust policy:decreasing margin}).
This characterization provides a clean performance guarantee for the kind of percentile-tier disclosures used in practice, it also shows that even the most basic uniform percentile buckets (e.g., $10\%, 20\%, \ldots$) achieve a ratio $1 + \frac{1}{\signalnum}$ exactly, which is tight asymptotically w.r.t.\ $\signalnum$ (see \Cref{cor:uniform quantile disclosure}).

\xhdr{Robustness limit of quality-threshold partitions} 
We also examine alternative ``simple'' policies such as those based on finite-signal monotone partitions on quality space.
In contrast to quantile partitions, we show that no such finite-signal monotone-partition policy can guarantee better than a factor-$2$ approximation, regardless of how the quality thresholds are chosen (see \Cref{thm:opt quality partition}).
This underscores that the quantile-based disclosure policy, i.e., pooling by prior ranks rather than raw quality levels, is what enables robust improvement with finitely many signals, by adapting the partition to the prior quality distribution while still remaining simple and interpretable.

\xhdr{Our techniques}
Our analysis proceeds by building a two-step reduction that decouples the economic modeling in our problem from the worst-case optimization.
First, a key observation with the affine structure of consumers' valuations is that a signal affects the revenue only through the \emph{posterior mean} of quality it induces. 
Conditional on a posterior mean $\pooledmean$, the revenue under the optimal downstream pricing against a population of heterogeneous buyers can be summarized by an \emph{indirect revenue} function $\cvx(\pooledmean)$, defined as the seller's optimal posted-price revenue when all
buyers share the belief that expected quality equals $\pooledmean$. We show that every such $\cvx$ generated by
any linear valuation function and type distribution is non-negative, non-decreasing, and convex.
This reduces the original economic worst-case optimization to bounding, uniformly over priors $\prior$ and feasible indirect utilities $\cvx$, the worst-case ratio $\expect{\cvx(\pooledmean)}$ computed under the prior quality distribution to $\expect{\cvx(\pooledmean)}$ computed under the distribution of posterior means induced by the disclosure policy.

Building on this observation, we introduce a general minimax framework, i.e., \emph{robust disclosure design}  (see \Cref{def:robust disclosure design}), that abstracts away underlying economic primitives: the designer selects a disclosure policy from a prescribed class while optimizing against the worst case over both (i) the prior and (ii) a convex functional class of all feasible indirect utilities. 
In this way, the economics of our quality-disclosure model is encoded entirely in the functional constraints imposed on the indirect utility functions $\cvx$; once these constraints are in place, the remaining task becomes a purely geometric optimization problem over a convex set of functions.
A central technical step is then to establish a tight connection between this abstract program and the original robust quality disclosure program: for $\signalnum$-quantile partitions, 
solving the robust disclosure design program yields bounds that are not only upper bounds but also exactly achievable in the original economic setting.

Formally, robust disclosure design is a relaxation of the robust quality disclosure problem, because it replaces the set of indirect revenue functions induced by admissible valuation and type distributions with a larger convex functional space. 
Perhaps surprisingly, we show that this relaxation is nevertheless tight once we impose a natural linear lower bound satisfied by all induced indirect revenue functions
(i.e., $\cvx(\pooledmean)\ge \pooledmean\cdot \cvx(1)$).
Intuitively, this inequality captures a minimal ``monotone scaling'' property of revenue with respect to expected quality, and it turns out to be the only additional structure needed to pin down the worst case.
This reduction, translating our economic quality-disclosure model into functional constraints on the induced indirect revenue function, is the central tool behind our tractable, exact characterization.

The remaining task then reduces to solve the robust disclosure design problem (see \Cref{thm:rdd:optimal robust policy}), which we use the following two tools. 
First, after normalizing by scale (e.g.,
setting $\cvx(1)=1$), we prove an \emph{extreme-point reduction}: every feasible indirect function $\cvx$ admits an integral representation as a mixture of hinge functions $\cvx_\kink(\pooledmean)=\max\{\kink,\cvx\}$, and the worst-case ratio is attained at such an extreme point, which significantly reduces the inner optimization over functions to an optimization over a single parameter $\kink\in[0,1]$ that parameterizes the hinge functions (see \Cref{lem:stoploss-representation} anad \Cref{lem:reduce-to-extremal}).
Second, exploiting the piecewise-linear structure of extremal function $\cvx_\kink$, we show that for a fixed quantile partition, only the unique bin (i.e., the quantile partition interval) containing the cutoff $\kink$ can generate the relevant loss compared to the benchmark, leading directly to the closed-form max-over-bins expression above. 
We then show that optimizing the thresholds becomes a balancing argument: minimax optimality requires equalizing the bin-wise worst-case terms, which yields the backward recursion and the scalar equation characterizing $\voa_\signalnum^*$. 
Finally, we establish the tightness for the original robust quality disclosure problem by explicitly constructing model primitives (including the prior quality distribution, a valuation function and a type distribution) that implement the hinge-shaped indirect utilities and match the upper bound
(see \Cref{lem:quantile partition:matching lower bound}).

\subsection{Related Work}

\newcommand{\posscite}[1]{\citeauthor{#1}'s \citeyearpar{#1}}

\xhdr{Quality disclosure}
The study of quality disclosure in markets with uncertain product quality has a long history in information economics,
dating back to the seminal work of the 
\posscite{A-78} ``lemons'' model,
which shows how quality uncertainty and asymmetric information can generate adverse selection.
Subsequent work by \cite{G-81,M-81} studies the implications of mandatory or voluntary disclosure under unraveling, and shows that when the sender cannot commit to a disclosure rule, strategic incentives can lead to full disclosure as the unique equilibrium outcome.
Later work then studies the settings where the seller/sender can commit to a disclosure policy before observing private information \citep{RS-10,KG-11} 
and shows that partial disclosure can arise in equilibrium.
Our work joins the later strand of where the seller commits to a quality disclosure policy (and then sets signal-contingent monopoly prices).

More recently, a growing literature has examined disclosure policies that interact with pricing and broader market or mechanism design primitives.
Notable examples include the public disclosure in second-price auctions with probabilistic goods \citep{BS-12,EFGPT-14},
the sale of hard information \citep{AHLS-22}, 
and platform-mediated pricing for disclosure services \citep{GL-24}. 
From a mechanism design perspective, \cite{RS-10} analyze optimal disclosure in a monopolistic setting with endogenous signaling and demonstrate that partial disclosure can be revenue-maximizing. 
\cite{LS-17} show that when a seller can disclose type-dependent additional information about the buyer's valuation distribution, discriminatory disclosure may strictly outperform full disclosure in terms of revenue. 
\cite{WG-24,GHS-25} further study joint design of transfers/prices and information, highlighting the role of information discrimination and ``price-experiment'' menus. 
\cite{LPY-25} study vertical differentiation where a (literal or metaphorical) designer chooses a public signal about product quality before firms set prices.
Their model can be viewed as an extension of our Bayesian optimal signaling-and-pricing benchmark to a competitive environment with multiple sellers.

Our work focuses signal-contingent pricing, wherein the seller posts a (possibly different) price for each realized signal, which also aligns with prior works \citep{RS-10,LPY-25}.
This framework abstracts away from general transfers, sequential screening, or type-contingent ``experiments,'' and instead focuses on a practically prevalent design space, namely, the combination of public disclosure and posted pricing. 
Within this setting, we aim to understand how well simple disclosure policies can perform in a robust manner, uniformly across all prior quality distributions and consumer preferences (including valuation functions and type distributions).

\xhdr{Robust information design}
There has been a growing interest in designing information policies that are robust to uncertainty about model primitives. 
A prominent strand of this literature studies robustness to unknown receiver utilities, often framed in terms of regret minimization or online learning in persuasion settings 
(e.g., \citealp{CCMG-20,CMCG-21,ZIX-21,FTX-22,AFT-23}).
Other works consider robustness to uncertainty about receiver types (see, e.g., \citealp{BTXZ-22, FHT-24}), or to ambiguity in the receiver’s prior belief, typically by optimizing against the worst-case prior within a given set (see, e.g., \citealp{DP-22,K-22,LL-25}).

Our paper contributes to this robustness agenda with focusing on a market-oriented setting of quality disclosure coupled with pricing. 
Crucially, our notion of robustness is simultaneously uniform over all prior quality distributions (i.e., the prior distributions), all admissible consumer valuation functions (i.e., receiver utilities), and all type distributions. As a result, our notion subsumes and unifies several dimensions of uncertainty that are typically treated in isolation in much of the existing literature. 
This makes our robustness criterion notably stronger and more comprehensive than those in prior work, which usually fix all but one source of model uncertainty.

\xhdr{Simple v.s.\ optimal information/mechanism design}
Our focus on approximation guarantees for simple disclosure policies parallels the broader ``simple v.s.\ optimal'' paradigm in mechanism design, which asks how much performance is lost when restricting to simple and interpretable mechanisms (see, e.g., \citealp{CHMS-10,HL-10,AHNPY-19,FJ-24}, and the surveys of \citealp{HR-09,CS-14,Rou-15} for an overview). 
A related line of work has recently developed in information design, studying when simple or structured signaling schemes can approximate optimal persuasion outcomes (see, e.g., \citealp{GHHS-21,M-21,BCVZ-22,KLZ-25,CLTT-25,CLW-25}).

A major focus of this paper is on quantile-based simple disclosure rules (e.g., percentile tiers / quantile partitions), which have also appeared in previous works under different contexts. 
\cite{SY-24}, for example, study privacy-preserving signaling and propose a canonical disclosure that reveals only the rank of the underlying information; they show that any privacy-preserving disclosure can be obtained by garbling an appropriately reordered quantile signal, which in general is not a finite partitional policy.
\cite{BHMSW-22} show that in a classic second-price auction, the revenue-maximizing disclosure rule takes an explicit quantile-threshold form, fully separating low values while pooling sufficiently high values.
Finally, our interest in partitional disclosure more broadly connects to work on monotone categorisations of quality \citep{OR-23} and to the broader information-design literature studying when monotone/partitional structures are optimal (see, e.g., \citealp{DM-19, KLZ-25}).
Unlike these works, our work emphasize robust approximation guarantees for simple quantile partitions rather than environment-specific optimality.

\section{Preliminaries}
\label{sec:prelim}

In this paper, we study the \emph{robust quality disclosure} problem. In the following, we describe various components of our model.

\subsection{Basic Model}

\xhdr{Environment}
A monopolist seller offers a product with varying quality to a continuum population of heterogeneous consumers. 
We consider a Bayesian model in which the product is associated with a quality $\quality \in [0,1]$,\footnote{For ease of presentation, we assume that both product quality $\quality$ and consumer type $\type$ lie in $[0,1]$. 
Our results and analysis extend straightforwardly to the general setting in which quality and type take values in $\reals_+$.} drawn from a discrete or continuous prior distribution with cumulative distribution function (CDF)\footnote{Following the convention in the literature, we define CDF $\prior$ of a distribution $\prior$ as $\prior(\quality)\triangleq \prob[x\sim\prior]{x<\quality}$.} $\prior\in\Delta([0,1])$. 
The quality distribution $\prior$  is common knowledge to both the seller and the consumers. We denote by $\priorPDF(\quality)$ the prior probability mass (in the discrete case) or probability density (in the continuous case) associated with quality realization $\quality \sim \prior$.

consumers have heterogeneous preferences (i.e., different willingness to pay) for the product. In particular, Each consumer's preference is captured by a \emph{private} type $\type$, drawn independently and identically (i.i.d.) from a type distribution with CDF $\typeCDF \in \Delta([0, 1])$. The seller does not observe individual realizations of $\type$ but knows the type distribution $\typeCDF$. Quality $\quality$ and consumer type $\type$ are assumed independent.

A consumer with type $\type \in [0, 1]$ makes a once-and-for-all purchase decision: he either buys the product or abstains from purchasing. If he buys, his payoff is given by the quasi-linear utility function $\buyerUtility(\type, \quality) - \price$, where $\buyerUtility(\type, \quality)$ denotes the consumer's non-negative valuation for a product of quality $\quality$ when his type is $\type$, and $\price$ is the price charged by the seller (which may be set randomly). If he does not buy, his payoff is zero.

Throughout this paper, we impose the following assumption on the consumer's valuation function. 
We also let $\buyerUtilityClass$ denote the class of all valuation functions satisfying \Cref{asp:linear valuation function}.\footnote{
Typical valuation functions include $\buyerUtility(\type, \quality) =\type + \quality$ \citep[cf.][]{IMSZ-19,KMZL-17,CIMS-17,SVZ-22} and the multiplicative function $\buyerUtility(\type, \quality) =\type\quality + \type$ \citep[cf.][]{CS-21,LSX-21,BHM-26}.
We discuss how our results generalize in the when the valuation function is beyond linear w.r.t.\ quality in \Cref{apx:extensions general val}. 
}
\begin{assumption}[Linear valuation function]
\label{asp:linear valuation function}
The valuation function $\buyerUtility(\type, \quality)$ is affine and increasing in quality $\quality \in [0,1]$ for every fixed type $\type \in [0,1]$, and non-decreasing in type $\type$ for every fixed quality $\quality \in [0,1]$.
\end{assumption}

\xhdr{Bayesian optimal signaling \& pricing design}
A platform (market designer) controls how the seller's quality information is disclosed by 
flexibly designing a \emph{signaling scheme} that strategically reveals information about the product quality $\quality$ to consumers. Specifically, a signaling scheme $\signalscheme = \{\signalscheme(\cdot|\quality)\}_{\quality\in[0, 1]}$ is defined by a signal space $\signalSpace$, where for each $\quality \in [0, 1]$, the mapping $\signalscheme(\cdot|\quality) \in \Delta(\signalSpace)$ specifies the conditional distribution over signals in $\signalSpace$ sent to consumers when the realized quality is $\quality$. 
After the signal $\signal\in\signalSpace$ is realized, the seller posts a signal-contingent price \citep{RS-10} according to a pricing scheme $\pricingscheme: \signalSpace \rightarrow \reals_+$, where $\pricingscheme(\signal)$ is the nonnegative posted price associated with signal $\signal$.
Throughout, we evaluate disclosure policies by the seller's resulting expected revenue. In other words, in this Bayesian environment, the platform's objective is to maximize this expected revenue.
\footnote{Note that, it is without loss of generality to consider pricing schemes that depend only on the realized signal (and not the quality), since the seller and the platform share the same objective in this Bayesian environment.}

Consumers do not observe the realized quality $\quality$, but they do observe the signal $\signal \sim \signalscheme(\cdot|\quality)$ generated by the seller. Using this signal and the quality distribution (prior) $\prior$, they form a Bayesian posterior distribution over quality: 
$\posterior(\quality|\signal) \propto \signalscheme(\signal|\quality) \cdot \priorPDF(\quality)$.
Conditional on the observed signal $\signal$, the price $\pricingscheme(\signal)$ conveys no additional information about the product. A consumer with type $\type$ purchases the product if and only if his expected valuation under the posterior,
$\expect[\quality \sim \posterior(\cdot|\signal)]{\buyerUtility(\type, \quality)}$, is at least the posted price $\pricingscheme(\signal)$.

Given a signaling scheme $\signalscheme$ and its corresponding pricing scheme $\pricingscheme$, let $\Rev(\signalscheme, \pricingscheme)$ denote the seller's expected revenue, defined as
\begin{align*}
    \Rev(\signalscheme, \pricingscheme) \triangleq
    \expect[\quality\sim \prior, \signal\sim\signalscheme(\cdot\mid \quality), \type\sim\typeCDF]{\pricingscheme(\signal) \cdot \action^*(\type,\signal)}~,
\end{align*}
where $\action^*(\type,\signal) \in \{0, 1\}$ indicates the purchase decision of a consumer with type $\type$ upon observing signal $\signal$.

We remark that, for a fixed signaling scheme $\signalscheme$, the revenue-maximizing pricing scheme is well understood: upon observing a realized signal $\signal$, the seller should set the price equal to the optimal posted price given consumers' posterior belief $\posterior(\cdot\mid\signal)$. Slightly abusing notation, we write $\Rev(\signalscheme)$ to denote the seller's expected revenue under the signaling scheme $\signalscheme$ after optimizing over prices:
\begin{align*}
    \Rev(\signalscheme) \triangleq
    \sup\nolimits_{\{\pricingscheme(\signal)\}_{\signal\in\signalSpace}}
    \Rev(\signalscheme, \pricingscheme)~.
\end{align*}
Finally, let $\OPT$ denote the optimal revenue achievable by any signaling scheme:
\begin{align*}
    \OPT \triangleq \sup\nolimits_{\signalscheme}
    \Rev(\signalscheme)~.
\end{align*}
As we show in \Cref{cor:full infor optimal no infor 2 approx}, the supremum is attained, and we refer to the signaling scheme that achieves the maximum revenue as the Bayesian optimal signaling scheme.

\subsection{Robust Quality Disclosure Framework}
In practice, the platform may wish to disclose quality in a manner that is both simple and robust. To formalize this desideratum, we introduce the notions of a \emph{disclosure policy} and the \emph{robust quality disclosure} problem.

\begin{definition}[Disclosure policy]
\label{def:disclosure policy}
A \emph{disclosure policy} $\disclosurefn$ takes as input a quality distribution $\prior \in \Delta([0, 1])$ and outputs a signaling scheme $\signalscheme = \disclosurefn(\prior)$.
\end{definition}

\begin{definition}[Robust quality disclosure]
    \label{def:robust disclosure framework}
    Given a class $\disclosurefnClass$ of disclosure policies, the \emph{robust quality disclosure} problem is defined as
    \begin{align}
        \tag{$\RQD[\disclosurefnClass]$}
        \label{eq:RQD}
        \inf_{\disclosurefn \in \disclosurefnClass} \sup_{\substack{
        \buyerUtility \in \buyerUtilityClass, \\
        \prior,\typeCDF \in \Delta([0, 1])}}
        \;\;
        \frac{\OPT_{\buyerUtility,\prior,\typeCDF}}{\Rev_{\buyerUtility,\prior,\typeCDF}(\disclosurefn(\prior))}~,
    \end{align}
    where $\OPT_{\buyerUtility,\prior,\typeCDF}$ denotes the seller’s expected revenue under the Bayesian optimal signaling scheme, and $\Rev_{\buyerUtility,\prior,\typeCDF}(\disclosurefn(\prior))$ denotes the expected revenue under the signaling scheme $\disclosurefn(\prior)$ induced by the disclosure policy $\disclosurefn$, in the setting where the buyer’s valuation function is $\buyerUtility$, the type distribution is $\typeCDF$, and the product's quality distribution is $\prior$.

    A disclosure policy $\disclosurefn^*$ that attains the infimum in~\ref{eq:RQD} (if one exists) is referred to as an \emph{optimal robust disclosure policy} of \ref{eq:RQD}, and the corresponding objective value is referred to as the \emph{optimal robust competitive ratio} of \ref{eq:RQD}.
\end{definition}
Despite its clean mathematical formulation, we make the following remarks regarding the robust quality disclosure problem~\ref{eq:RQD}, which constitutes a minimax optimization program.

In this definition, the problem is parameterized by a class of disclosure policies $\disclosurefnClass$. By varying $\disclosurefnClass$, program~\ref{eq:RQD} becomes a unified tool for evaluating different families of ``simple'' signaling schemes studied in the literature.
For instance, one may restrict $\disclosurefnClass$ to contain a single policy that maps every prior to its corresponding full-information signaling scheme (respectively, no-information signaling scheme).\footnote{\label{footnote:full info and no info}The {full-information signaling scheme} $\fullinfor$ uses signal space $\signalSpace = [0, 1]$ and satisfies $\signalscheme(\quality \mid \quality) = 1$ for all $\quality \in [0, 1]$. The {no-information signaling scheme} $\noinfor$ uses a singleton signal space $\signalSpace = \{\signal\}$ and satisfies $\signalscheme(\signal \mid \quality) = 1$ for all $\quality \in [0, 1]$.} Then program~\ref{eq:RQD} simply returns the worst-case approximation factor of full information (respectively, no information) relative to the Bayesian optimum. 
Alternatively, one may let $\disclosurefnClass$ include all policies that output only monotone-partition signaling schemes \citep{KLZ-25}, censorship signaling schemes \citep{KMZ-22}, or signaling schemes using at most $\signalnum$ signals \citep{GHHS-21}. Under such restrictions, program~\ref{eq:RQD} characterizes the optimal worst-case competitive ratio achievable by monotone-partition, censorship, or $\signalnum$-signal signaling schemes against the Bayesian optimum---without knowledge of the consumer's valuation function $\buyerUtility$ or type distribution $\typeCDF$.

Second, it is important to note that while the platform's disclosure policy (and the signaling scheme it induces) is required to be independent of the buyers' valuation function $\buyerUtility$ and type distribution $\typeCDF$, the expected revenue is computed under the assumption that the corresponding pricing scheme $\pricingscheme$ is optimally chosen given full knowledge of $\buyerUtility$ and $\typeCDF$. 
This consumer-dependent pricing optimization is not only theoretically necessary, without it, the competitive ratio can be unbounded, but also well-motivated in practice.
In real-world platforms, disclosure mechanisms such as badges, tiers, or rating displays are typically designed once and applied uniformly across users, whereas seller frequently adjusts prices in response to their specific consumer segments (e.g., by region, category, or audience). 
This two-stage structure, where the platform commits to a disclosure rule without precise demand-side knowledge, and seller subsequently tailors prices using their own consumer insights, is captured by our robust quality disclosure framework.\footnote{In our model, we conceptually separate the platform (who designs the disclosure policy) from the seller (who sets signal-contingent posted prices). All of our results carry over to an integrated setting in which a single entity jointly chooses a simple disclosure rule and the corresponding downstream pricing scheme.}

\xhdr{Simple disclosure policies}
A particular class of simple disclosure policies is central to this work: {\KQuantilePartitionDisclosure}. These policies restrict attention to monotone-partition signaling schemes that use at most $\signalnum$ signals and impose a consistent structure on how the quality space is partitioned across different prior distributions. The formal definition is as follows.

A {\KQuantilePartitionDisclosure} is parameterized by $(\signalnum+1)$ weakly increasing quantile thresholds $0 = \quant_0 \leq \quant_1 \leq \dots \leq \quant_\signalnum = 1$. Given a quality distribution $\prior \in \Delta([0,1])$, it outputs a monotone-partition signaling scheme that pools all qualities whose prior quantiles fall within the interval $[\quant_{\quantileidx-1}, \quant_{\quantileidx}]$ into a single signal indexed by $\quantileidx$, for each $\quantileidx \in [\signalnum]$. Consequently, the induced signaling scheme $\signalscheme$ has signal space $\signalSpace = [\signalnum]$. The mapping $\{\signalscheme(\cdot \mid \quality)\}$ satisfies, for each $\quantileidx \in [\signalnum]$,\footnote{For a continuous quality distribution $\prior$, the first condition simplifies to $\signalscheme(\quantileidx \mid \quality) > 0$ only if $\prior(\quality) \in [\quant_{\quantileidx-1}, \quant_{\quantileidx}]$.}
\begin{align*}
    &\signalscheme(\quantileidx \mid \quality) > 0 
    \;\;
    \text{only if}
    \;\;
    [\prior(\quality),\, \lim\nolimits_{x \to \quality^+} \prior(x)]
    \cap
    [\quant_{\quantileidx-1},\, \quant_{\quantileidx}] \neq \emptyset, \\
    &\expect[\quality \sim \prior]{\signalscheme(\quantileidx \mid \quality)} = \quant_{\quantileidx} - \quant_{\quantileidx-1}.
\end{align*}
We remark that these conditions uniquely determine the signaling scheme for any given product quality distribution $\prior$. 
See \Cref{fig:quantile} for an illustration. 
(In \Cref{apx:quality partition}, we consider another class of simple disclosure policies based on quality thresholds rather than quantiles.)

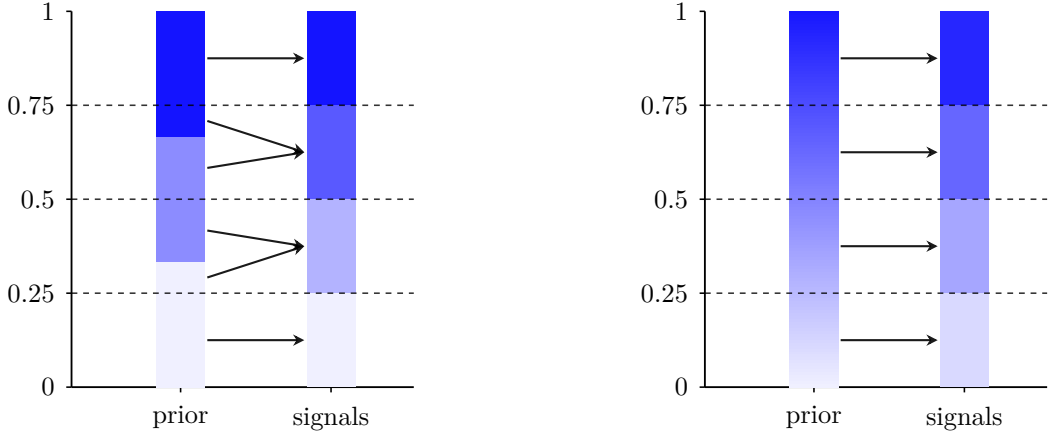
\begin{figure}[H]
\centering

\tikzset{
  dashedCut/.style={dash pattern=on 2.2pt off 2.2pt, line width=0.6pt, black, opacity=0.8},
  axisLine/.style={line width=0.7pt, black},
  tick/.style={line width=0.6pt, black},
  arr/.style={
    -{Stealth[length=1.5mm,width=1.5mm]},
    line width=0.8pt,
    black,
    opacity=0.9
  }
}

\begin{minipage}[t]{0.49\textwidth}
\centering
\begin{tikzpicture}[font=\small]
\def\H{5.0}
\def\barW{0.65}
\def\xAxisStart{0.55}
\def\xPrior{2}
\def\xSig{4}
\def\xAxisEnd{5.10}

\def\qA{0.25}\def\qB{0.50}\def\qC{0.75}
\newcommand{\yq}[1]{#1*\H}

\colorlet{BlueLo}{blue!6}
\colorlet{BlueMid}{blue!45}
\colorlet{BlueHi}{blue!92}
\colorlet{SigD1}{BlueLo}
\colorlet{SigD2}{blue!30}
\colorlet{SigD3}{blue!65}
\colorlet{SigD4}{BlueHi}

\draw[axisLine] (\xAxisStart,0) -- (\xAxisEnd,0);
\draw[axisLine] (\xAxisStart,0) -- (\xAxisStart,\H);

\foreach \qq/\lab in {0/0,\qA/0.25,\qB/0.5,\qC/0.75,1/1}{
  \draw[tick] (\xAxisStart,\yq{\qq}) -- ++(-0.07,0);
  \node[anchor=east] at (\xAxisStart-0.09,\yq{\qq}) {\lab};
}

\draw[tick] (\xPrior,0) -- ++(0,-0.07);
\draw[tick] (\xSig,0) -- ++(0,-0.07);
\node[anchor=north] at (\xPrior,-0.12) {prior};
\node[anchor=north] at (\xSig,-0.12) {signals};

\fill[BlueLo]  (\xPrior-\barW/2, 0)      rectangle (\xPrior+\barW/2, \H/3);
\fill[BlueMid] (\xPrior-\barW/2, \H/3)   rectangle (\xPrior+\barW/2, 2*\H/3);
\fill[BlueHi]  (\xPrior-\barW/2, 2*\H/3) rectangle (\xPrior+\barW/2, \H);

\fill[SigD1] (\xSig-\barW/2, 0)        rectangle (\xSig+\barW/2, \yq{\qA});
\fill[SigD2] (\xSig-\barW/2, \yq{\qA}) rectangle (\xSig+\barW/2, \yq{\qB});
\fill[SigD3] (\xSig-\barW/2, \yq{\qB}) rectangle (\xSig+\barW/2, \yq{\qC});
\fill[SigD4] (\xSig-\barW/2, \yq{\qC}) rectangle (\xSig+\barW/2, \H);

\def\gap{0.03}
\def\xL{\xPrior+\barW/2+\gap}
\def\xR{\xSig-\barW/2-\gap}
\def\mOne{0.125}\def\mTwo{0.375}\def\mThr{0.625}\def\mFou{0.875}

\draw[arr] (\xL, \yq{0.125})  -- (\xR, \yq{\mOne});
\draw[arr] (\xL, \yq{0.2917}) -- (\xR, \yq{\mTwo});
\draw[arr] (\xL, \yq{0.4167}) -- (\xR, \yq{\mTwo});
\draw[arr] (\xL, \yq{0.5833}) -- (\xR, \yq{\mThr});
\draw[arr] (\xL, \yq{0.7083}) -- (\xR, \yq{\mThr});
\draw[arr] (\xL, \yq{0.875})  -- (\xR, \yq{\mFou});

\draw[dashedCut] (\xAxisStart,\yq{\qA}) -- (\xAxisEnd,\yq{\qA});
\draw[dashedCut] (\xAxisStart,\yq{\qB}) -- (\xAxisEnd,\yq{\qB});
\draw[dashedCut] (\xAxisStart,\yq{\qC}) -- (\xAxisEnd,\yq{\qC});
\end{tikzpicture}
\end{minipage}\hfill
\begin{minipage}[t]{0.49\textwidth}
\centering
\begin{tikzpicture}[font=\small]
\def\H{5.0}
\def\barW{0.65}
\def\xAxisStart{0.55}
\def\xPrior{2}
\def\xSig{4}
\def\xAxisEnd{5.10}
\def\qA{0.25}\def\qB{0.50}\def\qC{0.75}
\newcommand{\yq}[1]{#1*\H}

\colorlet{SigC1}{blue!15}
\colorlet{SigC2}{blue!35}
\colorlet{SigC3}{blue!60}
\colorlet{SigC4}{blue!85}

\draw[axisLine] (\xAxisStart,0) -- (\xAxisEnd,0);
\draw[axisLine] (\xAxisStart,0) -- (\xAxisStart,\H);

\foreach \qq/\lab in {0/0,\qA/0.25,\qB/0.5,\qC/0.75,1/1}{
  \draw[tick] (\xAxisStart,\yq{\qq}) -- ++(-0.07,0);
  \node[anchor=east] at (\xAxisStart-0.09,\yq{\qq}) {\lab};
}

\draw[tick] (\xPrior,0) -- ++(0,-0.07);
\draw[tick] (\xSig,0) -- ++(0,-0.07);
\node[anchor=north] at (\xPrior,-0.12) {prior};
\node[anchor=north] at (\xSig,-0.12) {signals};

\shade[bottom color=blue!6, top color=blue!92]
  (\xPrior-\barW/2,0) rectangle (\xPrior+\barW/2,\H);

\fill[SigC1] (\xSig-\barW/2, 0)        rectangle (\xSig+\barW/2, \yq{\qA});
\fill[SigC2] (\xSig-\barW/2, \yq{\qA}) rectangle (\xSig+\barW/2, \yq{\qB});
\fill[SigC3] (\xSig-\barW/2, \yq{\qB}) rectangle (\xSig+\barW/2, \yq{\qC});
\fill[SigC4] (\xSig-\barW/2, \yq{\qC}) rectangle (\xSig+\barW/2, \H);

\def\gap{0.03}
\def\xL{\xPrior+\barW/2+\gap}
\def\xR{\xSig-\barW/2-\gap}
\foreach \m in {0.125,0.375,0.625,0.875}{
  \draw[arr] (\xL, \yq{\m}) -- (\xR, \yq{\m});
}

\draw[dashedCut] (\xAxisStart,\yq{\qA}) -- (\xAxisEnd,\yq{\qA});
\draw[dashedCut] (\xAxisStart,\yq{\qB}) -- (\xAxisEnd,\yq{\qB});
\draw[dashedCut] (\xAxisStart,\yq{\qC}) -- (\xAxisEnd,\yq{\qC});
\end{tikzpicture}
\end{minipage}

\caption{An illustration of {\KQuantilePartitionDisclosure} disclosure. Here $\signalnum = 4$ and $\quant_0 = 0, \quant_1 = 0.25, \quant_2 = 0.5, \quant_3 = 0.75, \quant_4 = 1$. The $y$-axis is cumulative probability mass. 
Left: uniform discrete prior with $\supp(\prior)=\{0,0.5,1\}$. Right: uniform continuous prior with $\supp(\prior)=[0,1]$.
With the same {\KQuantilePartitionDisclosure}, the induced posterior-mean distributions can differ across priors.}
\label{fig:quantile}
\end{figure}

\section{Robust Quantile Partition Disclosure Policy}
\label{sec:opt quantile partition}

In this section, we study the class of {\KQuantilePartitionDisclosure} policies, denoted by $\quantDisClass$. Our main results are as follows: in \Cref{sec:opt quantile partition:optimal policy}, we characterize the optimal robust {\KQuantilePartitionDisclosure} policy and its corresponding optimal robust competitive ratio; in \Cref{sec:opt quantile partition:general policy}, we provide a succinct expression for the robust competitive ratio (i.e., the objective value of program~$\RQDQuant$) for an arbitrary {\KQuantilePartitionDisclosure} policy. The analysis behind our characterization are given in \Cref{sec:rdd}.

\subsection{Optimal Robust \texorpdfstring{{\KQuantilePartitionDisclosure}}{K-Quantile Partition}}
\label{sec:opt quantile partition:optimal policy}

In this section, we characterize the optimal robust {\KQuantilePartitionDisclosure}, i.e., the optimal solution of program~$\RQDQuant$.

\begin{theorem}[Robust optimal quantile partition]
\label{thm:opt quantile partition:optimal policy}
    Fix any integer $\signalnum \in \naturals$. Define the function
    \begin{align*}
        \basefn_\voa(z) \triangleq z + (\voa - 1) \cdot (1+\sqrt{z})^2
    \end{align*}
    for $z \in [0, 1]$, parameterized by $\voa > 1$. Let its $\signalnum$-fold composition starting from $0$ be denoted by
    \begin{align*}
        \recfn_\signalnum(\voa) \triangleq \underbrace{\basefn_\voa \circ \basefn_\voa \circ \dots \circ \basefn_\voa}_{\text{$\signalnum$ times}}(0).
    \end{align*}
    Then the robust quality disclosure problem~$\RQDQuant$ satisfies the following:
    \begin{enumerate}
        \item[(i)] The optimal robust competitive ratio is
        \begin{align*}
            \voa_\signalnum^*,
        \end{align*}
        where $\voa_\signalnum^* > 1$ is the unique solution to the equation $\recfn_\signalnum(\voa) = 1$.

        \item[(ii)] The optimal robust {\KQuantilePartitionDisclosure} $\disclosurefn^*$ is characterized by the quantile threshold profile $\quants^* = (\quant_0^*, \quant_1^*, \dots, \quant_\signalnum^*)$, which is defined recursively by
        \begin{align}
        \label{eq:qstar-def}
\quant_\signalnum^* \triangleq 1,
            \qquad
            \quant_{\quantileidx-1}^*
            \triangleq 1 - \basefn_{\voa_\signalnum^*}(1 - \quant_\quantileidx^*),
            \quad \quantileidx \in [\signalnum].
        \end{align}
    \end{enumerate}
\end{theorem}
In the above characterization, the quantile threshold profile $\quants^* = (\quant_0^*, \quant_1^*, \dots, \quant_\signalnum^*)$ of the optimal robust {\KQuantilePartitionDisclosure} is computed via the backward recursion in~\eqref{eq:qstar-def}. To verify the feasibility of this construction, observe that by definition,
\begin{align*}
    1 - \quant_\signalnum^* \equiv 0,
    \qquad
    1 - \quant_{\quantileidx-1}^*
    \equiv \basefn_{\voa_\signalnum^*}(1 - \quant_\quantileidx^*),
    \quad \quantileidx \in [\signalnum].
\end{align*}
Since $\voa_\signalnum^* > 1$ is the unique solution to $\recfn_\signalnum(\voa) = 1$, it follows that $1 - \quant_0^* = \recfn_\signalnum(\voa_\signalnum^*) = 1$, which implies $\quant_0^* = 0$. Moreover, the sequence $\{\quant_\quantileidx^*\}_{\quantileidx=0}^\signalnum$ is increasing, a property guaranteed by the monotonicity of the function $\basefn_{\voa_\signalnum^*}(\cdot)$.

This recursive construction also reveals the intuition behind the condition that determines the optimal robust competitive ratio $\voa_\signalnum^*$. Specifically, for any candidate $\voa > 1$, the recursion yields a candidate quantile threshold profile; however, this profile may not be valid---most notably, the resulting $\quant_0$ may be strictly positive. (If quantile threshold $\quant_0 > 0$, the corresponding {\KQuantilePartitionDisclosure} is infeasible because it fails to cover the full quantile range $[0,1]$. In contrast, $\quant_0 \leq 0$ still yields a valid signaling scheme, as the effective lower bound is truncated to $0$.) 
Crucially, function $\recfn_\signalnum(\voa)$ is continuous and strictly increasing in $\voa \in (1,\infty)$ (\Cref{lem:Tk-monotone}). Therefore, the optimal robust competitive ratio $\voa_\signalnum^*$ is precisely the \emph{smallest} value of $\voa > 1$ for which the induced quantile threshold $\quant_0$ equals zero (and hence is not positive).

\xhdr{Structural properties of the optimal robust policy}
Invoking \Cref{thm:opt quantile partition:optimal policy}, we report the optimal robust {\KQuantilePartitionDisclosure} policies for $\signalnum = 1, 2, \dots, 5$ in \Cref{tab:robust-quantile-thresholds}. As expected, the optimal robust competitive ratio is strictly decreasing in $\signalnum$, starting from $2$ when $\signalnum = 1$ and decreasing to $1.0904$ for $\signalnum = 5$. This implies that, regardless of the quality distribution $\prior$ or the consumer preferences $(\buyerUtility, \typeCDF)$, 
there exists a {\KQuantilePartitionDisclosure} using just $5$ signals that performs robustly well: achieving at least $91.71\%$ (i.e., $\sfrac{1}{1.0904}$) fraction of the optimal revenue.

\begin{table}[H]
\centering
\renewcommand{\arraystretch}{1.35}
\setlength{\tabcolsep}{10pt}
\begin{tabular}{cccccccc}
\toprule
 & ${\quant_1^*}$ & ${\quant_2^*}$ & ${\quant_3^*}$ & ${\quant_4^*}$ & ${\quant_5^*}$ & $\voa_\signalnum^*$ & $1/\voa_\signalnum^*$ \\
\midrule
${K=1}$ & $1$      & --     & --     & --     & --          & $2$      & $50.00\%$ \\
${K=2}$ & $0.7044$ & $1$    & --     & --     & --          & $1.2956$ & $77.18\%$ \\
${K=3}$ & $0.4946$ & $0.8310$ & $1$  & --     & --         & $1.1690$ & $85.54\%$ \\
${K=4}$ & $0.3772$ & $0.6695$ & $0.8822$ & $1$ & --          & $1.1178$ & $89.46\%$ \\
${K=5}$ & $0.3041$ & $0.5552$ & $0.7567$ & $0.9096$ & $1$   & $1.0904$  & $91.71\%$ \\
\bottomrule
\end{tabular}
\caption{Optimal robust $\signalnum$-quantile threshold profile and their optimal robust competitive ratios.}
\label{tab:robust-quantile-thresholds}
\end{table}

Moreover, for each fixed $\signalnum$, the corresponding optimal quantile threshold profile $(\quant_1^*, \dots, \quant_\signalnum^*)$ exhibits decreasing marginal intervals. For example, when $\signalnum = 2$, the interval lengths are $\quant_1^* - \quant_0^* \approx 0.7044$ and $\quant_2^* - \quant_1^* \approx 0.2956$, with the latter being smaller. This pattern suggests that the optimal robust policy allocates finer discrimination to higher-quality regions and coarser discrimination to lower-quality regions.

We formalize both observations for general $\signalnum \in \naturals$ as follows.

\begin{proposition}[Monotonicity of $\voa_\signalnum^*$]
\label{prop:optimal robust CR monotonicity}
    For any integer $\signalnum \in \naturals$, the optimal robust competitive ratio $\voa_\signalnum^*$ of program~$\RQDQuant$ satisfies
    \begin{align*}
        1 + \frac{1}{4\signalnum} \leq 
        \voa_\signalnum^*
        \leq
        1 + \frac{1}{\signalnum}.
    \end{align*}
    Consequently, $\voa_\signalnum^* = 1 + \Theta(1/\signalnum)$.
\end{proposition}

\begin{proposition}[Decreasing margin]
\label{prop:optimal robust policy:decreasing margin}
    For any integer $\signalnum \in \naturals$, the quantile threshold profile $(\quant_0^*, \quant_1^*, \dots, \quant_\signalnum^*)$ of the optimal robust {\KQuantilePartitionDisclosure} satisfies
    \begin{align*}
        \quant_\quantileidx^* - \quant_{\quantileidx-1}^* > 
        \quant_{\quantileidx+1}^* - \quant_{\quantileidx}^*
    \end{align*}
    for every $\quantileidx \in [\signalnum - 1]$.
\end{proposition}

\xhdr{Implication for the no-information signaling scheme}
For $\signalnum = 1$, the {\KQuantilePartitionDisclosure} outputs the no-information\textsuperscript{\ref{footnote:full info and no info}} signaling scheme $\noinfor$ for every quality distribution $\prior$. Consequently, \Cref{thm:opt quantile partition:optimal policy} implies that, for any linear valuation function $\buyerUtility$, any consumer type distribution $\typeCDF$, and any product quality distribution $\prior$, using the no-information signaling scheme together with the (uniform) consumer-dependent monopoly price (i.e., the optimal posted price for the distribution of the random posterior value $\expect[\quality \sim \prior]{\buyerUtility(\type, \quality)}$) yields a revenue that is always within a factor of 2 of the revenue achievable by the Bayesian optimal signaling and pricing schemes.

\subsection{Robustness of General \texorpdfstring{{\KQuantilePartitionDisclosure}}{K-Quantile Partition}}
\label{sec:opt quantile partition:general policy}

In this section, we provide a succinct expression for the robust competitive ratio (i.e., the objective value of program~$\RQDQuant$) for an arbitrary {\KQuantilePartitionDisclosure} policy.

\begin{theorem}[Robust competitive ratio]
\label{thm:quant partition:general policy}
    For any $\signalnum \in \naturals$ and quantile threshold profile $\quants = (\quant_0, \quant_1, \dots, \quant_\signalnum)$ satisfying $0 = \quant_0 \leq \quant_1 \leq \dots \leq \quant_\signalnum = 1$, the robust competitive ratio of the {\KQuantilePartitionDisclosure} policy $\disclosurefn$ parameterized by $\quants$ is given by
    \begin{align*}
        \sup_{\substack{
        \buyerUtility \in \buyerUtilityClass, \\
        \prior,\typeCDF \in \Delta([0, 1])}}
        \;\;
        \frac{\OPT_{\buyerUtility,\prior,\typeCDF}}{\Rev_{\buyerUtility,\prior,\typeCDF}(\disclosurefn(\prior))}
        =
        \max_{\quantileidx \in [\signalnum]}~
        \left(
        1 + \frac{\quant_\quantileidx - \quant_{\quantileidx - 1}}{\left(1+\sqrt{1-\quant_\quantileidx}\right)^2}
        \right).
    \end{align*}
\end{theorem}

The above characterization reveals a significant simplification in the structure of the robust quality disclosure problem. Although the original minimax program~$\RQDQuant$ involves optimization over all linear valuation functions, product quality distributions, and consumer type distributions, the worst-case competitive ratio for any fixed {\KQuantilePartitionDisclosure} policy depends only on its quantile threshold profile $\quants$ in a remarkably simple way. Specifically, the robust competitive ratio is given by the maximum over $\signalnum$ explicit expressions. Each expression, indexed by $\quantileidx \in [\signalnum]$, corresponds to a candidate worst-case instance (whose precise form will become clearer in our analysis) and whose revenue gap is entirely determined by the upper quantile threshold $\quant_\quantileidx$ and the quantile margin $\quant_\quantileidx - \quant_{\quantileidx-1}$. Moreover, the specific functional form ${(\quant_\quantileidx - \quant_{\quantileidx - 1})}/{(1+\sqrt{1-\quant_\quantileidx})^2}$ again suggests that optimal robustness requires finer discrimination in higher-quality regions and coarser discrimination in lower-quality regions, as already formalized in \Cref{prop:optimal robust policy:decreasing margin}.

We can also use \Cref{thm:quant partition:general policy} to evaluate the worst-case competitive ratio implied by percentile-based badge systems used in practice. For {Upwork}'s badges \citep{Upwork}, corresponding to the quantile thresholds $\quant_0 = 0, \quant_1 = 0.9, \quant_2 = 0.97, \quant_3 = 0.99, \quant_4 = 1$, the theorem yields a competitive ratio of approximately $1.5195$. 
Likewise, for {Airbnb}'s percentile labels \citep{AirbnbTopHomesHighlight}, corresponding to $\quant_0 = 0, \quant_1 = 0.1, \quant_2 = 0.9, \quant_3 = 0.95, \quant_4 = 0.99, \quant_5  = 1$, 
the resulting competitive ratio is approximately $1.4617$.

This characterization enables a clean analytical characterization of optimality. In particular, the backward recursive construction in Theorem~\ref{thm:opt quantile partition:optimal policy} precisely equalizes these $\signalnum$ expressions: the optimal quantile thresholds $\quants^*$ are chosen so that every interval yields the same competitive ratio $\voa_\signalnum^*$. Thus, the optimal policy achieves perfect balance across all potential adversarial scenarios, i.e., a hallmark of minimax optimality.

\xhdr{Implications of \Cref{thm:quant partition:general policy}}
We now discuss several implications of \Cref{thm:quant partition:general policy}.

First, for any {\KQuantilePartitionDisclosure}, each of the $\signalnum$ terms inside the maximum in \Cref{thm:quant partition:general policy} is upper bounded by $2$. Consequently, every {\KQuantilePartitionDisclosure} achieves a robust competitive ratio of at most $2$.

This observation has a broader implication: for any fixed instance of the Bayesian signaling and pricing problem (i.e., given a linear valuation function, a consumer type distribution, and a product quality distribution), every monotone-partition signaling scheme can be induced by some {\KQuantilePartitionDisclosure} with an appropriately chosen quantile threshold profile. Since the robust competitive ratio of that particular {\KQuantilePartitionDisclosure} is at most $2$, it follows that any monotone-partition signaling scheme is also a $2$-approximation to the Bayesian optimal revenue for that instance.

A natural choice of {\KQuantilePartitionDisclosure} is the one based on a uniform quantile threshold profile, i.e., $\quant_\quantileidx = \quantileidx / \signalnum$ for $\quantileidx \in [0:\signalnum]$. For this simple disclosure policy, \Cref{thm:quant partition:general policy} implies that the robust competitive ratio is exactly $1 + 1/\signalnum$, which converges to $1$ as $\signalnum \to \infty$. Comparing this with the convergence rate of the optimal robust {\KQuantilePartitionDisclosure} established in \Cref{prop:optimal robust CR monotonicity}, we see that the uniform quantile policy attains the same asymptotic rate (though with a suboptimal constant factor).

Finally, as $\signalnum \to \infty$, the {\KQuantilePartitionDisclosure} with uniform quantiles converges to the full-information\textsuperscript{\ref{footnote:full info and no info}} signaling scheme $\fullinfor$ for every quality distribution. This, combined with the fact that its robust competitive ratio approaches $1$, implies that, for any fixed instance of the Bayesian signaling and pricing problem, the Bayesian optimal signaling scheme for any fixed instance is indeed the full-information signaling scheme.

We summarize all the above implications in the following two corollaries.

\begin{corollary}
\label{cor:uniform quantile disclosure}
    In the robust quality disclosure problem~$\RQDQuant$, the {\KQuantilePartitionDisclosure} with the uniform quantile threshold profile (i.e., $\quant_\quantileidx = \quantileidx / \signalnum$ for every $\quantileidx \in [\signalnum]$) achieves a robust competitive ratio of
    \begin{align*}
        1 + \frac{1}{\signalnum}.
    \end{align*}
    Moreover, every {\KQuantilePartitionDisclosure} achieves a robust competitive ratio of at most $2$.
\end{corollary}

\begin{corollary}
\label{cor:full infor optimal no infor 2 approx}
    Fix any linear valuation function $\buyerUtility \in \buyerUtilityClass$, consumer type distribution $\typeCDF \in \Delta([0, 1])$, and product quality distribution $\prior \in \Delta([0, 1])$. Then the full-information signaling scheme $\fullinfor$ achieves the optimal revenue, i.e.,
    \begin{align*}
        \Rev(\fullinfor) = \OPT.
    \end{align*}
    Moreover, any monotone-partition signaling scheme $\signalscheme$, including the no-information signaling scheme $\noinfor$, guarantees at least half of the optimal revenue, i.e.,
    \begin{align*}
        2\cdot \Rev(\signalscheme) \geq \OPT.
    \end{align*}
Furthermore, this approximation ratio of $2$ is tight for any monotone-partition signaling scheme using finite number of signals.
\end{corollary}
The optimality of  $\fullinfor$ also follows immediately once we relax the robust quality disclosure program $\RQDQuant$ to the robust disclosure design program $\RDDQuantTitle$ that we introduce in \Cref{def:robust disclosure design} and establish in \Cref{lem:revenue function necessary condition} that the seller's optimal posted-price revenue function is convex.\footnote{We note that this result is also related to the findings of \citet{OP-01,ES-07,RS-10}.}

The tightness of $2$-approximation for monotone-partition signaling scheme $\signalscheme$
for any finite $\signalnum$ is as follows: for any $\varepsilon\in(0, 1)$ there exists an instance with valuation function $\buyerUtility \in \buyerUtilityClass$, product quality distribution and type distribution $\typeCDF_\varepsilon, \prior_\varepsilon \in \Delta([0, 1])$ that depend on $\varepsilon$ such that $\frac{2}{1+\varepsilon}\cdot \Rev(\signalscheme) \le \OPT$. The formal arguments and analysis are provided in \Cref{apx:quality partition}. 
\section{From Robust Quality Disclosure to Robust Disclosure Design}
\label{sec:rdd}

In this section, we first step back from the robust quality disclosure problem and introduce a more general minimax framework, which we term \emph{robust disclosure design}, in \Cref{sec:rdd:robust disclosure design framework}. 
We then show how our robust quality disclosure problem arises as a special case of this broader framework,
and then leverage this reduction to establish the proofs of \Cref{thm:opt quantile partition:optimal policy,thm:quant partition:general policy} in \Cref{sec:rdd:proof of rqd quantile policy}. Finally, we solve the robust disclosure design problem in \Cref{sec:rdd:proof of rdd optimal robust policy,sec:rdd:missing proofs}.

\subsection{Robust Disclosure Design Framework}
\label{sec:rdd:robust disclosure design framework}

The robust quality disclosure problem, as formulated in program~\ref{eq:RQD}, involves optimizing over disclosure policies that output signaling schemes without knowledge of the consumer's preferences, while permitting downstream pricing to fully adapt to those preferences. A key observation is that, under the assumption that the consumer's valuation function is affine in product quality (\Cref{asp:linear valuation function}), the seller's expected revenue depends on the signaling scheme only through the \emph{posterior means} of quality it induces. Consequently, the seller's expected revenue can be expressed as the expectation of a certain convex \emph{indirect utility function} $\cvx: [0,1] \to \reals_+$, where $\cvx(\pooledmean)$ denotes the maximum revenue attainable when all consumers share a common belief that the product's expected quality is $\pooledmean$.
With slight abuse of notation, we use $\disclosurefn(\prior)$ interchangeably to denote both the signaling scheme induced by prior quality distribution $\prior$ and the corresponding distribution over posterior means. Accordingly, we write $\pooledmean \sim \disclosurefn(\prior)$ to indicate that $\pooledmean$ is a posterior mean drawn from the distribution induced by the signaling scheme $\disclosurefn(\prior)$ under prior $\prior$.

The shape of the seller's indirect utility function $\cvx$ is determined by the consumer's valuation function $\buyerUtility$ and the type distribution $\typeCDF$.
Due to the convexity of $\cvx$, the seller's optimal revenue is achieved by using full-information signaling scheme (see \Cref{cor:full infor optimal no infor 2 approx}), and it equals to the expectation of this indirect utility function $\cvx$ under the prior quality distribution $\prior$, i.e., $\expect[\pooledmean \sim \prior]{\cvx(\pooledmean)}$. 
This observation leads to a crucial simplification, once $\cvx$ is fixed, the robust competitive ratio of any disclosure policy $\disclosurefn$ reduces to the ratio 
\begin{align*}
    \frac{\expect[\pooledmean \sim \prior]{\cvx(\pooledmean)}}{\expect[\pooledmean \sim \disclosurefn(\prior)]{\cvx(\pooledmean)}}~.
\end{align*}
This insight motivates the following general abstraction: the \emph{robust disclosure design} framework, which studies, for a given class $\cvxspace$ of indirect utility functions, how well a simple disclosure policy $\disclosurefn$ performs under the worst-case prior distribution and worst-case indirect utility function $\cvx\in\cvxspace$. 
Notably, this framework operates entirely at the level of induced posterior distributions and does not require reference to the underlying economic primitives such as valuation functions or type distributions.

\begin{definition}[Robust disclosure design]
\label{def:robust disclosure design}
Given a class $\disclosurefnClass$ of disclosure policies and a convex functional space $\cvxspace \subseteq \{\cvx : [0,1] \to \reals_+\}$ consisting of convex functions, the \emph{robust disclosure design problem} is defined as
\begin{align}
    \label{eq:rdd}
    \tag{$\RDD[\disclosurefnClass, \cvxspace]$}
    \inf_{\disclosurefn \in \disclosurefnClass}
    \sup_{\substack{\cvx \in \cvxspace \\ \prior \in \Delta([0,1])}}
    ~
    \frac{\expect[\pooledmean \sim \prior]{\cvx(\pooledmean)}}{\expect[\pooledmean \sim \disclosurefn(\prior)]{\cvx(\pooledmean)}}~,
\end{align}
where $\pooledmean \sim \disclosurefn(\prior)$ denotes the posterior mean induced by the signaling scheme $\disclosurefn(\prior)$ under prior $\prior$.

A disclosure policy $\disclosurefn^*$ that attains the infimum in~\ref{eq:rdd} (if one exists) is referred to as an \emph{optimal robust disclosure policy} for the pair $(\disclosurefnClass, \cvxspace)$, and the corresponding objective value is called the \emph{optimal robust competitive ratio} of~\ref{eq:rdd}.
\end{definition}
Motivated by program~\ref{eq:rdd}, we also define the robust competitive ratio of a given disclosure policy $\disclosurefn$ with respect to the function class $\cvxspace$ as
\begin{align*}
    \voa(\disclosurefn \mid \cvxspace) \triangleq 
    \sup_{\substack{\cvx \in \cvxspace \\ \prior \in \Delta([0,1])}}
    ~
    \frac{\expect[\pooledmean \sim \prior]{\cvx(\pooledmean)}}{\expect[\pooledmean \sim \disclosurefn(\prior)]{\cvx(\pooledmean)}}~.
\end{align*}
Compared with the robust quality disclosure problem~\ref{eq:RQD}, whose inner supremum is taken over the product quality distribution $\prior$, consumer valuation function $\buyerUtility$, and consumer type distribution $\typeCDF$, the robust disclosure design framework~\ref{eq:rdd} abstracts away the underlying economic primitives. As a result, its inner supremum is taken only over a prior distribution $\prior$ and an indirect utility function $\cvx$.

Moreover, it is natural to assume that the feasible space $\cvxspace$ of indirect utility functions forms a convex functional space consisting of convex functions. This structural property enables the application of tools from convex analysis (e.g., characterizing extreme points in the space of convex functions) to simplify the worst-case analysis. As we demonstrate in the later parts in \Cref{sec:rdd}, this approach is indeed effective for our setting.

Owing to its generality and analytical tractability, we believe the robust disclosure design framework~\ref{eq:rdd} is of independent interest.

\subsection{Tight Reduction from \texorpdfstring{$\RQDQuantTitle$ to $\RDDQuantTitle$}{RQD to RDD}} 
\label{sec:rdd:proof of rqd quantile policy}
In this section, we formalize the connection between the robust quality disclosure problem~\ref{eq:RQD} and the robust disclosure design framework~\ref{eq:rdd}. We begin by identifying a set of necessary conditions satisfied by the seller's indirect utility function $\cvx$ in our model.

\begin{definition}[Optimal posted-price revenue function]
Given any linear valuation function $\buyerUtility$ and any type distribution $\typeCDF \in \Delta([0,1])$, let $\cvx_{\buyerUtility, \typeCDF}(\cdot) : [0, 1] \to \reals_+$ denote the \emph{optimal posted-price revenue function}, defined as follows: for any induced posterior mean $\pooledmean \in [0,1]$,
\begin{align*}
    \cvx_{\buyerUtility, \typeCDF}(\pooledmean)
    \triangleq 
    \sup_{\price \ge 0}\; \price \cdot 
    \prob[\type \sim \typeCDF]{\buyerUtility(\type, \pooledmean) \ge \price}~.
\end{align*}
\end{definition}
\begin{lemma}[Characterization of $\cvx_{\buyerUtility, \typeCDF}$]
\label{lem:revenue function necessary condition}
    For any linear valuation function $\buyerUtility$ and any type distribution $\typeCDF\in\Delta([0, 1])$, the optimal posted-price revenue function $\cvx_{\buyerUtility, \typeCDF}$ is non-negative, non-decreasing, and convex. Moreover, for any induced posterior mean $\pooledmean \in [0, 1]$,
    \begin{align*}
        \cvx_{\buyerUtility, \typeCDF}(\pooledmean) \geq \pooledmean \cdot \cvx_{\buyerUtility, \typeCDF}(1)~.
    \end{align*}
\end{lemma}

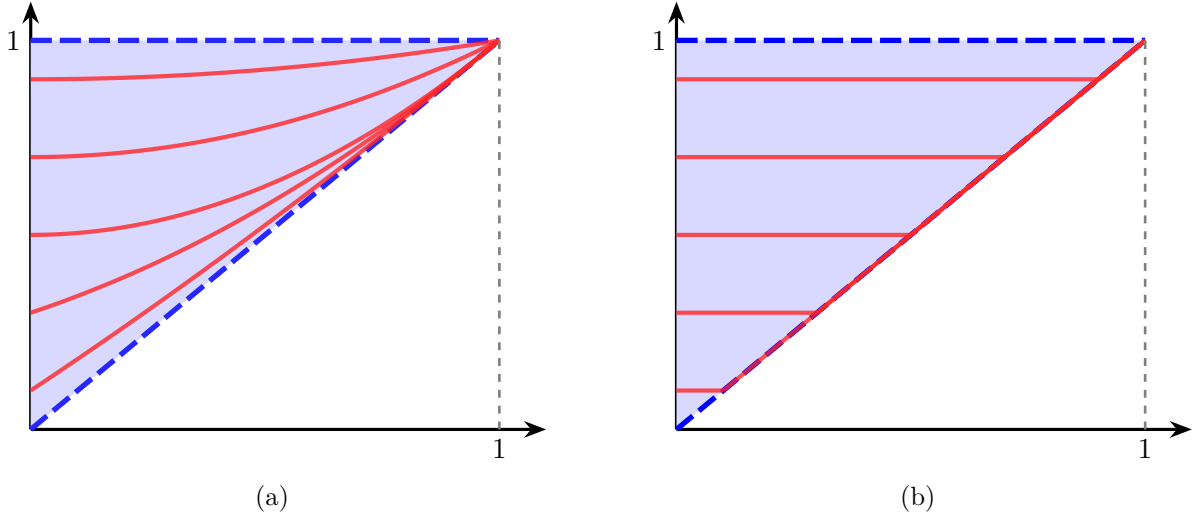
\begin{figure}[t]
\centering

\begin{subfigure}[t]{0.48\textwidth}
\centering
\begin{tikzpicture}
    [x=5cm,y=5cm,>=Stealth,
    ]
\tikzset{reddash/.style={red!85, line width=1.6pt, opacity=0.8}}
\begin{axis}[
xmin=0, xmax=1.1,
    ymin=0, ymax=1.1,
    axis lines=left,
    axis line style={-Stealth, line width=1.2pt},
    ticks=none,
    clip=false,
]

\addplot[name path=lower, draw=none, domain=0:1] {x};
\addplot[name path=upper, draw=none, domain=0:1] {1};
\addplot[draw=none, fill=blue!15] fill between[of=lower and upper];

\addplot[blue!85, dashed, line width=2.0pt, dash pattern=on 8pt off 3pt, domain=0:1] {x};  \addplot[blue!85, dashed, line width=2.0pt, dash pattern=on 8pt off 3pt, domain=0:1] {1};  

\addplot[reddash, line width=1.6pt, domain=0:1, samples=200] {0.1 + 0.8*x + 0.1*x^2};
\addplot[reddash, line width=1.6pt, domain=0:1, samples=200] {0.3 + 0.4*x + 0.3*x^2};
\addplot[reddash, line width=1.6pt, domain=0:1, samples=200] {0.5 + 0.5*x^2};
\addplot[reddash, line width=1.6pt, domain=0:1, samples=200] {0.7 + 0.3*x^2};
\addplot[reddash, line width=1.6pt, domain=0:1, samples=200] {0.9 + 0.1*x^2};

\node[anchor=north] at (axis cs:1,0) {$1$};
\node[anchor=east]  at (axis cs:0,1) {$1$};
\draw[gray,dashed,line width=1.0pt] (1,0) -- (1,1);
\end{axis}
\end{tikzpicture}
\caption{}
\label{fig:left tight bound}
\end{subfigure}
\hfill
\begin{subfigure}[t]{0.48\textwidth}
\centering
\begin{tikzpicture}[x=5cm,y=5cm,>=Stealth]
\tikzset{reddash/.style={red!85, line width=1.6pt, opacity=0.8}}
\begin{axis}[
xmin=0, xmax=1.1,
    ymin=0, ymax=1.1,
    axis lines=left,
    axis line style={-Stealth, line width=1.2pt},
    ticks=none,
    clip=false,
]

\addplot[name path=lower2, draw=none, domain=0:1] {x};
\addplot[name path=upper2, draw=none, domain=0:1] {1};
\addplot[draw=none, fill=blue!15] fill between[of=lower2 and upper2];

\addplot[blue, dashed, line width=2.0pt, dash pattern=on 8pt off 3pt, domain=0:1] {x};  \addplot[blue, dashed, line width=2.0pt, dash pattern=on 8pt off 3pt, domain=0:1] {1};

\addplot[reddash, line width=1.6pt, domain=0:0.1] {0.1};
\addplot[reddash, line width=1.6pt, domain=0.1:1] {x};

\addplot[reddash, line width=1.6pt, domain=0:0.3] {0.3};
\addplot[reddash, line width=1.6pt, domain=0.3:1] {x};

\addplot[reddash, line width=1.6pt, domain=0:0.5] {0.5};
\addplot[reddash, line width=1.6pt, domain=0.5:1] {x};

\addplot[reddash, line width=1.6pt, domain=0:0.7] {0.7};
\addplot[reddash, line width=1.6pt, domain=0.7:1] {x};

\addplot[reddash, line width=1.6pt, domain=0:0.9] {0.9};
\addplot[reddash, line width=1.6pt, domain=0.9:1] {x};

\node[anchor=north] at (axis cs:1,0) {$1$};
\node[anchor=east]  at (axis cs:0,1) {$1$};
\draw[gray,dashed,line width=1.0pt] (1,0) -- (1,1);
\end{axis}
\end{tikzpicture}
\caption{}
\label{fig:right tight bound}
\end{subfigure}
\caption{The blue dashed lines are upper boundary (i.e., $\cvx(\pooledmean) \equiv 1$) and lower boundary (i.e., $\cvx(\pooledmean) = \pooledmean$) of the functional space defined in \Cref{thm:rdd:optimal robust policy} (here we normalize $\cvx(1) = 1$ for the presentation simplicity).
The red solid lines in \Cref{fig:left tight bound} are some feasible convex function $\cvx \in \cvxspace$ in this functional space $\cvxspace$, while the red solid lines in \Cref{fig:right tight bound} are some extremal convex functions (i.e., $\cvx(\pooledmean) = \max\{\pooledmean, \kink\}$ that we establish in \Cref{lem:stoploss-representation} for some $\kink\in[0, 1]$) in this functional space $\cvxspace$.}
\label{fig:tight bound}

\end{figure}

\begin{proof}
Non-negativity follows immediately from the definition, since prices and probabilities are non-negative.

\xhdr{Monotonicity}
Fix $\pooledmean_1, \pooledmean_2 \in [0,1]$ with $\pooledmean_1 \leq \pooledmean_2$. By \Cref{asp:linear valuation function}, the valuation function $\buyerUtility$ is affine and increasing in quality $\quality$ for every type $\type$. It follows that $\buyerUtility(\type, \pooledmean_1) \leq \buyerUtility(\type, \pooledmean_2)$ for every type~$\type$.
For any price $\price \geq 0$, define the demand events
\begin{align*}
    A_i(\price) \triangleq \left\{ \type : \buyerUtility(\type, \pooledmean_i) \geq \price \right\}, \quad i \in \{1,2\}~.
\end{align*}
The pointwise inequality implies $A_1(\price) \subseteq A_2(\price)$, and thus
\begin{align*}
    \price \cdot \prob[\type \sim \typeCDF]{A_1(\price)} \leq \price \cdot \prob[\type \sim \typeCDF]{A_2(\price)}~.
\end{align*}
Taking the supremum over price $\price \geq 0$ yields
$\cvx_{\buyerUtility, \typeCDF}(\pooledmean_1) \leq \cvx_{\buyerUtility, \typeCDF}(\pooledmean_2)$,
establishing that $\cvx_{\buyerUtility, \typeCDF}$ is non-decreasing.

\xhdr{Convexity}
By \Cref{asp:linear valuation function}, the valuation function $\buyerUtility(\type, \pooledmean)$ is affine in quality (i.e., in $\pooledmean$) for each fixed type $\type$, and non-decreasing in type $\type$ for each fixed quality. This monotonicity ensures that, for any posterior mean $\pooledmean$, the optimal posted price corresponds to a cutoff type: there exists some $\type \in [0,1]$ such that all types $\type\primed \geq \type$ purchase the product, while all types $\type\primed < \type$ do not. Consequently, the optimal posted-price revenue function can be written as
\begin{align*}
    \cvx_{\buyerUtility, \typeCDF}(\pooledmean)
    = \sup_{\type \in [0,1]} \left(1 - \typeCDF(\type)\right) \cdot \buyerUtility(\type, \pooledmean)~.
\end{align*}
For each fixed $\type$, define the auxiliary function
\begin{align*}
    f_\type(\pooledmean) \triangleq \left(1 - \typeCDF(\type)\right) \cdot \buyerUtility(\type, \pooledmean)~.
\end{align*}
Since $\buyerUtility(\type, \cdot)$ is affine in $\pooledmean$ and $1 - \typeCDF(\type)$ is a non-negative constant with respect to $\pooledmean$, the function $f_\type(\cdot)$ is affine (and therefore convex) on $[0,1]$. Finally, note that function $\cvx_{\buyerUtility, \typeCDF}$ is the pointwise supremum of the family $\{f_\type : \type \in [0,1]\}$. Since the supremum of any collection of convex functions is convex, it guarantees that function $\cvx_{\buyerUtility, \typeCDF}$ is convex on $[0,1]$.

\xhdr{Linear lower bound} 
Building on the previous analysis, recall that for each fixed $\type$, the auxiliary function $f_\type(\pooledmean)$ is affine in $\pooledmean$. Therefore,
\begin{align*}
    f_\type(\pooledmean) = (1 - \pooledmean) \cdot f_\type(0) + \pooledmean \cdot f_\type(1) \geq \pooledmean \cdot f_\type(1),
\end{align*}
where the inequality follows from the non-negativity of the valuation function, which implies $f_\type(0) \geq 0$.
Taking the supremum over $\type \in [0,1]$ on both sides yields
\begin{align*}
    \cvx_{\buyerUtility, \typeCDF}(\pooledmean) \geq \pooledmean \cdot \cvx_{\buyerUtility, \typeCDF}(1),
\end{align*}
as desired.
\end{proof}

Given the necessary conditions established in \Cref{lem:revenue function necessary condition}, it is natural to consider the robust disclosure design program~$\RDDQuant$ with a functional space $\cvxspace$ consisting of indirect utility functions that satisfy these conditions (see \Cref{fig:left tight bound} for an illustration). The following theorem characterizes both the optimal robust competitive ratio and the corresponding optimal disclosure policy for this program. Its formal proof is given in \Cref{sec:rdd:missing proofs}.

\begin{theorem}
\label{thm:rdd:optimal robust policy}
    Fix any integer $\signalnum \in \naturals$. Let $\quantDisClass$ denote the class of all {\KQuantilePartitionDisclosure} policies, and let $\cvxspace$ be the space of all non-negative, non-decreasing, convex functions $\cvx : [0,1] \to \reals_+$ satisfying $\cvx(\pooledmean) \geq \pooledmean \cdot \cvx(1)$ for all $\pooledmean \in [0,1]$. Then:
    \begin{enumerate}
        \item[(i)] The optimal robust competitive ratio of program~$\RDDQuant$ is $\voa_\signalnum^*$ (as defined in \Cref{thm:opt quantile partition:optimal policy}), and the optimal robust {\KQuantilePartitionDisclosure} policy is the one specified by the quantile threshold profile in Eqn.~\eqref{eq:qstar-def} of \Cref{thm:opt quantile partition:optimal policy}.
        
        \item[(ii)] For every {\KQuantilePartitionDisclosure} policy $\disclosurefn$ with quantile thresholds $\quants = (\quant_0, \dots, \quant_\signalnum)$, its robust competitive ratio in program~$\RDDQuant$ is given by
        \begin{align*}
            \voa(\disclosurefn \mid \cvxspace)
            =
            \max_{\quantileidx \in [\signalnum]}
            \left(
            1 + \frac{\quant_\quantileidx - \quant_{\quantileidx - 1}}{\left(1+\sqrt{1-\quant_\quantileidx}\right)^2}
            \right).
        \end{align*}
    \end{enumerate}
\end{theorem}
We begin by observing that \Cref{thm:rdd:optimal robust policy} is the natural counterpart of \Cref{thm:opt quantile partition:optimal policy,thm:quant partition:general policy}, but formulated within the robust disclosure design framework~$\RDDQuant$ rather than the original robust quality disclosure problem~$\RQDQuant$. By the necessary conditions established in \Cref{lem:revenue function necessary condition}, the inner supremum in $\RDDQuant$ is a relaxation of that in $\RQDQuant$, and thus the robust competitive ratio under $\RDDQuant$ provides an upper bound for $\RQDQuant$.

Crucially, although these conditions are only necessary (not sufficient), they are precisely strong enough to ensure that this relaxation is \emph{tight}: the upper bound obtained via $\RDDQuant$ matches the true robust competitive ratio of $\RQDQuant$. In other words, \Cref{lem:revenue function necessary condition} identifies the ``right'' relaxation over the space of indirect utility functions—one that preserves the robust competitive ratio exactly.

This tightness is not automatic. In general, when applying the robust disclosure design framework to a specific economic setting (e.g., downstream pricing in our model), the choice of structural conditions on the indirect utility function space is critical. As we illustrate in  \Cref{apx:failure sandwich}, there exists another natural necessary condition satisfied by the optimal posted-price revenue function $\cvx_{\buyerUtility, \typeCDF}$ that, if used in place of those in \Cref{lem:revenue function necessary condition}, leads to a strictly looser upper bound, i.e., it fails to recover the tight characterization in \Cref{thm:rdd:optimal robust policy} (see \Cref{fig:sandwich} for the illustration of this natural necessary condition and \Cref{prop:sandwich} for the formal statements).

Given this tight reduction, the proofs of \Cref{thm:opt quantile partition:optimal policy,thm:quant partition:general policy} follow a clean two-step structure: we invoke \Cref{thm:rdd:optimal robust policy} to obtain the matching upper bound, and complement it with a matching lower bound, which is established in \Cref{lem:quantile partition:matching lower bound} below.

\begin{lemma}[Matching lower bound]
\label{lem:quantile partition:matching lower bound}
Fix any $\signalnum \in \naturals$ and a quantile threshold profile $\quants = (\quant_0, \quant_1, \dots, \quant_\signalnum)$ satisfying $0 = \quant_0 \leq \quant_1 \leq \dots \leq \quant_\signalnum = 1$. Let $\disclosurefn$ be the {\KQuantilePartitionDisclosure} parameterized by $\quants$, and let
\begin{align*} 
\quantileidx\primed \in \argmax_{\quantileidx \in [\signalnum]} \left( 1 + \frac{\quant_\quantileidx - \quant_{\quantileidx - 1}}{\left(1 + \sqrt{1 - \quant_\quantileidx}\right)^2} \right).
\end{align*}
Then the following hold:
\begin{itemize}
    \item If $\quantileidx\primed \in [\signalnum - 1]$, there exist a linear valuation function $\buyerUtility$, a consumer type distribution $\typeCDF \in \Delta([0,1])$, and a product quality distribution $\prior \in \Delta([0,1])$ such that
    \begin{align*}
        \frac{\OPT_{\buyerUtility,\prior,\typeCDF}}{\Rev_{\buyerUtility,\prior,\typeCDF}(\disclosurefn(\prior))}
        =
        1 + \frac{\quant_{\quantileidx\primed} - \quant_{\quantileidx\primed - 1}}{\left(1 + \sqrt{1 - \quant_{\quantileidx\primed}}\right)^2}~.
    \end{align*}

    \item If $\quantileidx\primed = \signalnum$, then for every $\varepsilon \in (0, \quant_\signalnum - \quant_{\signalnum - 1})$, there exist a linear valuation function $\buyerUtility$, $\typeCDF \in \Delta([0,1])$, and $\prior \in \Delta([0,1])$ such that
    \begin{align*}
        \frac{\OPT_{\buyerUtility,\prior,\typeCDF}}{\Rev_{\buyerUtility,\prior,\typeCDF}(\disclosurefn(\prior))}
        > 1 + (\quant_\signalnum - \quant_{\signalnum - 1}) - \varepsilon~.
    \end{align*}
\end{itemize}
Here, $\OPT_{\buyerUtility,\prior,\typeCDF}$ denotes the seller’s expected revenue under the Bayesian optimal signaling and pricing scheme, and $\Rev_{\buyerUtility,\prior,\typeCDF}(\disclosurefn(\prior))$ denotes the expected revenue under the signaling scheme $\disclosurefn(\prior)$ induced by $\disclosurefn$, given consumer valuation $\buyerUtility$, type distribution $\typeCDF$, and quality distribution $\prior$.
\end{lemma}

The proof of \Cref{lem:quantile partition:matching lower bound} relies on the following unified construction of the consumer's valuation function and type distribution.

\begin{example}[Consumer preference $(\buyerUtility, \typeCDF)$ with $\cvx_{\buyerUtility, \typeCDF}(\pooledmean) = \max\{\kink, \pooledmean\}$]
\label{lem:stoploss-realize}
Fix any $\kink \in (0,1)$. Let consumer type distribution $\typeCDF$ be a Bernoulli distribution supported on $\{0,1\}$ with mean $\kink$, and define the linear valuation function
\begin{align*}
    \buyerUtility(\type, \quality) \triangleq \type(1 - \quality) + \quality, \quad \text{for all } (\type, \quality) \in [0,1]^2.
\end{align*}
For this pair $(\buyerUtility, \typeCDF)$, the optimal posted-price revenue function is given by
\begin{align*}
\cvx_{\buyerUtility, \typeCDF}(\pooledmean) = \max\{\kink, \pooledmean\}, \quad \forall\, \pooledmean \in [0,1].
\end{align*}
\end{example}
\begin{proof}[Proof of \Cref{lem:quantile partition:matching lower bound}]
We prove the two cases separately.

\xhdr{Case 1: index $\quantileidx\primed \in [\signalnum - 1]$} In this case,
define auxiliary notations 
\begin{align*}
\auxratio_{\quantileidx\primed}
\triangleq
\frac{\sqrt{1 - \quant_{\quantileidx\primed}}}{1 + \sqrt{1 - \quant_{\quantileidx\primed}}}
\in (0,1),
\qquad
\priormean_{\quantileidx\primed}
\triangleq
(1 - \quant_{\quantileidx\primed}) + (\quant_{\quantileidx\primed} - \quant_{\quantileidx\primed - 1}) \cdot \auxratio_{\quantileidx\primed}
\in (0,1).
\end{align*}
Consider a Bernoulli quality distribution $\prior$ supported on $\{0,1\}$ with mean $\priormean_{\quantileidx\primed}$.

Let $u \sim \mathrm{Unif}(0,1)$ and define $\pooledmean = \prior^{-1}(u)$. Since quality distribution $\prior$ places mass $1 - \priormean_{\quantileidx\primed}$ at $0$ and mass $\priormean_{\quantileidx\primed}$ at $1$, there exists a threshold
\begin{align*}
u_0 \triangleq 1 - \priormean_{\quantileidx\primed}
\end{align*}
such that $\pooledmean = 0$ for $u \in (0, u_0]$ and $\pooledmean = 1$ for $u \in (u_0, 1]$.
We claim that $u_0 \in (\quant_{\quantileidx\primed - 1}, \quant_{\quantileidx\primed}]$. To see this, note that 
\begin{align*}
u_0 
= 1 - \priormean_{\quantileidx\primed} 
= 1 - \bigl[(1 - \quant_{\quantileidx\primed}) + (\quant_{\quantileidx\primed} - \quant_{\quantileidx\primed - 1}) \auxratio_{\quantileidx\primed}\bigr] 
= \quant_{\quantileidx\primed} - (\quant_{\quantileidx\primed} - \quant_{\quantileidx\primed - 1}) \auxratio_{\quantileidx\primed}.
\end{align*}
Since $\auxratio_{\quantileidx\primed} \in (0,1)$, it follows that
$\quant_{\quantileidx\primed - 1}
< u_0
< \quant_{\quantileidx\primed}$,
as required.

Let auxiliary notation $m_\quantileidx$ denote the posterior mean given signal $\quantileidx\in[\signalnum]$ under signaling scheme $\disclosurefn(\prior)$ induced from prior $\prior$. Since $u_0 \in (\quant_{\quantileidx\primed - 1}, \quant_{\quantileidx\primed})$, we have:
(i) For $\quantileidx < \quantileidx\primed$, the entire quantile interval $(\quant_{i-1}, \quant_i]$ lies in $(0, u_0]$, so $\pooledmean = 0$ almost surely and $m_\quantileidx = 0$.
(ii) For $\quantileidx > \quantileidx\primed$, the entire quantile interval lies in $(u_0, 1]$, so $\pooledmean = 1$ almost surely and $m_\quantileidx = 1$.
(iii) In quantile interval indexed by $\quantileidx\primed$, the subinterval $(u_0, \quant_{\quantileidx\primed}]$ has length $(\quant_{\quantileidx\primed} - \quant_{\quantileidx\primed - 1}) \auxratio_{\quantileidx\primed}$, which is a fraction $\auxratio_{\quantileidx\primed}$ of the total bin mass. Hence, $m_{\quantileidx\primed} = \auxratio_{\quantileidx\primed}$.
Thus, the induced posterior mean distribution is
\begin{align*}
\disclosurefn(\prior) = 
\quant_{\quantileidx\primed - 1} \cdot \delta_0
+ (\quant_{\quantileidx\primed} - \quant_{\quantileidx\primed - 1}) \cdot \delta_{\auxratio_{\quantileidx\primed}}
+ (1 - \quant_{\quantileidx\primed}) \cdot \delta_1.
\end{align*}
Now choose $(\buyerUtility, \typeCDF)$ as in \Cref{lem:stoploss-realize} with $\kink = \auxratio_{\quantileidx\primed}$. Then
\begin{align*}
\cvx_{\buyerUtility, \typeCDF}(\pooledmean) = \max\{\auxratio_{\quantileidx\primed}, \pooledmean\}, \quad \forall\, \pooledmean \in [0,1].
\end{align*}
The seller's revenue under $\disclosurefn(\prior)$ is
\begin{align*}
\Rev_{\buyerUtility,\prior,\typeCDF}(\disclosurefn(\prior))
= \expect[\pooledmean \sim \disclosurefn(\prior)]{\cvx(\pooledmean)} 
& = \quant_{\quantileidx\primed - 1} \cdot \auxratio_{\quantileidx\primed}
+ (\quant_{\quantileidx\primed} - \quant_{\quantileidx\primed - 1}) \cdot \auxratio_{\quantileidx\primed}
+ (1 - \quant_{\quantileidx\primed}) \cdot 1 \\
& = \quant_{\quantileidx\primed} \auxratio_{\quantileidx\primed} + (1 - \quant_{\quantileidx\primed}).
\end{align*}
The Bayesian optimal revenue (under full-information signaling scheme) is
\begin{align*}
\OPT_{\buyerUtility,\prior,\typeCDF}
&= \expect[\pooledmean \sim \prior]{\cvx(\pooledmean)} 
= (1 - \priormean_{\quantileidx\primed}) \cdot \auxratio_{\quantileidx\primed}
+ \priormean_{\quantileidx\primed} \cdot 1 
= \auxratio_{\quantileidx\primed} + \priormean_{\quantileidx\primed} (1 - \auxratio_{\quantileidx\primed}) 
\\
&= \Rev_{\buyerUtility,\prior,\typeCDF}(\disclosurefn(\prior)) + (\quant_{\quantileidx\primed} - \quant_{\quantileidx\primed - 1}) \auxratio_{\quantileidx\primed} (1 - \auxratio_{\quantileidx\primed}).
\end{align*}
Therefore,
\begin{align*}
\frac{\OPT_{\buyerUtility,\prior,\typeCDF}}{\Rev_{\buyerUtility,\prior,\typeCDF}(\disclosurefn(\prior))}
= 1 + \frac{(\quant_{\quantileidx\primed} - \quant_{\quantileidx\primed - 1}) \auxratio_{\quantileidx\primed} (1 - \auxratio_{\quantileidx\primed})}
{\quant_{\quantileidx\primed} \auxratio_{\quantileidx\primed} + (1 - \quant_{\quantileidx\primed})}.
\end{align*}
Substituting $\auxratio_{\quantileidx\primed} = \frac{\sqrt{1 - \quant_{\quantileidx\primed}}}{1 + \sqrt{1 - \quant_{\quantileidx\primed}}}$ and simplifying, one verifies that the fraction equals
\begin{align*}
\frac{\quant_{\quantileidx\primed} - \quant_{\quantileidx\primed - 1}}{(1 + \sqrt{1 - \quant_{\quantileidx\primed}})^2}.
\end{align*}
Hence,
\begin{align*}
\frac{\OPT_{\buyerUtility,\prior,\typeCDF}}{\Rev_{\buyerUtility,\prior,\typeCDF}(\disclosurefn(\prior))}
= 1 + \frac{\quant_{\quantileidx\primed} - \quant_{\quantileidx\primed - 1}}{(1 + \sqrt{1 - \quant_{\quantileidx\primed}})^2},
\end{align*}
as claimed.

\xhdr{Case 2: index $\quantileidx\primed = \signalnum$}
In this case, the maximizing term corresponds to the last quantile interval $[\quant_{\signalnum - 1}, 1]$, with length $\quant_\signalnum - \quant_{\signalnum - 1} = 1 - \quant_{\signalnum - 1}$.
Fix any $\varepsilon \in \bigl(0, 1 - \quant_{\signalnum - 1}\bigr)$. Consider a three-point quality distribution $\prior_\varepsilon$ supported on $\{0, t, 1\}$ for some $t \in (0,1)$ (to be chosen later), with
\begin{align*}
\prob[\prior_\varepsilon]{0} = 1 - \varepsilon, \quad
\prob[\prior_\varepsilon]{t} = \frac{\varepsilon}{2}, \quad
\prob[\prior_\varepsilon]{1} = \frac{\varepsilon}{2}.
\end{align*}
Since $1 - \varepsilon \geq \quant_{\signalnum - 1}$, the first $\signalnum - 1$ quantile intervals are fully contained in the atom at $0$, so their posterior means are $0$. The last quantile interval, of mass $1 - \quant_{\signalnum - 1}$, contains:
(i) mass $(1 - \varepsilon) - \quant_{\signalnum - 1} = (1 - \quant_{\signalnum - 1}) - \varepsilon$ at $0$,
(ii) mass $\varepsilon/2$ at $t$,
(iii) mass $\varepsilon/2$ at $1$.
Thus, the posterior mean in the last quantile interval is
\begin{align*}
m_\signalnum
= \frac{0 \cdot ((1 - \quant_{\signalnum - 1}) - \varepsilon) + t \cdot (\varepsilon/2) + 1 \cdot (\varepsilon/2)}{1 - \quant_{\signalnum - 1}}
= \frac{\varepsilon (1 + t)}{2 (1 - \quant_{\signalnum - 1})}.
\end{align*}
Choose $t \in (0,1)$ such that $m_\signalnum < t$. Now apply \Cref{lem:stoploss-realize} with $\kink = m_\signalnum$, yielding $\cvx(\pooledmean) = \max\{m_\signalnum, \pooledmean\}$ for all $\pooledmean\in[0, 1]$. By construction,
the revenue under $\disclosurefn(\prior_\varepsilon)$ is
\begin{align*}
\Rev_{\buyerUtility,\prior_\varepsilon,\typeCDF}(\disclosurefn(\prior_\varepsilon)) = \quant_\signalnum\cdot m_\signalnum + (1 - \quant_{\signalnum - 1}) \cdot m_\signalnum = m_\signalnum.
\end{align*}
The Bayesian optimal revenue (under full-information signaling scheme) is
\begin{align*}
\OPT_{\buyerUtility,\prior_\varepsilon,\typeCDF}
&= (1 - \varepsilon) \cdot m_\signalnum + \frac{\varepsilon}{2} \cdot t + \frac{\varepsilon}{2} \cdot 1 
= m_\signalnum + \frac{\varepsilon}{2}(1 + t - 2 m_\signalnum)
\\
&\overset{(a)}{=}
m_\signalnum + m_\signalnum (1 - \quant_{\signalnum - 1}) - \varepsilon m_\signalnum
= m_\signalnum (1 + (1 - \quant_{\signalnum - 1}) - \varepsilon)
\end{align*}
where equality (a) holds since $\frac{\varepsilon}{2}(1 + t) = m_\signalnum (1 - \quant_{\signalnum - 1})$ (implied by the definition of $m_\signalnum$).
Therefore,
\begin{align*}
\frac{\OPT_{\buyerUtility,\prior_\varepsilon,\typeCDF}}{\Rev_{\buyerUtility,\prior_\varepsilon,\typeCDF}(\disclosurefn(\prior))} = 1 + (1 - \quant_{\signalnum - 1}) - \varepsilon = 1 + (\quant_\signalnum - \quant_{\signalnum - 1}) - \varepsilon,
\end{align*}
as required.
Combining the two cases finishes the proof of \Cref{lem:quantile partition:matching lower bound}.
\end{proof}

We conclude this section by proving \Cref{thm:opt quantile partition:optimal policy,thm:quant partition:general policy}.

\begin{proof}[Proof of \Cref{thm:opt quantile partition:optimal policy,thm:quant partition:general policy}]
By \Cref{lem:revenue function necessary condition}, the indirect utility function $\cvx_{\buyerUtility, \typeCDF}$ induced by any linear valuation function $\buyerUtility$ and consumer type distribution $\typeCDF$ belongs to the function class $\cvxspace$ defined in \Cref{thm:rdd:optimal robust policy}. Consequently, the robust quality disclosure problem~$\RQDQuant$ is a restriction of the robust disclosure design problem~$\RDDQuant$, and thus
\begin{align*}
\sup_{\substack{\buyerUtility \in \buyerUtilityClass \\ \prior, \typeCDF \in \Delta([0,1])}}
\frac{\OPT_{\buyerUtility,\prior,\typeCDF}}{\Rev_{\buyerUtility,\prior,\typeCDF}(\disclosurefn(\prior))}
\leq
\voa(\disclosurefn \mid \cvxspace)
\end{align*}
for every {\KQuantilePartitionDisclosure} $\disclosurefn$.

Applying \Cref{thm:rdd:optimal robust policy}, we obtain an upper bound on the robust competitive ratio of any such policy, and in particular, the optimal robust competitive ratio of $\RQDQuant$ is at most $\voa_\signalnum^*$, achieved by the quantile threshold profile $\quants^*$ defined in~Eqn.\eqref{eq:qstar-def}.

To show this bound is tight, we invoke \Cref{lem:quantile partition:matching lower bound}, which constructs, for the policy $\disclosurefn$ parameterized by $\quants^*$, explicit primitives $(\buyerUtility, \typeCDF, \prior)$ such that the competitive ratio of $\disclosurefn$ equals $\voa_\signalnum^*$. This establishes that the upper bound is attained, and hence the policy $\disclosurefn^*$ with thresholds $\quants^*$ is optimal for $\RQDQuant$ with robust competitive ratio exactly $\voa_\signalnum^*$.

The same argument, i.e., upper bound from \Cref{thm:rdd:optimal robust policy} and matching lower bound from \Cref{lem:quantile partition:matching lower bound} applies to any {\KQuantilePartitionDisclosure} $\disclosurefn$, yielding the exact expression for its robust competitive ratio stated in \Cref{thm:quant partition:general policy}.
\end{proof}

\subsection{Solving \texorpdfstring{$\RDDQuantTitle$}{RDD}: Key Steps and Technical Lemmas}
\label{sec:rdd:proof of rdd optimal robust policy}

We now outline the main steps in solving the robust disclosure design problem $\RDDQuant$ defined in \Cref{thm:rdd:optimal robust policy}. The analysis leverages convex-analytic tools and extremal constructions to derive a closed-form expression for the robust competitive ratio of any {\KQuantilePartitionDisclosure}.

\xhdr{Step 1: Normalization via $\cvx(1) = 1$ (\Cref{lem:normalization})}
The objective in $\RDDQuant$ is invariant under positive scaling of the indirect utility function $\cvx$. Hence, without loss of generality, we may restrict attention to the normalized class $\cvxspace_1\subseteq \cvxspace$ such that $\cvxspace_1 \triangleq \{\cvx\in\cvxspace:\cvx(1) = 1\}$.

\begin{restatable}[Scale normalization]{lemma}{lemNormalization}
\label{lem:normalization}
Let $\cvxspace$ denote the space of all non-negative, non-decreasing, convex functions $\cvx : [0,1] \to \reals_+$ satisfying $\cvx(\pooledmean) \geq \pooledmean \cdot \cvx(1)$ for all $\pooledmean \in [0,1]$, and define the normalized subclass
$\cvxspace_1 \triangleq \{\cvx \in \cvxspace : \cvx(1) = 1\}$.
Then, for any distributions $\prior, \distOfMean \in \Delta([0,1])$,
\begin{align*}
\sup_{\cvx \in \cvxspace} \frac{\expect[\pooledmean \sim \prior]{\cvx(\pooledmean)}}{\expect[\pooledmean \sim \distOfMean]{\cvx(\pooledmean)}}
=
\sup_{\cvx \in \cvxspace_1} \frac{\expect[\pooledmean \sim \prior]{\cvx(\pooledmean)}}{\expect[\pooledmean \sim \distOfMean]{\cvx(\pooledmean)}}.
\end{align*}
\end{restatable}

\xhdr{Step 2: Reduction to extremal functions (\Cref{lem:stoploss-representation,lem:reduce-to-extremal})}
The set $\cvxspace_1$ is convex and compact under the topology of pointwise convergence. Its extreme points admit a simple and explicit characterization (see \Cref{fig:right tight bound} for an illustration).

\begin{restatable}[Extremal representation]{lemma}{lemExtRepresentation}
\label{lem:stoploss-representation}
Let $\cvxspace_1$ denote the functional space defined in \Cref{lem:normalization}.
For every function $\cvx \in \cvxspace_1$, there exists a unique probability measure $\nu_\cvx \in \Delta([0,1])$ such that
\begin{align*}
\cvx(\pooledmean) = \int_0^1 \max\{\kink, \pooledmean\}\,\dd\nu_\cvx(\kink), \quad \forall\, \pooledmean \in [0,1].
\end{align*}
Consequently, the set of extreme points of $\cvxspace_1$ is given by $\ext(\cvxspace_1) = \left\{ \cvx_\kink(\cdot) \triangleq \max\{\kink, \cdot\} : \kink \in [0,1] \right\}$.
\end{restatable}

Although the objective in program~$\RDDQuant$ is not linear in $\cvx$, the supremum over $\cvxspace_1$ is attained at an extreme point, as shown below.

\begin{restatable}[Reduction to extremal functions]{lemma}{lemReducToExt}
\label{lem:reduce-to-extremal}
Let $\cvxspace_1$ be as defined in \Cref{lem:normalization}. For any distributions $\prior, \distOfMean \in \Delta([0,1])$,
\begin{align*}
\sup_{\cvx \in \cvxspace_1} \frac{\expect[\pooledmean \sim \prior]{\cvx(\pooledmean)}}{\expect[\pooledmean \sim \distOfMean]{\cvx(\pooledmean)}}
=
\sup_{\kink \in [0,1]} \frac{\expect[\pooledmean \sim \prior]{\max\{\kink, \pooledmean\}}}{\expect[\pooledmean \sim \distOfMean]{\max\{\kink, \pooledmean\}}}.
\end{align*}
\end{restatable}

\xhdr{Step 3: Closed-form expression for $\voa(\disclosurefn \mid \cvxspace)$ (Part (ii) of \Cref{thm:rdd:optimal robust policy})}
For any fixed $\kink \in [0,1]$, consider the extremal function $\cvx_\kink(\pooledmean) \triangleq \max\{\kink, \pooledmean\}$. Its piecewise-linear structure ensures that, for any quantile threshold profile $\quants = (\quant_0, \quant_1, \dots, \quant_\signalnum)$, only the unique quantile interval $[\quant_{\quantileidx - 1}, \quant_{\quantileidx}]$ containing $\prior(\kink)$ contributes a Jensen gap (i.e., a loss in the objective of program~$\RDDQuant$). All other intervals yield identical contributions to the numerator and denominator of the ratio in the objective. Consequently, the worst-case competitive ratio reduces to a one-dimensional optimization over the position of $\kink$ within its interval, which admits the closed-form solution stated in part (ii) of \Cref{thm:rdd:optimal robust policy}.

\xhdr{Step 4: Optimizing over quantile threshold profiles (Part (i) of \Cref{thm:rdd:optimal robust policy})}
Minimizing the expression in part (ii) of \Cref{thm:rdd:optimal robust policy} over all valid quantile threshold profiles $\quants$ yields the optimal robust disclosure policy described in part (i) of the same theorem. The optimum is attained when the competitive ratio is equalized across all quantile intervals, which leads to the backward recursion and boundary condition stated in the theorem.

\subsection{Proofs in the Analysis of \texorpdfstring{$\RDDQuantTitle$}{RDD}}
\label{sec:rdd:missing proofs}

In this section, we include all the proofs of the technical lemmas and then establish \Cref{thm:rdd:optimal robust policy}.

\lemNormalization*
\begin{proof}[Proof of \Cref{lem:normalization}]
Fix an arbitrary function $\cvx \in \cvxspace$. If $\cvx(1) = 0$, then since $\cvx$ is non-negative, non-decreasing, and satisfies $\cvx(\pooledmean) \geq \pooledmean \cdot \cvx(1) = 0$, it follows that $\cvx(\pooledmean) = 0$ for all $\pooledmean \in [0,1]$. Such a function contributes zero to both numerator and denominator, and thus does not affect the supremum.
Now suppose $\cvx(1) > 0$, and define the normalized function
\begin{align*}
\newcvx(\pooledmean) \triangleq \frac{\cvx(\pooledmean)}{\cvx(1)}, \quad \pooledmean \in [0,1].
\end{align*}
Since $\cvx$ is non-negative, non-decreasing, and convex, so is $\newcvx$. Moreover, for all $\pooledmean \in [0,1]$,
\begin{align*}
\newcvx(\pooledmean) = \frac{\cvx(\pooledmean)}{\cvx(1)} \geq \frac{\pooledmean \cdot \cvx(1)}{\cvx(1)} = \pooledmean = \pooledmean \cdot \newcvx(1),
\end{align*}
where we used the defining property $\cvx(\pooledmean) \geq \pooledmean \cdot \cvx(1)$. Hence $\newcvx \in \cvxspace$ and $\newcvx(1) = 1$, so $\newcvx \in \cvxspace_1$.
Furthermore, the scaling cancels in the ratio:
\begin{align*}
\frac{\expect[\pooledmean \sim \prior]{\newcvx(\pooledmean)}}{\expect[\pooledmean \sim \distOfMean]{\newcvx(\pooledmean)}}
= \frac{\expect[\pooledmean \sim \prior]{\cvx(\pooledmean)} / \cvx(1)}{\expect[\pooledmean \sim \distOfMean]{\cvx(\pooledmean)} / \cvx(1)}
= \frac{\expect[\pooledmean \sim \prior]{\cvx(\pooledmean)}}{\expect[\pooledmean \sim \distOfMean]{\cvx(\pooledmean)}}.
\end{align*}
Since every $\cvx \in \cvxspace$ with $\cvx(1) > 0$ corresponds to a $\newcvx \in \cvxspace_1$ yielding the same objective value, and functions with $\cvx(1) = 0$ are irrelevant to the supremum, the two suprema coincide. This completes the proof.
\end{proof}

We provide the following auxiliary lemma, which characterizes several structural properties of functions in the normalized indirect utility space $\cvxspace_1$.

\begin{lemma}
\label{lem:derivative-bounds}
Let $\cvxspace_1$ denote the functional space defined in \Cref{lem:normalization}. For any function $\cvx \in \cvxspace_1$:
\begin{enumerate}[label=(\roman*)]
    \item $\pooledmean \leq \cvx(\pooledmean) \leq 1$ for all $\pooledmean \in [0,1]$;
    \item the right derivative $\cvx'_+(\pooledmean)$ exists for all $\pooledmean \in [0,1)$, is non-decreasing, and satisfies $0 \leq \cvx'_+(\pooledmean) \leq 1$;
    \item $\cvx$ is absolutely continuous, and for all $\pooledmean \in [0,1]$,
    \begin{align*}
    \cvx(\pooledmean) - \cvx(0) = \int_0^\pooledmean \cvx'_+(t)\,\dd t.
    \end{align*}
\end{enumerate}
\end{lemma}

\begin{proof}[Proof of \Cref{lem:derivative-bounds}]
We prove the three claims in order.

\xhdr{Part (i)} Since $\cvx \in \cvxspace_1$, it satisfies $\cvx(\pooledmean) \geq \pooledmean \cdot \cvx(1) = \pooledmean$ for all $\pooledmean \in [0,1]$. Moreover, because $\cvx$ is non-decreasing and $\cvx(1) = 1$, we have $\cvx(\pooledmean) \leq \cvx(1) = 1$. Hence $\pooledmean \leq \cvx(\pooledmean) \leq 1$.

\xhdr{Part (ii)} As $\cvx$ is convex on $[0,1]$, its right derivative $\cvx'_+(\pooledmean)$ exists for every $\pooledmean \in [0,1)$ and is non-decreasing. Non-negativity follows from monotonicity: since $\cvx$ is non-decreasing, $\cvx'_+(\pooledmean) \geq 0$.
To establish the upper bound, fix $\pooledmean \in [0,1)$. By convexity, for any $\pooledmean\primed \in (\pooledmean, 1]$,
\begin{align*}
\cvx'_+(\pooledmean) \leq \frac{\cvx(\pooledmean\primed) - \cvx(\pooledmean)}{\pooledmean\primed - \pooledmean}.
\end{align*}
Setting $\pooledmean\primed = 1$ and using $\cvx(1) = 1$, we obtain
\begin{align*}
\cvx'_+(\pooledmean) \leq \frac{1 - \cvx(\pooledmean)}{1 - \pooledmean}.
\end{align*}
Since $\cvx(\pooledmean) \geq \pooledmean$, it follows that $1 - \cvx(\pooledmean) \leq 1 - \pooledmean$, and thus
\begin{align*}
\cvx'_+(\pooledmean) \leq \frac{1 - \pooledmean}{1 - \pooledmean} = 1.
\end{align*}

\xhdr{Pary (iii)} From part (ii), all secant slopes of $\cvx$ lie in $[0,1]$, so $\cvx$ is $1$-Lipschitz and hence absolutely continuous. Furthermore, a convex function is differentiable almost everywhere on $[0,1]$, and at points of differentiability the derivative coincides with the right derivative. Therefore, $\cvx'(\pooledmean) = \cvx'_+(\pooledmean)$ for almost every $\pooledmean \in [0,1]$. By the fundamental theorem of calculus for absolutely continuous functions,
\begin{align*}
\cvx(\pooledmean) - \cvx(0) = \int_0^\pooledmean \cvx'(z)\,\dd z
= \int_0^\pooledmean \cvx'_+(z)\,\dd\ z,
\end{align*}
for all $\pooledmean \in [0,1]$, as claimed.
\end{proof}

\lemExtRepresentation*

\begin{proof}[Proof of \Cref{lem:stoploss-representation}]
By \Cref{lem:derivative-bounds}, the right derivative $\cvx'_+$ of any $\cvx \in \cvxspace_1$ is non-decreasing, takes values in $[0,1]$, and satisfies $\cvx(1) = 1$. We construct a probability measure $\nu_\cvx \in \Delta([0,1])$ as follows:
\begin{align*}
\nu_\cvx([0,\pooledmean]) \triangleq \cvx'_+(\pooledmean), \quad \forall\, \pooledmean \in [0,1),
\qquad
\nu_\cvx(\{1\}) \triangleq 1 - \lim_{\pooledmean \to 1^-} \cvx'_+(\pooledmean).
\end{align*}
Since $\cvx'_+$ is non-decreasing and bounded in $[0,1]$, the limit $\lim_{\pooledmean \to 1^-} \cvx'_+(\pooledmean)$ exists and lies in $[0,1]$, so $\nu_\cvx$ is a well-defined probability measure on $[0,1]$.
Define the function \begin{align}
    \label{eq:construct h via nu}
    \cvx\primed(\pooledmean) \triangleq \int_0^1 \max\{\kink, \pooledmean\}\,\dd\nu_\cvx(\kink), \quad \pooledmean \in [0,1]~.
\end{align}
We now show that $\cvx\primed = \cvx$. Fix $\pooledmean \in [0,1)$ and consider the right derivative of $\cvx\primed$. For any $\delta > 0$,
\begin{align*}
\frac{\cvx\primed(\pooledmean + \delta) - \cvx\primed(\pooledmean)}{\delta}
= \int_0^1 \frac{\max\{\kink, \pooledmean + \delta\} - \max\{\kink, \pooledmean\}}{\delta}\,\dd\nu_\cvx(\kink).
\end{align*}
For each $\kink \in [0,1]$, the integrand is bounded in $[0,1]$ and converges pointwise to $\mathbf{1}\{\kink \leq \pooledmean\}$ as $\delta \downarrow 0$. By the dominated convergence theorem,
\begin{align*}
(\cvx\primed)'_+(\pooledmean)
= \int_0^1 \mathbf{1}\{\kink \leq \pooledmean\}\,\dd\nu_\cvx(\kink)
= \nu_\cvx([0,\pooledmean])
= \cvx'_+(\pooledmean), \quad \forall\, \pooledmean \in [0,1).
\end{align*}
Moreover, $\cvx\primed(1) = \int_0^1 1\,\dd\nu_\cvx(\kink) = 1 = \cvx(1)$.
It remains to verify that $\cvx\primed(0) = \cvx(0)$. Note that
\begin{align*}
\cvx\primed(0) = \int_0^1 \kink\,\dd\nu_\cvx(\kink)
= \int_0^1 \prob[\kink \sim \nu_\cvx]{\kink > t}\,\dd t
= \int_0^1 \left(1 - \nu_\cvx([0,t])\right)\,\dd t
= \int_0^1 \left(1 - \cvx'_+(t)\right)\,\dd t.
\end{align*}
On the other hand, by part (iii) of \Cref{lem:derivative-bounds} and the fact that $\cvx(1) = 1$,
\begin{align*}
\cvx(1) - \cvx(0) = \int_0^1 \cvx'_+(t)\,\dd t
\quad \Rightarrow \quad
\cvx(0) = 1 - \int_0^1 \cvx'_+(t)\,\dd t = \int_0^1 (1 - \cvx'_+(t))\,\dd t.
\end{align*}
Thus $\cvx\primed(0) = \cvx(0)$. Since $\cvx\primed$ and $\cvx$ share the same right derivative on $[0,1)$ and agree at $0$, absolute continuity (part (iii) of \Cref{lem:derivative-bounds}) implies $\cvx\primed \equiv \cvx$ on $[0,1]$.

Finally, we establish uniqueness. Suppose there exist two probability measures $\nu_1, \nu_2 \in \Delta([0,1])$ such that their constructed function $\cvx\primed_1$ and $\cvx\primed_2$ defined according to Eqn.~\eqref{eq:construct h via nu} satisfies $\cvx\primed_1 = \cvx\primed_2$.
Then for all $\pooledmean \in [0,1)$,
\begin{align*}
\nu_1([0,\pooledmean]) = (\cvx\primed_1)'_+(\pooledmean) = (\cvx\primed_2)'_+(\pooledmean) = \nu_2([0,\pooledmean]).
\end{align*}
Hence $\nu_1$ and $\nu_2$ agree on all intervals of the form $[0,\pooledmean]$, which generate the Borel $\sigma$-algebra on $[0,1]$. Therefore, $\nu_1 = \nu_2$ and the proof of \Cref{lem:stoploss-representation} is completed as desired.
\end{proof}

\lemReducToExt*

\begin{proof}[Proof of \Cref{lem:reduce-to-extremal}]
Fix any distributions $\prior, \distOfMean \in \Delta([0, 1])$. By \Cref{lem:stoploss-representation}, for any $\cvx\in\cvxspace_1$ there exists a probability measure $\cvxmeasure$ such that $\cvx= \int_0^1 \cvx_\kink\,\dd\cvxmeasure(\kink)$, where $\cvx_\kink(\pooledmean)\triangleq \max\{\kink,\pooledmean\}$. Define auxiliary functions 
\begin{align*}
    A(\kink) \triangleq 
    \expect[\pooledmean\sim \prior]{\cvx_\kink(\pooledmean)},\qquad 
    B(\kink)\triangleq 
    \expect[\pooledmean\sim \distOfMean]{\cvx_\kink(\pooledmean)}~.
\end{align*}
As function $\cvx_\kink$ is non-negative, we have $A(\kink),B(\kink)\ge 0$. 
If $\int_0^1 B(\kink)\,\dd \cvxmeasure(\kink)=0$, then $B(\kink)=0$ for $\cvxmeasure$-almost every $\kink$, which forces $A(\kink)=0$ $\cvxmeasure$-almost everywhere as well (since $\cvx_\kink\ge 0$ and $\prior, \distOfMean$ are probabilities on $[0,1]$), so the ratio is $0$ and does not affect the supremum. Otherwise we have $\int_0^1 B(\kink)\,\dd \cvxmeasure(\kink) > 0$ and thus
\begin{align*}
    \frac{\expect[\pooledmean\sim \prior]{\cvx(\pooledmean)}}{\expect[\pooledmean\sim \distOfMean]{\cvx(\pooledmean)}}
    =
    \frac{\displaystyle\int_0^1 A(\kink)\,\dd \cvxmeasure(\kink)}{\displaystyle\int_0^1 B(\kink)\,\dd \cvxmeasure(\kink)}
    =
    \frac{\displaystyle\int_0^1 B(\kink)\,\frac{A(\kink)}{B(\kink)}\,\dd \cvxmeasure(\kink)}{\displaystyle\int_0^1 B(\kink)\,\dd \cvxmeasure(\kink)}
    \le
    \sup\limits_{\kink\in[0,1]}\frac{A(\kink)}{B(\kink)}~.
\end{align*}
Taking the supremum over $\cvx\in\cvxspace_1$ would give us
\begin{align*}
    \sup\limits_{\cvx\in\cvxspace_1}\frac{\expect[\pooledmean\sim \prior]{\cvx(\pooledmean)}}{\expect[\pooledmean\sim \distOfMean]{\cvx(\pooledmean)}}
    \le
    \sup\limits_{\kink\in[0,1]}\frac{\expect[\pooledmean\sim \prior]{\cvx_\kink(\pooledmean)}}{\expect[\pooledmean\sim \distOfMean]{\cvx_\kink(\pooledmean)}}~.
\end{align*}
The reverse inequality holds because each $\cvx_\kink\in\cvxspace_1$. This completes the proof of \Cref{lem:reduce-to-extremal}.
\end{proof}

\begin{proof}[Proof of \Cref{thm:rdd:optimal robust policy}]
We first prove Part (ii) of the theorem statement and then prove Part (i).

\xhdr{Part (ii) of \Cref{thm:rdd:optimal robust policy}}
By \Cref{lem:reduce-to-extremal}, it suffices to prove that
\begin{equation}
\label{eq:Gamma-reduce-hc}
    \sup_{\prior \in \Delta([0,1]),\, \kink \in [0,1]}
    \frac{\expect[\pooledmean \sim \prior]{\max\{\kink, \pooledmean\}}}{\expect[y \sim \disclosurefn_\quants(\prior)]{\max\{\kink, y\}}}
    =
    \max_{\quantileidx \in [\signalnum]} \left( 1 + \frac{\quant_{\quantileidx} - \quant_{\quantileidx-1}}{\left(1 + \sqrt{1 - \quant_\quantileidx}\right)^2} \right),
\end{equation}
Fix any {\KQuantilePartitionDisclosure} $\disclosurefn$ with quantile threshold profile $\quants = (\quant_0, \quant_1, \dots, \quant_\signalnum)$. We establish \eqref{eq:Gamma-reduce-hc} by proving matching upper and lower bounds.

\xhdr{Upper bound}
Fix any prior $\prior \in \Delta([0,1])$ and threshold $\kink \in [0,1]$. Let $u \sim \mathrm{Unif}(0,1)$ and define $\pooledmean = \prior^{-1}(u)$, so that $\pooledmean \sim \prior$. For each quantile interval indexed by $\quantileidx \in [\signalnum]$, let $\binrange_\quantileidx = (\quant_{\quantileidx-1}, \quant_\quantileidx]$ and denote the posterior mean in quantile interval $\quantileidx$ by $m_\quantileidx = \expect[u]{\pooledmean \mid u \in \binrange_\quantileidx}$.

Let $\quantileidx \in [\signalnum]$ be the unique index such that $\pooledmean_{\quantileidx-1} \leq \kink \leq \pooledmean_\quantileidx$, where $\pooledmean_\quantileidx = \prior^{-1}(\quant_\quantileidx)$. Define auxiliary notation
\begin{align*}
A \triangleq \kink \cdot \quant_{\quantileidx-1} + \sum\nolimits_{i\in[\quantileidx+1:\signalnum]} (\quant_i - \quant_{i-1}) \cdot m_i.
\end{align*}
Because $\prior^{-1}$ is non-decreasing, we have:
(i) For $i < r$: $\pooledmean \leq \kink$ almost surely on $\binrange_i$, so $\max\{\kink, \pooledmean\} = \kink$ and $\max\{\kink, m_i\} = \kink$;
(ii) For $i > r$: $\pooledmean \geq \kink$ almost surely on $\binrange_i$, so $\max\{\kink, \pooledmean\} = \pooledmean$ and $\max\{\kink, m_i\} = m_i$.
Combing all the pieces, we obtain
\begin{equation}
\label{eq:ratio-split}
    \frac{\expect[\pooledmean \sim \prior]{\max\{\kink, \pooledmean\}}}{\expect[y \sim \disclosurefn(\prior)]{\max\{\kink, y\}}}
    =
    \frac{A + (\quant_\quantileidx - \quant_{\quantileidx-1}) \cdot \expect[\pooledmean \mid u \in \binrange_\quantileidx]{\max\{\kink, \pooledmean\}}}{A + (\quant_\quantileidx - \quant_{\quantileidx-1}) \cdot \max\{\kink, m_\quantileidx\}}.
\end{equation}
Since $\max\{\kink, \cdot\}$ is convex, Jensen's inequality implies the numerator is at least the denominator, and the ratio is non-increasing in $A$. Hence, we may lower bound $A$ to upper bound the ratio.
Define auxiliary notation $b\triangleq \prior^{-1}(\quant_\quantileidx)$. By definition, $m_i \geq b$ for every $i\in[\quantileidx + 1:\signalnum]$ and thus
\begin{align*}
A \geq A_{\min} \triangleq \kink \cdot \quant_{\quantileidx-1} + b \cdot (1-\quant_{\quantileidx}).
\end{align*}
Substituting into \eqref{eq:ratio-split} yields
\begin{align*}
\frac{\expect[\pooledmean \sim \prior]{\max\{\kink, \pooledmean\}}}{\expect[y \sim \disclosurefn(\prior)]{\max\{\kink, y\}}}
\leq
\frac{A_{\min} + (\quant_\quantileidx - \quant_{\quantileidx-1}) \cdot \expect[\pooledmean \mid u \in \binrange_\quantileidx]{\max\{\kink, \pooledmean\}}}{A_{\min} + (\quant_\quantileidx - \quant_{\quantileidx-1}) \cdot \max\{\kink, m_\quantileidx\}}.
\end{align*}
Now, the conditional distribution of $\pooledmean$ given $u \in \binrange_\quantileidx$ is supported on $[0, b]$ and has mean $m_\quantileidx$. By convexity of $\max\{\kink, \cdot\}$, its expectation is maximized when the distribution is two-point at $0$ and $b$. A direct calculation shows that for any such distribution,
\begin{align*}
\expect[\pooledmean \mid u \in \binrange_\quantileidx]{\max\{\kink, \pooledmean\}} \leq \kink + m_\quantileidx \left(1 - \frac{\kink}{b}\right).
\end{align*}
Let $z \triangleq \kink / b \in [0,1]$. A case analysis on whether $m_\quantileidx \leq \kink$ or $m_\quantileidx \geq \kink$ shows that the right-hand side of the ratio is maximized when $m_\quantileidx = \kink$. Substituting $m_\quantileidx = \kink$ and simplifying gives
\begin{align*}
\frac{\expect[\pooledmean \sim \prior]{\max\{\kink, \pooledmean\}}}{\expect[y \sim \disclosurefn(\prior)]{\max\{\kink, y\}}}
\leq
1 + \frac{(\quant_\quantileidx - \quant_{\quantileidx-1}) z(1 - z)}{(1-\quant_\quantileidx) + z \quant_\quantileidx}.
\end{align*}
Term $\frac{z(1 - z)}{(1-\quant_\quantileidx) + z \quant_\quantileidx}$ attains its maximum at
\begin{align*}
z^* = \frac{\sqrt{1-\quant_\quantileidx}}{1 + \sqrt{1-\quant_\quantileidx}},
\end{align*}
yielding the value $\frac{1}{(1 + \sqrt{1-\quant_\quantileidx})^2}$. Therefore,
\begin{align*}
\frac{\expect[\pooledmean \sim \prior]{\max\{\kink, \pooledmean\}}}{\expect[y \sim \disclosurefn(\prior)]{\max\{\kink, y\}}}
\leq
1 + \frac{\quant_\quantileidx-\quant_{\quantileidx-1}}{(1 + \sqrt{1-\quant_\quantileidx})^2}.
\end{align*}
Taking the supremum over $\prior$ and $\kink$ and noting that $\quantileidx$ depends on these choices, we obtain
\begin{equation}
\label{eq:upper-bound}
    \voa(\disclosurefn_\quants \mid \cvxspace)
    \leq
    \max_{\quantileidx \in [\signalnum]} \left(1 + \frac{\quant_\quantileidx-\quant_{\quantileidx-1}}{(1 + \sqrt{1-\quant_\quantileidx})^2}\right).
\end{equation}
which completes the upper bound direction as desired.

\xhdr{Lower bound}
To see the matching lower bound, note that the inner supremum in $\RDDQuant$ is a relaxation of that in $\RQDQuant$. Thus, applying \Cref{lem:quantile partition:matching lower bound} shows the lower bound direction and complete the proof of Part (ii) of \Cref{thm:rdd:optimal robust policy} as desired. 

\xhdr{Part (i) of \Cref{thm:rdd:optimal robust policy}}
Due to Part (ii) shown above, it suffices to show that 
\begin{align*}
    \inf\limits_{\quants=(\quant_0,\quant_1,\dots,\quant_\signalnum)}\;
    \max_{\quantileidx \in [\signalnum]} \left(1 + \frac{\quant_\quantileidx-\quant_{\quantileidx-1}}{(1 + \sqrt{1-\quant_\quantileidx})^2}\right)
    = \voa^*_\signalnum
\end{align*}
\xhdr{Existence and uniqueness of $\voa_\signalnum^*$}
By \Cref{lem:Tk-monotone}, function $\recfn_\signalnum(\voa)$ is continuous and strictly increasing in $\voa\in(1,\infty)$, with $\recfn_\signalnum(1) = 0 < 1$ and $\recfn_\signalnum(\voa)\to\infty$. 
Thus, by the intermediate value theorem, there exists a unique $\voa_\signalnum^*>1$ such that $\recfn_\signalnum(\voa_\signalnum^*)=1$.

\smallskip
\noindent
We next show that $\inf\nolimits_{
\disclosurefn\in\disclosurefnClass_\signalnum
}
\voa(\disclosurefn \mid \cvxspace)=\voa_\signalnum^*$.

\xhdr{Lower bounding $\inf\nolimits_{
    \disclosurefn\in\disclosurefnClass_\signalnum
    }
    \voa(\disclosurefn \mid \cvxspace)\ge \voa_\signalnum^*$}
    Let $\disclosurefn$ be any {\KQuantilePartitionDisclosure} and set $\voa \triangleq \voa(\disclosurefn \mid \cvxspace)$. 
    Then $\voa(\disclosurefn \mid \cvxspace)\le \voa$ holds trivially, so by \Cref{lem:feasibility-equivalence} we have $\recfn_\signalnum(\voa)\ge 1$. 
    Since $\recfn_\signalnum$ is strictly increasing and $\recfn_\signalnum(\voa_\signalnum^*)=1$, it follows that $\voa\ge\voa_\signalnum^*$. 
    Thus, for any {\KQuantilePartitionDisclosure} $\disclosurefn$, we have 
    \begin{align*}
        \voa(\disclosurefn \mid \cvxspace)\ge \voa_\signalnum^*~,
    \end{align*}
    which further implies that $\inf_{\disclosurefn\in\disclosurefnClass_\signalnum} \voa(\disclosurefn \mid \cvxspace)\ge \voa_\signalnum^*$.

\xhdr{Constructing {\KQuantilePartitionDisclosure} $\disclosurefn^*$ such that $\voa(\disclosurefn^* \mid \cvxspace)=\voa_\signalnum^*$}
    We next construct {\KQuantilePartitionDisclosure} $\disclosurefn^*$ with $\voa(\disclosurefn^* \mid \cvxspace) = \voa_\signalnum^*$. 
    In particular, consider $\quants^* = (\quant_0^*,\quant_1^*,\dots,\quant_\signalnum^*)$ according to Eqn.~\eqref{eq:qstar-def}. By construction, $\quants^*$ is a feasible quantile threshold profile and for each $\quantileidx\in[\signalnum]$,
    \begin{align*}
        \quant^*_\quantileidx-\quant^*_{\quantileidx-1} 
        = 
        (\voa_\signalnum^*-1)\left(1+\sqrt{1-\quant_\quantileidx^*}\right)^2
    \end{align*}
    Taking the maximum over $\quantileidx$ and using Part (i) of \Cref{thm:rdd:optimal robust policy} shown above ensures $\voa(\disclosurefn^* \mid \cvxspace)=\voa_\signalnum^*$. 

Combining the two pieces together, we obtain $\inf_{\disclosurefn\in\disclosurefnClass_\signalnum} \voa(\disclosurefn \mid \cvxspace)=\voa_\signalnum^*$, and the constructed {\KQuantilePartitionDisclosure} $\disclosurefn^*$ parameterized by $\quants^*$ is optimal. This completes the Part (i) of \Cref{thm:rdd:optimal robust policy} as desired.
\end{proof}

\newpage
\bibliography{mybib}

\begin{thebibliography}{55}
\providecommand{\natexlab}[1]{#1}
\providecommand{\url}[1]{\texttt{#1}}
\expandafter\ifx\csname urlstyle\endcsname\relax
  \providecommand{\doi}[1]{doi: #1}\else
  \providecommand{\doi}{doi: \begingroup \urlstyle{rm}\Url}\fi

\bibitem[Agrawal et~al.(2023)Agrawal, Feng, and Tang]{AFT-23}
Shipra Agrawal, Yiding Feng, and Wei Tang.
\newblock Dynamic pricing and learning with bayesian persuasion.
\newblock \emph{Advances in Neural Information Processing Systems}, 36:\penalty0 59273--59285, 2023.

\bibitem[{Airbnb}(2025)]{AirbnbTopHomesHighlight}
{Airbnb}.
\newblock What to know about the highlight for top homes.
\newblock \url{https://www.airbnb.com/resources/hosting-homes/a/what-to-know-about-the-highlight-for-top-homes-666}, 2025.
\newblock Accessed: 2026-01-26.

\bibitem[Akerlof(1978)]{A-78}
George~A Akerlof.
\newblock The market for ``lemons'': Quality uncertainty and the market mechanism.
\newblock In \emph{Uncertainty in economics}, pages 235--251. Elsevier, 1978.

\bibitem[Alaei et~al.(2019)Alaei, Hartline, Niazadeh, Pountourakis, and Yuan]{AHNPY-19}
Saeed Alaei, Jason Hartline, Rad Niazadeh, Emmanouil Pountourakis, and Yang Yuan.
\newblock Optimal auctions vs. anonymous pricing.
\newblock \emph{Games and Economic Behavior}, 118:\penalty0 494--510, 2019.

\bibitem[Ali et~al.(2022)Ali, Haghpanah, Lin, and Siegel]{AHLS-22}
S~Nageeb Ali, Nima Haghpanah, Xiao Lin, and Ron Siegel.
\newblock How to sell hard information.
\newblock \emph{The Quarterly Journal of Economics}, 137\penalty0 (1):\penalty0 619--678, 2022.

\bibitem[Babaioff et~al.(2012)Babaioff, Kleinberg, and Paes~Leme]{BKP-12}
Moshe Babaioff, Robert Kleinberg, and Renato Paes~Leme.
\newblock Optimal mechanisms for selling information.
\newblock In \emph{Proceedings of the 13th ACM Conference on Electronic Commerce}, pages 92--109, 2012.

\bibitem[Babichenko et~al.(2022)Babichenko, Talgam-Cohen, Xu, and Zabarnyi]{BTXZ-22}
Yakov Babichenko, Inbal Talgam-Cohen, Haifeng Xu, and Konstantin Zabarnyi.
\newblock Regret-minimizing bayesian persuasion.
\newblock \emph{Games and Economic Behavior}, 136:\penalty0 226--248, 2022.

\bibitem[Bergemann et~al.(2022{\natexlab{a}})Bergemann, Cai, Velegkas, and Zhao]{BCVZ-22}
Dirk Bergemann, Yang Cai, Grigoris Velegkas, and Mingfei Zhao.
\newblock Is selling complete information (approximately) optimal?
\newblock \emph{arXiv preprint arXiv:2202.09013}, 2022{\natexlab{a}}.

\bibitem[Bergemann et~al.(2022{\natexlab{b}})Bergemann, Heumann, Morris, Sorokin, and Winter]{BHMSW-22}
Dirk Bergemann, Tibor Heumann, Stephen Morris, Constantine Sorokin, and Eyal Winter.
\newblock Optimal information disclosure in auctions.
\newblock \emph{American Economic Review: Insights}, 2022{\natexlab{b}}.

\bibitem[Bergemann et~al.(2026)Bergemann, Heumann, and Morris]{BHM-26}
Dirk Bergemann, Tibor Heumann, and Stephen Morris.
\newblock Screening with persuasion.
\newblock \emph{Journal of Political Economy}, 134\penalty0 (2):\penalty0 000--000, 2026.

\bibitem[Bro~Miltersen and Sheffet(2012)]{BS-12}
Peter Bro~Miltersen and Or~Sheffet.
\newblock Send mixed signals: earn more, work less.
\newblock In \emph{Proceedings of the 13th ACM Conference on Electronic Commerce}, pages 234--247, 2012.

\bibitem[Cai et~al.(2025)Cai, Li, and Wu]{CLW-25}
Yang Cai, Yingkai Li, and Jinzhao Wu.
\newblock Information disclosure makes simple mechanisms competitive.
\newblock \emph{arXiv preprint arXiv:2502.17809}, 2025.

\bibitem[Candogan and Strack(2021)]{CS-21}
Ozan Candogan and Philipp Strack.
\newblock Optimal disclosure of information to a privately informed receiver.
\newblock In \emph{Proceedings of the 22nd ACM Conference on Economics and Computation}, pages 263--263, 2021.

\bibitem[Castiglioni et~al.(2020)Castiglioni, Celli, Marchesi, and Gatti]{CCMG-20}
Matteo Castiglioni, Andrea Celli, Alberto Marchesi, and Nicola Gatti.
\newblock Online bayesian persuasion.
\newblock \emph{Advances in Neural Information Processing Systems}, 33, 2020.

\bibitem[Castiglioni et~al.(2021)Castiglioni, Marchesi, Celli, and Gatti]{CMCG-21}
Matteo Castiglioni, Alberto Marchesi, Andrea Celli, and Nicola Gatti.
\newblock Multi-receiver online bayesian persuasion.
\newblock In \emph{International Conference on Machine Learning}, pages 1314--1323. PMLR, 2021.

\bibitem[Chawla and Sivan(2014)]{CS-14}
Shuchi Chawla and Balasubramanian Sivan.
\newblock Bayesian algorithmic mechanism design.
\newblock \emph{ACM SIGecom Exchanges}, 13\penalty0 (1):\penalty0 5--49, 2014.

\bibitem[Chawla et~al.(2010)Chawla, Hartline, Malec, and Sivan]{CHMS-10}
Shuchi Chawla, Jason~D Hartline, David~L Malec, and Balasubramanian Sivan.
\newblock Multi-parameter mechanism design and sequential posted pricing.
\newblock In \emph{Proceedings of the forty-second ACM symposium on Theory of computing}, pages 311--320, 2010.

\bibitem[Chen et~al.(2025)Chen, Lin, Tang, and Tucker-Foltz]{CLTT-25}
Yiling Chen, Tao Lin, Wei Tang, and Jamie Tucker-Foltz.
\newblock Explainable information design.
\newblock \emph{arXiv preprint arXiv:2508.14196}, 2025.

\bibitem[Crapis et~al.(2017)Crapis, Ifrach, Maglaras, and Scarsini]{CIMS-17}
Davide Crapis, Bar Ifrach, Costis Maglaras, and Marco Scarsini.
\newblock Monopoly pricing in the presence of social learning.
\newblock \emph{Management Science}, 63\penalty0 (11):\penalty0 3586--3608, 2017.

\bibitem[Dworczak and Martini(2019)]{DM-19}
Piotr Dworczak and Giorgio Martini.
\newblock The simple economics of optimal persuasion.
\newblock \emph{Journal of Political Economy}, 127\penalty0 (5):\penalty0 1993--2048, 2019.

\bibitem[Dworczak and Pavan(2022)]{DP-22}
Piotr Dworczak and Alessandro Pavan.
\newblock Preparing for the worst but hoping for the best: Robust (bayesian) persuasion.
\newblock \emph{Econometrica}, 90\penalty0 (5):\penalty0 2017--2051, 2022.

\bibitem[Emek et~al.(2014)Emek, Feldman, Gamzu, PaesLeme, and Tennenholtz]{EFGPT-14}
Yuval Emek, Michal Feldman, Iftah Gamzu, Renato PaesLeme, and Moshe Tennenholtz.
\newblock Signaling schemes for revenue maximization.
\newblock \emph{ACM Transactions on Economics and Computation (TEAC)}, 2\penalty0 (2):\penalty0 1--19, 2014.

\bibitem[Es{\H{o}} and Szentes(2007)]{ES-07}
P{\'e}ter Es{\H{o}} and Balazs Szentes.
\newblock Optimal information disclosure in auctions and the handicap auction.
\newblock \emph{The Review of Economic Studies}, 74\penalty0 (3):\penalty0 705--731, 2007.

\bibitem[Feng and Jin(2024)]{FJ-24}
Yiding Feng and Yaonan Jin.
\newblock Beyond regularity: Simple versus optimal mechanisms, revisited.
\newblock \emph{arXiv preprint arXiv:2411.03583}, 2024.

\bibitem[Feng et~al.(2022)Feng, Tang, and Xu]{FTX-22}
Yiding Feng, Wei Tang, and Haifeng Xu.
\newblock Online bayesian recommendation with no regret.
\newblock In \emph{Proceedings of the 23rd ACM Conference on Economics and Computation}, pages 818--819, 2022.

\bibitem[Feng et~al.(2024)Feng, Ho, and Tang]{FHT-24}
Yiding Feng, Chien-Ju Ho, and Wei Tang.
\newblock Rationality-robust information design: Bayesian persuasion under quantal response.
\newblock In \emph{Proceedings of the 2024 Annual ACM-SIAM Symposium on Discrete Algorithms (SODA)}, pages 501--546. SIAM, 2024.

\bibitem[Gan and Li(2024)]{GL-24}
Tan Gan and Hongcheng Li.
\newblock Robust pricing for quality disclosure.
\newblock \emph{arXiv preprint arXiv:2404.06019}, 2024.

\bibitem[Gradwohl et~al.(2021)Gradwohl, Hahn, Hoefer, and Smorodinsky]{GHHS-21}
Ronen Gradwohl, Niklas Hahn, Martin Hoefer, and Rann Smorodinsky.
\newblock Algorithms for persuasion with limited communication.
\newblock In \emph{Proceedings of the Thirty-Second Annual ACM-SIAM Symposium on Discrete Algorithms}, pages 637--652, 2021.

\bibitem[Grossman(1981)]{G-81}
Sanford~J Grossman.
\newblock The informational role of warranties and private disclosure about product quality.
\newblock \emph{The Journal of law and Economics}, 24\penalty0 (3):\penalty0 461--483, 1981.

\bibitem[Guo et~al.(2025)Guo, Hao, and Shi]{GHS-25}
Yingni Guo, Li~Hao, and Xianwen Shi.
\newblock Optimal discriminatory disclosure.
\newblock \emph{Journal of Economic Theory}, 224:\penalty0 105972, 2025.

\bibitem[Hartline and Lucier(2010)]{HL-10}
Jason~D Hartline and Brendan Lucier.
\newblock Bayesian algorithmic mechanism design.
\newblock In \emph{Proceedings of the forty-second ACM symposium on Theory of computing}, pages 301--310, 2010.

\bibitem[Hartline and Roughgarden(2009)]{HR-09}
Jason~D Hartline and Tim Roughgarden.
\newblock Simple versus optimal mechanisms.
\newblock In \emph{Proceedings of the 10th ACM conference on Electronic commerce}, pages 225--234, 2009.

\bibitem[Ifrach et~al.(2019)Ifrach, Maglaras, Scarsini, and Zseleva]{IMSZ-19}
Bar Ifrach, Costis Maglaras, Marco Scarsini, and Anna Zseleva.
\newblock Bayesian social learning from consumer reviews.
\newblock \emph{Operations Research}, 67\penalty0 (5):\penalty0 1209--1221, 2019.

\bibitem[Kamenica and Gentzkow(2011)]{KG-11}
Emir Kamenica and Matthew Gentzkow.
\newblock Bayesian persuasion.
\newblock \emph{American Economic Review}, 101\penalty0 (6):\penalty0 2590--2615, 2011.

\bibitem[Kolotilin et~al.(2017)Kolotilin, Mylovanov, Zapechelnyuk, and Li]{KMZL-17}
Anton Kolotilin, Tymofiy Mylovanov, Andriy Zapechelnyuk, and Ming Li.
\newblock Persuasion of a privately informed receiver.
\newblock \emph{Econometrica}, 85\penalty0 (6):\penalty0 1949--1964, 2017.

\bibitem[Kolotilin et~al.(2022)Kolotilin, Mylovanov, and Zapechelnyuk]{KMZ-22}
Anton Kolotilin, Timofiy Mylovanov, and Andriy Zapechelnyuk.
\newblock Censorship as optimal persuasion.
\newblock \emph{Theoretical Economics}, 17\penalty0 (2):\penalty0 561--585, 2022.

\bibitem[Kolotilin et~al.(2025)Kolotilin, Li, and Zapechelnyuk]{KLZ-25}
Anton Kolotilin, Hongyi Li, and Andy Zapechelnyuk.
\newblock On monotone persuasion.
\newblock In \emph{Proceedings of the 26th ACM Conference on Economics and Computation}, pages 253--253, 2025.

\bibitem[Kosterina(2022)]{K-22}
Svetlana Kosterina.
\newblock Persuasion with unknown beliefs.
\newblock \emph{Theoretical Economics}, 17\penalty0 (3):\penalty0 1075--1107, 2022.

\bibitem[Li and Shi(2017)]{LS-17}
Hao Li and Xianwen Shi.
\newblock Discriminatory information disclosure.
\newblock \emph{American Economic Review}, 107\penalty0 (11):\penalty0 3363--3385, 2017.

\bibitem[Li et~al.(2025)Li, Pu, and Yang]{LPY-25}
Ming Li, Binyan Pu, and Renkun Yang.
\newblock The design of quality disclosure policy and the limits to competition.
\newblock In \emph{Proceedings of the 26th ACM Conference on Economics and Computation}, pages 640--640, 2025.

\bibitem[Lin and Li(2025)]{LL-25}
Tao Lin and Ce~Li.
\newblock Information design with unknown prior.
\newblock In \emph{16th Innovations in Theoretical Computer Science Conference (ITCS 2025)}, pages 72--1. Schloss Dagstuhl--Leibniz-Zentrum f{\"u}r Informatik, 2025.

\bibitem[Liu et~al.(2021)Liu, Shen, and Xu]{LSX-21}
Shuze Liu, Weiran Shen, and Haifeng Xu.
\newblock Optimal pricing of information.
\newblock In \emph{Proceedings of the 22nd ACM Conference on Economics and Computation}, pages 693--693, 2021.

\bibitem[Mensch(2021)]{M-21}
Jeffrey Mensch.
\newblock Monotone persuasion.
\newblock \emph{Games and Economic Behavior}, 130:\penalty0 521--542, 2021.

\bibitem[Milgrom and Shannon(1994)]{MS-94}
Paul Milgrom and Chris Shannon.
\newblock Monotone comparative statics.
\newblock \emph{Econometrica}, pages 157--180, 1994.

\bibitem[Milgrom(1981)]{M-81}
Paul~R Milgrom.
\newblock Good news and bad news: Representation theorems and applications.
\newblock \emph{The Bell Journal of Economics}, pages 380--391, 1981.

\bibitem[Onuchic and Ray(2023)]{OR-23}
Paula Onuchic and Debraj Ray.
\newblock Conveying value via categories.
\newblock \emph{Theoretical Economics}, 18\penalty0 (4):\penalty0 1407--1439, 2023.

\bibitem[Ottaviani and Prat(2001)]{OP-01}
Marco Ottaviani and Andrea Prat.
\newblock The value of public information in monopoly.
\newblock \emph{Econometrica}, 69\penalty0 (6):\penalty0 1673--1683, 2001.

\bibitem[Rayo and Segal(2010)]{RS-10}
Luis Rayo and Ilya Segal.
\newblock Optimal information disclosure.
\newblock \emph{Journal of political Economy}, 118\penalty0 (5):\penalty0 949--987, 2010.

\bibitem[Roughgarden(2015)]{Rou-15}
Tim Roughgarden.
\newblock Approximately optimal mechanism design: Motivation, examples, and lessons learned.
\newblock \emph{ACM SIGecom Exchanges}, 13\penalty0 (2):\penalty0 4--20, 2015.

\bibitem[Shin et~al.(2022)Shin, Vaccari, and Zeevi]{SVZ-22}
Dongwook Shin, Stefano Vaccari, and Assaf Zeevi.
\newblock Dynamic pricing with online reviews.
\newblock \emph{Management Science}, 2022.

\bibitem[Strack and Yang(2024)]{SY-24}
Philipp Strack and Kai~Hao Yang.
\newblock Privacy-preserving signals.
\newblock \emph{Econometrica}, 92\penalty0 (6):\penalty0 1907--1938, 2024.

\bibitem[{Upwork Support}()]{Upwork}
{Upwork Support}.
\newblock Learn about {Upwork}'s talent badges.
\newblock \url{https://support.upwork.com/hc/en-us/articles/360049702614-Learn-about-Upwork-s-talent-badges}.
\newblock Accessed: 2026-01.

\bibitem[Wei and Green(2024)]{WG-24}
Dong Wei and Brett Green.
\newblock (reverse) price discrimination with information design.
\newblock \emph{American Economic Journal: Microeconomics}, 16\penalty0 (2):\penalty0 267--295, 2024.

\bibitem[Wilson(1985)]{W-85}
Robert Wilson.
\newblock Game-theoretic analysis of trading processes.
\newblock Technical report, 1985.

\bibitem[Zu et~al.(2024)Zu, Iyer, and Xu]{ZIX-21}
You Zu, Krishnamurthy Iyer, and Haifeng Xu.
\newblock Learning to persuade on the fly: Robustness against ignorance.
\newblock \emph{Operations Research}, 2024.

\end{thebibliography}
\newpage

\appendix

\section{Relaxing the Quality Linearity in Valuation Function}
\label{apx:extensions general val}

In this section, we discuss extensions of some of our results under a relaxation of the quality-linearity assumption in the valuation function (\Cref{asp:linear valuation function}).

\begin{assumption}[Constant marginal rate of substitution (CMRS)]
\label{assump:multiplicative separability}
The valuation function $\buyerUtility(\type, \quality)$ is increasing in quality $\quality \in [0,1]$ for every fixed type $\type \in [0,1]$, and non-decreasing in type $\type$ for every fixed quality $\quality \in [0,1]$, and it satisfies that for any $\type\in[0, 1]$, and any $\quality\primed,\quality\doubleprimed\in[0, 1]$,
\begin{align*}
    \frac{\left.\frac{\partial \buyerUtility(\type, \quality)}{\partial \type}\right|_{\type, \quality\primed}}{\left.\frac{\partial \buyerUtility(\type, \quality)}{\partial \type}\right|_{\type, \quality\doubleprimed}}=\frac{\cmrsfn\left(\quality\primed\right)}{\cmrsfn\left(\quality\doubleprimed\right)}~,
\end{align*}
where $\cmrsfn(\cdot): [0, 1] \rightarrow \R$ is some function that only depends on the quality.
\end{assumption}

Intuitively, this condition says that the relative effect of a higher type on valuation at two different qualities is pinned down solely by those qualities, and does not depend on the level of $\type$. Equivalently, type acts as a ``common shifter'' of marginal willingness-to-pay across qualities, with the quality-dependent scaling captured $\typefn(\cdot)$.
Many valuation functions including $\buyerUtility(\type, \quality) = \type+\quality$, $\buyerUtility(\type, \quality) = \type\cdot\quality$
and $\buyerUtility(\type, \quality) = \type\cdot\quality + \quality$ all satisfy this assumption.

We denote by $\buyerUtilityClass_{\cmrs}$ the set of valuation functions that satisfy \Cref{assump:multiplicative separability}.
\Cref{assump:multiplicative separability} allows the valuation function in $\buyerUtilityClass_{\cmrs}$ to depend on the quality in a non-linear way. 
However, we note that \Cref{assump:multiplicative separability} is not a strict relaxation of the \Cref{asp:linear valuation function}.
Consider the following valuation function $\buyerUtility(\type, \quality) = (1+\type^2)\quality+\type$. 
Clearly this function satisfies \Cref{asp:linear valuation function} but does not satisfy the above \Cref{assump:multiplicative separability}.

\begin{proposition}
\label{prop:full infor optimal no infor 2 approx cmrs}
    Fix any valuation function $\buyerUtility \in \buyerUtilityClass_{\cmrs}$, consumer type distribution $\typeCDF \in \Delta([0, 1])$, and product quality distribution $\prior \in \Delta([0, 1])$. 
    The full-information signaling scheme $\fullinfor$ achieves the optimal revenue, i.e.,
    \begin{align*}
        \Rev(\fullinfor) = \OPT~.
    \end{align*}
    Moreover, the no-information signaling scheme $\noinfor$, guarantees at least half of the optimal revenue, i.e.,
    \begin{align*}
        \Rev(\noinfor) \geq \frac{1}{2} \cdot \OPT~.
    \end{align*}
    Furthermore, this approximation ratio of $2$ is tight for the no-information signaling scheme $\noinfor$.
\end{proposition}

For linear valuation functions $\buyerUtility\in\buyerUtilityClass$ or CMRS valuation $\buyerUtility\in\buyerUtilityClass_{\cmrs}$, we have established that no-information signaling scheme is always an $2$-approximation to the optimal revenue. 
Unfortunately, this $2$-approximate optimality of the no-information signaling scheme need not continue to hold if one relaxes the condition of the constant marginal rate of substitution further.
Below, we show that when one considers a valuation function that satisfies a single-crossing condition, the approximation ratio of the no-information signaling scheme may exceed the factor $2$.
\begin{assumption}[Milgrom-Shannon single-crossing, \citealp{MS-94}]
\label{assump:single-crossing}
The valuation function $\buyerUtility(\type, \quality)$ is increasing in quality $\quality \in [0,1]$ for every fixed type $\type \in [0,1]$, and non-decreasing in type $\type$ for every fixed quality $\quality \in [0,1]$, and it satisfies that for any $\type_1 > \type_2$ and $\quality_1 > \quality_2$, we have that 
\begin{align*}
    \buyerUtility(\type_1, \quality_1)
    - 
    \buyerUtility(\type_2, \quality_1)
    > 
    \buyerUtility(\type_1, \quality_2)
    - 
    \buyerUtility(\type_2, \quality_2)~.
\end{align*}
\end{assumption}

\begin{proposition}
\label{prop:beyond 2 sc}
There exists an instance with a valuation function $\buyerUtility$ satisfying \Cref{assump:single-crossing}, and product quality distribution $\prior$ and a type distribution $\typeCDF$ such that 
$2.73 \Rev(\noinfor) \le \OPT$.
\end{proposition}

\subsection{Proof of Proposition~\ref{prop:full infor optimal no infor 2 approx cmrs}}

To prove \Cref{prop:full infor optimal no infor 2 approx cmrs}, we first establish the following equivalent representation of the valuation function $\buyerUtility\in\buyerUtilityClass_{\cmrs}$.

\begin{lemma}[Separable representation for $\buyerUtility\in \buyerUtilityClass_{\cmrs}$]
\label{lem:separable}
If $\buyerUtility\in\buyerUtilityClass_{\cmrs}$, then there exist functions
$\qualityfntwo:[0,1]\to\reals_+$, $\qualityfnone:[0,1]\to\reals_+$, and a non-decreasing function $\typefn:[0,1]\to\reals_+$ with $\typefn(0)=0$ such that for all $(\type,\quality)\in[0,1]^2$,
\begin{align*}
    \buyerUtility(\type,\quality)\;=\; \qualityfntwo(\quality)+\qualityfnone(\quality)\,\typefn(\type)~.
\end{align*}
\end{lemma}

\begin{proof}
If $\frac{\partial \buyerUtility(\type,\quality)}{\partial\type}=0$ for all $(\type,\quality)$, then the valuation function $\buyerUtility$ does not depend on $\type$.
In this case, we can set $\typefn(\type)=0$, $\qualityfnone(\quality)=0$, and $\qualityfntwo(\quality)=\buyerUtility(0,\quality)$.

Otherwise, $\frac{\partial \buyerUtility}{\partial\type}$ is not identically zero.
By \Cref{assump:multiplicative separability}, we know that there exists a function $\cmrsfn$ such that for all $\type$ and all $\quality\primed,\quality\doubleprimed$,
\begin{align*}
\frac{\dfrac{\partial \buyerUtility(\type,\quality\primed)}{\partial\type}}
{\dfrac{\partial \buyerUtility(\type,\quality\doubleprimed)}{\partial\type}}
\;=\;
\frac{\cmrsfn(\quality\primed)}{\cmrsfn(\quality\doubleprimed)}
\end{align*}
whenever well-defined.
Since $\buyerUtility$ is non-decreasing in $\type$ for each fixed $\quality$, we have $\frac{\partial \buyerUtility(\type,\quality)}{\partial\type}\ge 0$ for all $(\type,\quality)$, hence the ratio on the left is nonnegative.
This implies that $\cmrsfn(\quality)$ has a constant sign wherever it is nonzero; replacing $\cmrsfn$ by $-\cmrsfn$ if needed, we may assume
$\cmrsfn(\quality)\ge 0$ for all $\quality$.

Because $\frac{\partial \buyerUtility}{\partial\type}$ is not identically zero, there must exist $\quality_0\in[0,1]$ with $\cmrsfn(\quality_0)>0$. Fix such an $\quality_0$ and define
\begin{align*}
    \newtypefn(\type)\;\triangleq\; \frac{1}{\cmrsfn(\quality_0)}\frac{\partial \buyerUtility(\type,\quality_0)}{\partial\type}~,
    \qquad
    \qualityfnone(\quality)\;\triangleq\; \cmrsfn(\quality)~.
\end{align*}
Then for every $(\type,\quality)$, we have that
\begin{align*}
    \frac{\partial \buyerUtility(\type,\quality)}{\partial\type}
    =
    \frac{\cmrsfn(\quality)}{\cmrsfn(\quality_0)}\frac{\partial \buyerUtility(\type,\quality_0)}{\partial\type}
    =
    \qualityfnone(\quality)\,\newtypefn(\type)~.
\end{align*}
Since $\frac{\partial \buyerUtility(\type,\quality_0)}{\partial\type}\ge 0$ and $\cmrsfn(\quality_0)>0$, we have $\newtypefn(\type)\ge 0$ for all $\type$.
Define
\begin{align*}
    \typefn(\type)\;\triangleq\;\int_0^\type \newtypefn(s)\, \mathrm{d}s.
\end{align*}
Then $\typefn(0)=0$ and $\typefn$ is non-decreasing with $\typefn(\type)\ge 0$ for all $\type$.
Finally define $\qualityfntwo(\quality)\triangleq \buyerUtility(0,\quality)$. For each fixed $\quality$, the fundamental theorem of calculus gives
\begin{align*}
    \buyerUtility(\type,\quality)
    =
    \buyerUtility(0,\quality)+\int_0^\type \frac{\partial \buyerUtility(s,\quality)}{\partial\type}\, \mathrm{d}s
    =
    \qualityfntwo(\quality)+\int_0^\type \qualityfnone(\quality)\,\newtypefn(s)\, \mathrm{d}s
    =
    \qualityfntwo(\quality)+\qualityfnone(\quality)\,\typefn(\type)~.
\end{align*}
This proves the claimed representation.
\end{proof}

\xhdr{Optimality of $\fullinfor$}
When valuation function satisfies \Cref{asp:linear valuation function}, we have shown that its induced optimal posted-price revenue function is convex (see \Cref{lem:revenue function necessary condition}).
To prove the optimality of $\fullinfor$, we first show that similar results also hold for the valuation function in $\buyerUtilityClass_{\cmrs}$.
\begin{definition}[Optimal posted-price revenue function]
Given any valuation function $\buyerUtility\in\buyerUtilityClass_{\cmrs}$ and any type distribution $\typeCDF \in \Delta([0,1])$, let $\cvx_{\buyerUtility, \typeCDF}(\cdot) : [0, 1] \to \reals_+$ denote the \emph{optimal posted-price revenue function}, defined as follows: for any induced posterior $\posterior$,
\begin{align*}
    \cvx_{\buyerUtility, \typeCDF}(\posterior)
    \triangleq 
    \sup_{\price \ge 0}\; \price \cdot 
    \prob[\type \sim \typeCDF]{\expect[\quality\sim \posterior]{\buyerUtility(\type, \quality)} \ge \price}~.
\end{align*}
\end{definition}
\begin{lemma}[Convexity of $\cvx_{\buyerUtility, \typeCDF}$]
\label{lem:convex cmrs}
The function $\cvx_{\buyerUtility, \typeCDF}:\Delta([0, 1]) \to \reals_+$ is convex.
\end{lemma}
\begin{proof}
Fix $\buyerUtility\in\buyerUtilityClass_{\cmrs}$ and the decomposition in \Cref{lem:separable}.
Given any posterior $\posterior$ over quality $\quality$, we define the two posterior moments
\begin{align*}
    \qualityfntwo(\posterior)\;\triangleq\; \expect[\quality\sim\posterior]{\qualityfntwo(\quality)},
    \qquad
    \qualityfnone(\posterior)\;\triangleq\; \expect[\quality\sim\posterior]{\qualityfnone(\quality)}~.
\end{align*}
Then for any $\type\in[0,1]$,
\begin{align*}
    \E_{\quality\sim\posterior}\bigl[\buyerUtility(\type,\quality)\bigr]
    =
    \E_{\quality\sim\posterior}\bigl[\qualityfntwo(\quality)+\qualityfnone(\quality)\typefn(\type)\bigr]
    =
    \qualityfntwo(\posterior)+\qualityfnone(\posterior)\,\typefn(\type)~.
\end{align*}
Then by definition, we can represent the optimal posted-price revenue function $\cvx_{\buyerUtility, \typeCDF}$ as follows:
\begin{align*}
    \cvx_{\buyerUtility, \typeCDF}(\posterior)
    = 
    \sup_{z\ge 0}\; (\qualityfntwo(\posterior) + \qualityfnone(\posterior) z)\, \cdot \prob[\type \sim \typeCDF]{\typefn(\type) \ge z}~.
\end{align*}
For each fixed $z\ge 0$, the map $\posterior \mapsto (\qualityfntwo(\posterior) + \qualityfnone(\posterior) z)\, \cdot \prob[\type \sim \typeCDF]{\typefn(\type) \ge z}$ is affine.
Since $\cvx_{\buyerUtility, \typeCDF}(\posterior)$ is the pointwise supremum over $z\ge 0$ of affine functions, it is convex.
\end{proof}
With \Cref{lem:convex cmrs}, we are ready to prove the optimality of the full-information signaling scheme $\fullinfor$.
\begin{proof}[Proof of optimality of $\fullinfor$]
Given any $\buyerUtility\in\buyerUtilityClass_{\cmrs}$, and any quality distribution $\prior$ and type distribution $\typeCDF$,
for any signaling scheme $\signalscheme$, we can represent the revenue function $\Rev(\signalscheme)$ as follows:
\begin{align*}
    \Rev(\signalscheme) 
    & =
    \sup\nolimits_{\{\pricingscheme(\signal)\}_{\signal\in\signalSpace}}
    \Rev(\signalscheme, \pricingscheme) \\
    & = 
    \expect[\quality\sim\prior]{\expect[\signal\sim \signalscheme(\cdot\mid \quality)]{\sup\nolimits_{\pricingscheme(\signal)\ge 0}
    \; \pricingscheme(\signal) \cdot 
    \prob[\type \sim \typeCDF]{\expect[\quality\sim \posterior]{\buyerUtility(\type, \quality)} \ge \pricingscheme(\signal)}}}\\
    & = 
    \expect[\quality\sim\prior]{\expect[\signal\sim \signalscheme(\cdot\mid \quality)]{\cvx_{\buyerUtility, \typeCDF}(\posterior(\cdot\mid\signal)}}
    \le \Rev(\fullinfor)~,
\end{align*}
where the last inequality is due to the convexity of the function $\cvx_{\buyerUtility, \typeCDF}$ according to \Cref{lem:convex cmrs}, and the Bayes-consistency condition $\expect[\quality\sim\prior]{\expect[\signal\sim\signalscheme(\cdot\mid\quality)]{\posterior(\cdot\mid\signal)}} = \prior$. We thus have proved the optimality of $\fullinfor$.
\end{proof}

\xhdr{The tight $2$-approximation of $\noinfor$}
We now prove the tight $2$-approximation of $\noinfor$ for any $\buyerUtility\in\buyerUtilityClass_{\cmrs}$.
We first show that $2 \Rev(\noinfor)\ge \OPT$.
\begin{lemma}
Fix any valuation function $\buyerUtility\in\buyerUtilityClass_{\cmrs}$ with its separable representation in \Cref{lem:separable}, 
for any $\prior, \typeCDF \in \Delta([0, 1])$, we have
\begin{align*}
    \OPT 
    & \le \expect[\prior]{\qualityfnone(\quality)}\cdot\Myer(\typeCDF) + \expect[\prior]{\qualityfntwo(\quality)}~,\\
    \Rev(\noinfor) 
    & \ge \expect[\prior]{\qualityfnone(\quality)}\cdot \Myer(\typeCDF) \vee \expect[\prior]{\qualityfntwo(\quality)}~.
\end{align*}
where $\Myer(\typeCDF) \triangleq \max_\price \price\cdot (1 - \typeCDF(\price))$. 
Therefore, we have $2 \Rev(\noinfor)\ge \OPT$.
\end{lemma}
\begin{proof}
We first upper bound the $\OPT$, then lower bound the $\Rev(\noinfor)$.

\xhdr{Upper bounding $\textsc{OPT}$}
Fix any information policy $\signalscheme$ with the signal space $\signalSpace$, let $\posterior(\cdot\mid \signal) \in\Delta([0, 1])$ be the consumer's posterior belief induced by observing the realized signal $\signal\in\signalSpace$.
For any signal $\signal\in\signalSpace$ and its corresponding price $\pricingscheme(\signal)$, we define 
\begin{align*}
    \inverseVal(\pricingscheme(\signal), \posterior(\cdot\mid \signal)) 
    & \triangleq \arg\inf_{\type} \left\{\type: 
    \type\cdot \expect[\posterior(\cdot\mid \signal)]{\qualityfnone(\quality)} + \expect[\posterior(\cdot\mid \signal)]{\qualityfntwo(\quality)} \ge \pricingscheme(\signal)
    \right\} \\
    \Leftrightarrow 
    \inverseVal(\pricingscheme(\signal),  \posterior(\cdot\mid \signal)) 
    & = \frac{\pricingscheme(\signal) - \expect[ \posterior(\cdot\mid \signal)]{\qualityfntwo(\quality)}}{\expect[ \posterior(\cdot\mid \signal)]{\qualityfnone(\quality)}}
\end{align*}
We use $\signalscheme(\signal) = \sum_{\quality\in [0, 1]}\priorPDF(\quality)\cdot \signalscheme(\signal\mid\quality)$ to denote the marginal probability for generating such signal $\signal$.
We can now upper bound the seller's revenue $\Rev(\signalscheme, \price)$ under any information policy $\signalscheme$ and the pricing scheme $\pricingscheme$ as follows:
\begin{align*}
    \Rev(\signalscheme, \pricingscheme)
    & =
    \sum\nolimits_{\signal\in\signalSpace} \pricingscheme(\signal) \cdot \signalscheme(\signal) \cdot \left(1-\typeCDF\left( \frac{\pricingscheme(\signal) - \expect[\posterior(\cdot\mid \signal)]{\qualityfntwo(\quality)}}{\expect[\posterior(\cdot\mid \signal)]{\qualityfnone(\quality)}}\right)\right) \\
    & =
    \sum\nolimits_{\signal\in\signalSpace} \left(\inverseVal(\pricingscheme(\signal), \posterior) \cdot \expect[\posterior(\cdot\mid \signal)]{\qualityfnone(\quality)}  + \expect[\posterior(\cdot\mid \signal)]{\qualityfntwo(\quality)}\right) \cdot \signalscheme(\signal) \cdot (1-\typeCDF(\inverseVal(\pricingscheme(\signal), \posterior))) \\ 
    & = 
    \sum\nolimits_{\signal\in\signalSpace} \inverseVal(\pricingscheme(\signal), \posterior) \cdot \expect[\posterior(\cdot\mid \signal)]{\qualityfnone(\quality)} \cdot \signalscheme(\signal) \cdot (1-\typeCDF(\inverseVal(\pricingscheme(\signal), \posterior))) + \\
    & \quad \sum\nolimits_{\signal\in\signalSpace} \expect[\posterior(\cdot\mid \signal)]{\qualityfntwo(\quality)} \signalscheme(\signal) \cdot (1-\typeCDF(\inverseVal(\pricingscheme(\signal), \posterior)))  \\ 
    & \le 
    \sum\nolimits_{\signal\in\signalSpace} \signalscheme(\signal)\expect[\posterior(\cdot\mid \signal)]{\qualityfnone(\quality)}\cdot \Myer(\typeCDF) + \sum\nolimits_{\signal\in\signalSpace} \signalscheme(\signal)   \expect[\posterior(\cdot\mid \signal)]{\qualityfntwo(\quality)} \\
    & = 
    \sum\nolimits_{\signal\in\signalSpace} \signalscheme(\signal)\sum\nolimits_{\quality\in[0, 1]} \posterior(\quality\mid\signal)\qualityfnone(\quality)\cdot \Myer(\typeCDF) + \sum\nolimits_{\signal\in\signalSpace} \signalscheme(\signal)\sum\nolimits_{\quality\in[0, 1]} \posterior(\quality\mid\signal) \qualityfntwo(\quality)\\
    & = 
    \sum\nolimits_{\quality\in[0, 1]} \sum\nolimits_{\signal\in\signalSpace} \signalscheme(\signal)\posterior(\quality\mid\signal)\qualityfnone(\quality)\cdot \Myer(\typeCDF) + \sum\nolimits_{\quality\in[0, 1]} \sum\nolimits_{\signal\in\signalSpace} \signalscheme(\signal) \posterior(\quality\mid\signal) \qualityfntwo(\quality)\\
    & = 
    \expect[\prior]{\qualityfnone(\quality)}\cdot\Myer(\typeCDF) + \expect[\prior]{\qualityfntwo(\quality)}~.
\end{align*}

\xhdr{Lower bounding $\Rev(\noinfor)$}
Let $\noinforPrice \triangleq \argmax_{\price} \Rev(\noinfor, \price)$.
We know
\begin{align*}
    \Rev(\noinfor, \noinforPrice)
    & \ge 
    \Rev\left(\noinfor, \expect[\prior]{\qualityfntwo(\quality)}\right) \\
    & = \expect[\prior]{\qualityfntwo(\quality)} \cdot \prob[\type\sim \typeCDF]{\expect[\quality\sim\prior]{\buyerUtility(\type, \quality)} \ge \expect[\prior]{\qualityfntwo(\quality)}} \\
    & = \expect[\prior]{\qualityfntwo(\quality)} \cdot \prob[\type\sim \typeCDF]{\type\cdot\expect[\quality\sim\prior]{\qualityfnone(\quality)} + \expect[\quality\sim\prior]{\qualityfntwo(\quality)}  \ge \expect[\prior]{\qualityfntwo(\quality)}} \\
    & = \expect[\prior]{\qualityfntwo(\quality)}~;\\
    \Rev(\noinfor, \noinforPrice)
    & = \noinforPrice \cdot \left(1 - \typeCDF\left(\frac{\noinforPrice - \expect[\prior]{\qualityfntwo(\quality)}}{\expect[\prior]{\qualityfnone(\quality)}} \right)\right)\\
    & \ge 
    \price^\dagger \cdot \left(1 - \typeCDF\left(\frac{\price^\dagger - \expect[\prior]{\qualityfntwo(\quality)}}{\expect[\prior]{\qualityfnone(\quality)}} \right)\right)\\
    & = \left(\MyerPrice \cdot \expect[\prior]{\qualityfnone(\quality)} + \expect[\prior]{\qualityfntwo(\quality)}\right) \cdot \left(1 - \typeCDF\left(\MyerPrice\right)\right)\\
    & \ge 
    \Myer(\typeCDF) \cdot \expect[\prior]{\qualityfnone(\quality)}
\end{align*}
where $p^\dagger \triangleq \MyerPrice \cdot \expect[\prior]{\qualityfnone(\quality)} + \expect[\prior]{\qualityfntwo(\quality)}$ and $\MyerPrice = \argmax \price\cdot (1-\typeCDF(\price))$.

We thus complete the proof.
\end{proof}
We next argue that the $2$-approximation of $\noinfor$ is tight.
\begin{lemma}
For any $\varepsilon\in(0, 1)$,
there exists an instance with valuation function $\buyerUtility\in\buyerUtilityClass_{\cmrs}$, a quality distribution $\prior_\varepsilon$ and a type distribution $\typeCDF_\varepsilon$ that depends on $\varepsilon$ such that 
$\frac{2}{1+\varepsilon}\cdot \Rev(\noinfor) \le \OPT$
\end{lemma}
\begin{proof}
We use the valuation function $\buyerUtility (\type, \quality) = \type(1-\quality) + \quality$ constructed in \Cref{lem:stoploss-realize}.
Note that this valuation function $\buyerUtility\in\buyerUtilityClass_{\cmrs}$. 
Thus the hard instance constructed in the proof of \Cref{lem:quantile partition:matching lower bound} also holds here. 
Invoking \Cref{lem:quantile partition:matching lower bound} with $\signalnum = 1$ would finish the proof.
\end{proof}

\subsection{Proof of Proposition~\ref{prop:beyond 2 sc}}
We define the following increasing indicator functions:
\begin{align*}
    A_1(\quality) =\indicator{\quality\ge \tfrac{1}{3}},\quad
    A_2(\quality)=\indicator{\quality\ge \tfrac{2}{3}}, \quad
    B_1(\type) = \indicator{\type\ge \tfrac{1}{3}},\quad
    B_2(\type)= \indicator{\type\ge \tfrac{2}{3}}~.
\end{align*}
We next define the following valuation function 
\begin{align*}
    \buyerUtility(\type, \quality) 
    = 
    1-A_1(\quality)-B_2(\type)+A_1(\quality)B_1(\type)+A_2(\quality)B_2(\type)~.
\end{align*}
It is easy to verify that the above valuation function satisfies \Cref{assump:single-crossing}.
We next construct a quality distribution $\prior$ and a type distribution $\typeCDF$ as follows:
\begin{align*}
    \prob[\quality\sim\prior]{\quality = 0} = \frac{1}{60}, \quad
    \prob[\quality\sim\prior]{\quality = \frac{1}{2}} = \frac{53}{60}, \quad
    \prob[\quality\sim\prior]{\quality = 1} = \frac{1}{10}~,\\
    \prob[\type\sim\typeCDF]{\type = 0} = \frac{5}{6}, \quad
    \prob[\type\sim\typeCDF]{\type = \frac{1}{2}} = \frac{1}{60}, \quad
    \prob[\type\sim\typeCDF]{\type = 1} = \frac{3}{20}~.
\end{align*}
We can now compute that 
\begin{align*}
    \Rev(\noinfor) = \frac{1}{60}~, \quad
    \OPT \ge \frac{41}{900}~,
\end{align*}
which gives us $\frac{41}{15} \cdot \Rev(\noinfor) \approx 2.73 \cdot\Rev(\noinfor)\le  \OPT$.

\section{Failure of Sandwich-Based Functional Space \texorpdfstring{$\cvxspacesandwich$}{H'}}
\label{apx:failure sandwich}

In this section, we demonstrate that replacing the conditions in \Cref{lem:revenue function necessary condition} with another natural necessary condition satisfied by the optimal posted-price revenue function $\cvx_{\buyerUtility, \typeCDF}$ yields a strictly looser upper bound of the optimal robust competitive ratio.
In particular, in dosing so, it fails to recover the tight characterization established in \Cref{thm:rdd:optimal robust policy}.
For clarity and simplicity of exposition, we focus throughout this section on the additive valuation function $\buyerUtility(\type, \quality) = \type + \quality$.

\begin{lemma}[Sandwich-bound characterization of $\cvx_{\buyerUtility, \typeCDF}$]
\label{lem:sandwich}
    Fix additive valuation function $\buyerUtility(\type, \quality) = \type + \quality$. For any type distribution $\typeCDF\in\Delta([0, 1])$, the optimal posted-price revenue function $\cvx_{\buyerUtility, \typeCDF}$ is non-negative, non-decreasing, and convex. Moreover, for any induced posterior mean $\pooledmean \in [0, 1]$,
    \begin{align*}
        \pooledmean \vee \Myer(\typeCDF) 
        \le \cvx_{\buyerUtility, \typeCDF}(\pooledmean)\le \pooledmean + \Myer(\typeCDF)~,
    \end{align*}
    where $\Myer(\typeCDF) \triangleq \max_\price \price\cdot (1 - \typeCDF(\price))$.
\end{lemma}

\begin{figure}[H]
\centering

\begin{subfigure}[t]{0.4\textwidth}
\centering
\begin{tikzpicture}[x=5cm,y=5cm,>=Stealth]
\tikzset{reddash/.style={red!85, line width=1.6pt, opacity=0.8}}
\def\c{0.15}
\pgfmathsetmacro{\chalf}{\c/2}
\pgfmathsetmacro{\ctwo}{2*\c}
\pgfmathsetmacro{\onepc}{1+\c}

\fill[blue!15]
  (0,\c) -- (\c,\c) -- (1,1) -- (1,\onepc) -- cycle;
\draw[->,line width=1.2pt] (0,0) -- (1.10,0);
\draw[->,line width=1.2pt] (0,0) -- (0,1.2);

\draw[gray,dashed,line width=1.0pt] (0,1) -- (1,1);
\draw[gray,dashed,line width=1.0pt] (0,\c) -- (1,\c);
\draw[gray,dashed,line width=1.0pt] (\c,0) -- (\c,\c);
\draw[gray,dashed,line width=1.0pt] (1,0) -- (1,\onepc);

\draw[gray,dashed,line width=1.2pt] (0,0) -- (\c,\c);

\draw[blue,dashed,line width=2.0pt, dash pattern=on 8pt off 3pt] (0,\c) -- (\c,\c);
\draw[blue,dashed,line width=2.0pt, dash pattern=on 8pt off 3pt] (\c,\c) -- (1,1);
\draw[blue,dashed,line width=2.0pt, dash pattern=on 8pt off 3pt] (0,\c) -- (1,\onepc);
\draw[blue,dashed,line width=2.0pt, dash pattern=on 8pt off 3pt] (1,1) -- (1,\onepc);

\draw[reddash,line width=1.6pt] (0,\c) -- (\chalf,\c);
\draw[reddash,line width=1.6pt] (\chalf,\c) -- (\ctwo,\ctwo);
\draw[reddash,line width=1.6pt] (\ctwo,\ctwo) -- (1,\onepc);

\node[anchor=east]  at (0,1) {$1$};
\node[anchor=east]  at (0,\c) {$c$};
\node[anchor=north] at (\c,0) {$c$};
\node[anchor=north] at (1,0) {$1$};

\end{tikzpicture}
\caption{}
\label{fig:left sandwich}
\end{subfigure}\hspace{4em}
\begin{subfigure}[t]{0.4\textwidth}
\centering
\begin{tikzpicture}[x=5cm,y=5cm,>=Stealth]
\tikzset{reddash/.style={red!85, line width=1.6pt, opacity=0.8}}
\def\c{0.15}
\def\eps{0.06}
\pgfmathsetmacro{\xe}{1-\eps}
\pgfmathsetmacro{\xeloc}{1-2.5*\eps}
\pgfmathsetmacro{\onepc}{1+\c}
\fill[blue!15]
  (0,\c) -- (\c,\c) -- (1,1) -- (1,\onepc) -- cycle;
\draw[->,line width=1.2pt] (0,0) -- (1.10,0);
\draw[->,line width=1.2pt] (0,0) -- (0,1.2);

\draw[gray,dashed,line width=1.0pt] (0,1) -- (1,1);
\draw[gray,dashed,line width=1.0pt] (0,\c) -- (1,\c);
\draw[gray,dashed,line width=1.0pt] (\c,0) -- (\c,\c);
\draw[gray,dashed,line width=1.0pt] (\xe,0) -- (\xe,\xe);
\draw[gray,dashed,line width=1.0pt] (1,0) -- (1,\onepc);

\draw[gray,dashed,line width=1.2pt] (0,0) -- (\c,\c);

\draw[blue, dashed, line width=2.0pt, dash pattern=on 8pt off 3pt] (0,\c) -- (\c,\c);
\draw[blue,dashed, line width=2.0pt, dash pattern=on 8pt off 3pt] (\c,\c) -- (1,1);
\draw[blue,dashed, line width=2.0pt, dash pattern=on 8pt off 3pt] (0,\c) -- (1,\onepc);
\draw[blue,dashed, line width=2.0pt, dash pattern=on 8pt off 3pt] (1,1) -- (1,\onepc);

\draw[reddash,line width=1.6pt] (0,\c) -- (\c,\c);
\draw[reddash,line width=1.6pt] (\c,\c) -- (\xe,\xe);
\draw[reddash,line width=1.6pt] (\xe,\xe) -- (1,\onepc);

\node[anchor=east]  at (0,1) {$1$};
\node[anchor=east]  at (0,\c) {$c$};
\node[anchor=north] at (\c,0) {$c$};
\node[anchor=north] at (\xeloc,0) {$1-\delta \varepsilon$};
\node[anchor=north] at (1,0) {$1$};

\end{tikzpicture}
\caption{}
\label{fig:right sandwich}
\end{subfigure}
\caption{
The blue dashed lines are upper boundary (i.e., $\cvx(\pooledmean) = \pooledmean + \kink$) and lower boundary (i.e., $\cvx(\pooledmean) = \max\{\pooledmean, \kink\}$) of the functional space $\cvxspacesandwich$ when we fix a particular value of $\kink$.
The red solid line in \Cref{fig:left sandwich} is a feasible convex function $\cvx \in\cvxspacesandwich$ in this functional space, while the red solid line in \Cref{fig:right sandwich} is the constructed convex function $\cvx$ in this functional space such that it leads to worst-case competitive ratio.}
\label{fig:sandwich}
\end{figure}
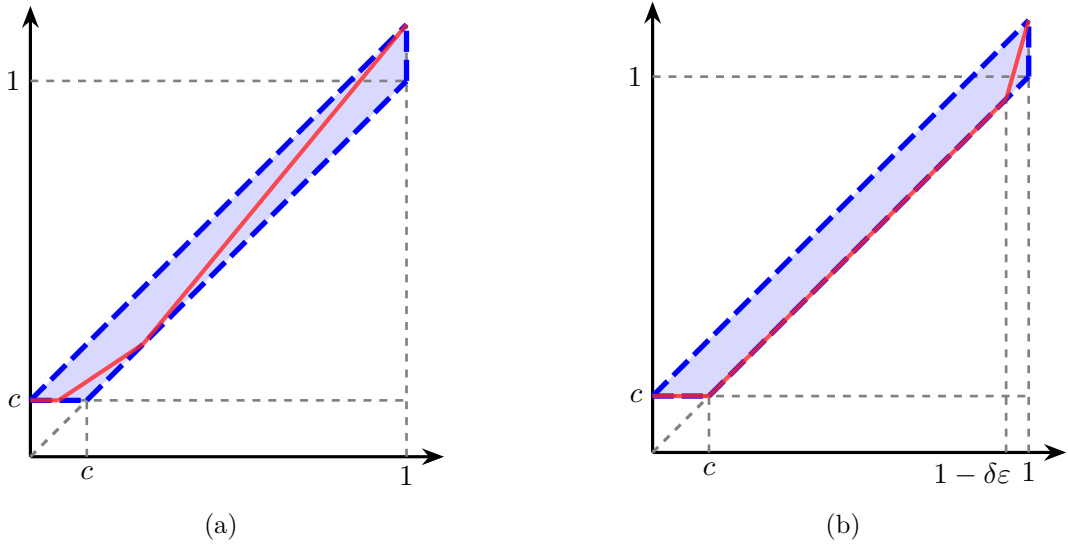
 
\begin{proof}
Fix additive valuation function $\buyerUtility(\type, \quality) = \type + \quality$.
The non-negativity, monotonicity and convexity of $\cvx_{\buyerUtility, \typeCDF}$ follows directly from \Cref{lem:revenue function necessary condition}. 

Fix any $\pooledmean\in[0,1]$. Under the additive valuation, we have
\begin{align*}
    \cvx_{\buyerUtility, \typeCDF}(\pooledmean) =\sup_{\type\in[0,1]} (1-\typeCDF(\type))\,(\pooledmean + \type).
\end{align*}

\paragraph{Lower bound by $\pooledmean\vee \Myer(\typeCDF)$.}
Choose $\type=0$. 
Since $\typeCDF(0)= \prob{\theta<0} = 0$, we have $1-\typeCDF(0)=1$, and therefore $\cvx_{\buyerUtility, \typeCDF}(\pooledmean) \ge (1-\typeCDF(0))\,(\pooledmean+0) = \pooledmean$.
For any $\type\in[0,1]$, since $\pooledmean \ge 0$ and $1-\typeCDF(\type)\ge 0$, we have
\begin{align*}
    (1-\typeCDF(\type))\,(\pooledmean + \type)\;\ge\;(1-\typeCDF(\type))\,\type~.
\end{align*}
Taking the supremum over $\type\in[0,1]$ yields
\begin{align*}
    \cvx_{\buyerUtility, \typeCDF}(\pooledmean)
    =\sup_{\type\in[0,1]} (1-\typeCDF(\type))\,(\pooledmean + \type)
    \;\ge\;
    \sup_{\type\in[0,1]} (1-\typeCDF(\type))\,\type
    =\Myer(\typeCDF)~.
\end{align*}
Combining the two inequalities implies $\cvx_{\buyerUtility, \typeCDF}(\pooledmean)\;\ge\; \pooledmean \vee \Myer(\typeCDF)$

\xhdr{Upper bound by $\pooledmean+\Myer(\typeCDF)$}
Fix any $\pooledmean\in[0,1]$ and any $\type\in[0,1]$. We can write
\begin{align*}
    (1-\typeCDF(\type))\,(\pooledmean + \type)
    = \pooledmean(1-\typeCDF(\type)) + \type (1-\typeCDF(\type))
    \le \pooledmean + \type (1-\typeCDF(\type))~. \tag{Due to $1-\typeCDF(\type)\le 1$ and $\pooledmean \ge 0$}
\end{align*}
Taking the supremum over $\type\in[0,1]$ gives us
\begin{align*}
    \cvx_{\buyerUtility, \typeCDF}(\pooledmean)
    =\sup_{\type\in[0,1]} (1-\typeCDF(\type))\,(\pooledmean + \type)
    \;\le\; \pooledmean+\Myer(\typeCDF)~.
\end{align*}
We thus finish the proof.
\end{proof}

Comparing the necessary conditions in \Cref{lem:sandwich} and \Cref{lem:revenue function necessary condition}, the sandwich-bound characterization in \Cref{lem:sandwich} is arguably more intuitive for additive valuation functions. Moreover, it can be verified that the condition in \Cref{lem:sandwich} is strictly weaker than the one in \Cref{lem:revenue function necessary condition}. See \Cref{fig:sandwich} for an illustration.

Similar to the reduction from $\RQDQuant$ to $\RDDQuant$ in \Cref{sec:rdd:proof of rqd quantile policy}, we consider the following robust disclosure design program~$\RDDQuantsand$, where $\cvxspacesandwich$ denotes the space of all non-negative, non-decreasing, convex functions $\cvx : [0,1] \to \reals_+$ satisfying
\begin{align*}
\max\{\pooledmean, \kink\} \leq \cvx(\pooledmean) \leq \pooledmean + \kink, \quad \forall\, \pooledmean \in [0,1],
\end{align*}
for all values $\kink \in [0,1]$. It can be verified that the posted-price revenue function $\cvx_{\buyerUtility, \typeCDF}$ induced by additive valuation function $\buyerUtility$ belongs to $\cvxspacesandwich$ (\Cref{lem:sandwich}). Consequently, the robust disclosure design program~$\RDDQuantsand$ serves as a valid relaxation of the original robust quality disclosure problem~$\RQDQuant$, as formalized below.

\begin{proposition}
\label{prop:sandwich}
    Fix any integer $\signalnum \in \naturals$. Let $\quantDisClass$ denote the class of all {\KQuantilePartitionDisclosure} policies, and let $\cvxspacesandwich$ be the space of all non-negative, non-decreasing, convex functions $\max\{\pooledmean, \kink\} \leq \cvx(\pooledmean) \leq \pooledmean + \kink$ for all $\pooledmean \in [0,1]$. Then:
    \begin{enumerate}
        \item[(i)] The optimal robust competitive ratio of program~$\RDDQuantsand$ is at least the optimal robust competitive ratio of program~$\RQDQuant$.
        
        \item[(ii)] For every {\KQuantilePartitionDisclosure} policy $\disclosurefn$ with quantile thresholds $\quants = (\quant_0, \dots, \quant_\signalnum)$, its robust competitive ratio in program~$\RDDQuantsand$ is at least
        \begin{align*}
            \voa(\disclosurefn \mid \cvxspacesandwich)
            \geq
            \max_{\quantileidx \in [\signalnum]}
\sup_{\kink\in[0, 1]} 
            \frac{\kink\cdot \quant_{\quantileidx-1} 
            + (\kink(1-\kink) + \kink) (\quant_{\quantileidx} - \quant_{\quantileidx-1})
+ (1+\kink)(1-\quant_{\quantileidx})}
            {\kink\cdot \quant_{\quantileidx} 
            + (\quant_{\quantileidx+1} - \quant_{\quantileidx})
+ (1+\kink)(1-\quant_{\quantileidx+1})}~.
        \end{align*}
    \end{enumerate}
\end{proposition}

\begin{table}[H]
\centering
\renewcommand{\arraystretch}{1.35}
\setlength{\tabcolsep}{10pt}
\begin{tabular}{ccccccc}
\toprule
 & ${\quant_1}$ & ${\quant_2}$ & ${\quant_3}$ & ${\quant_4}$ & ${\quant_5}$ &  $\voa_\signalnum$ \\
\midrule
${K=1}$ & $1$      & --     & --     & --     & --          & $2$      \\
${K=2}$ & $0.667$ & $1$    & --     & --     & --          & $1.500$ \\
${K=3}$ & $0.664$ & $0.734$ & $1$  & --     & --         & $1.2658$ \\
${K=4}$ & $0.615$ & $0.660$ & $0.840$ & $1$ & --          & $1.2167$ \\
${K=5}$ & $0.560$ & $0.560$ & $0.760$ & $0.880$ & $1$   & $1.1628$  \\
\bottomrule
\end{tabular}
\caption{The numerical results for solving Part (ii) of \Cref{prop:sandwich}.}
\label{tab:robust-quantile-thresholds sandwitch}
\end{table}

Part (i) of \Cref{prop:sandwich} follows directly from \Cref{lem:sandwich}, since the function space $\cvxspacesandwich$ contains all posted-price revenue functions (as established in \Cref{lem:sandwich}), and thus the program $\RDDQuantsand$ is a relaxation of $\RDDQuant$. 
For Part (ii), we explicitly construct an indirect utility function $\cvx \in \cvxspacesandwich$ and a prior $\prior \in \Delta([0,1])$, and compute the corresponding objective value. (In fact, by an argument analogous to that in \Cref{sec:rdd:proof of rdd optimal robust policy}, one can show that the lower bound in Part (ii) is indeed the exact value of the robust competitive ratio $\voa(\disclosurefn \mid \cvxspacesandwich)$.)

We conduct a numerical evaluation of $\RDDQuantsand$ for $\signalnum = 1, 2, \dots, 5$, using a grid search over quantile thresholds followed by evaluation via the formula in Part (ii) of \Cref{prop:sandwich}. The resulting optimal quantile threshold profiles and robust competitive ratios are reported in \Cref{tab:robust-quantile-thresholds sandwitch}. 
It can be observed that, except for the case $\signalnum = 1$---where we recover the tight optimal robust competitive ratio of $2$, matching that of the original program $\RQDQuant$---for $\signalnum = 2, 3, 4, 5$, the optimal robust competitive ratio under $\RDDQuantsand$ is strictly worse (i.e., larger) than the optimal ratio $\voa_{\signalnum}^*$ characterized in \Cref{thm:opt quantile partition:optimal policy}. Moreover, the quantile threshold profiles identified under $\RDDQuantsand$ differ from those that are optimal for $\RQDQuant$. 
This result indicates that, unlike the program $\RDDQuant$ analyzed in \Cref{sec:rdd:proof of rqd quantile policy}, the conditions in \Cref{lem:sandwich} and the associated relaxation $\RDDQuantsand$ are insufficient to recover the tight bound for the original problem $\RQDQuant$.

\begin{proof}[Proof of \Cref{prop:sandwich}]
We prove each part separately.

\xhdr{Part (i)}
By \Cref{lem:sandwich}, the indirect utility function $\cvx_{\buyerUtility, \typeCDF}$ induced by any linear valuation function $\buyerUtility$ and consumer type distribution $\typeCDF$ belongs to the function class $\cvxspacesandwich$. Consequently, the robust quality disclosure problem~$\RQDQuant$ is a restriction of the robust disclosure design problem~$\RDDQuantsand$, and thus
\begin{align*}
\sup_{\substack{\buyerUtility \in \buyerUtilityClass \\ \prior, \typeCDF \in \Delta([0,1])}}
\frac{\OPT_{\buyerUtility,\prior,\typeCDF}}{\Rev_{\buyerUtility,\prior,\typeCDF}(\disclosurefn(\prior))}
\leq
\voa(\disclosurefn \mid \cvxspacesandwich)
\end{align*}
for every {\KQuantilePartitionDisclosure} $\disclosurefn$.
Taking the $\inf_{\disclosurefn}$ over all {\KQuantilePartitionDisclosure} disclosure policies on both sides would give us the desired result.

\xhdr{Part (ii)} We prove this part by construction.
In particular, we show that for any given {\KQuantilePartitionDisclosure} $\disclosurefn$, for each $\quantileidx\in[\quantileidx]$,
there exists a pair of quality distribution $\prior$ and a convex function $\cvx\in\cvxspacesandwich$, such that 
\begin{align}
    \frac{\expect[\pooledmean\sim \prior]{\cvx(\pooledmean)}}{\expect[\pooledmean\sim \disclosurefn(\prior)]{\cvx(\pooledmean)}}
    = 
\sup_{\kink\in[0, 1]} 
    \frac{\kink\cdot \quant_{\quantileidx-1} 
    + (\kink(1-\kink) + \kink) (\quant_{\quantileidx} - \quant_{\quantileidx-1})
+ (1+\kink)(1-\quant_{\quantileidx})}
    {\kink\cdot \quant_{\quantileidx} 
    + (\quant_{\quantileidx+1} - \quant_{\quantileidx})
+ (1+\kink)(1-\quant_{\quantileidx+1})}~.
    \label{sandwich: revenue ratio}
\end{align}
To see this, let us fix an index $\quantileidx\in[\signalnum]$ and fix $\varepsilon, \delta\in(0, 1)$,
we construct quality distribution $\prior$ such that its $\supp(\prior) = \{0, 1-\varepsilon, 1\}$ such that 
\begin{align*}
    \prob[\prior]{\quality = 0}
    & = \quant_{\quantileidx-1} + (1-\kink) \cdot (\quant_{\quantileidx} - \quant_{\quantileidx-1})~,\\
    \prob[\prior]{\quality = 1-\varepsilon}
    & = \kink\cdot (\quant_{\quantileidx} - \quant_{\quantileidx-1}) + \delta\cdot (\quant_{\quantileidx+1} - \quant_{\quantileidx})\\
    \prob[\prior]{\quality = 1}
    & = 1 - \quant_{\quantileidx} - \delta\cdot (\quant_{\quantileidx+1} - \quant_{\quantileidx})~.
\end{align*}
Define the convex function $\cvx$ as follows:
\begin{align*}
    \cvx(\pooledmean)=
    \begin{cases}
    \kink, & \pooledmean \le \kink,\\[3pt]
    \pooledmean, & \pooledmean \in [\kink,\, 1-\delta\varepsilon],\\[3pt]
    \dfrac{\kink+\delta\varepsilon}{\delta\varepsilon}\left(\pooledmean-(1-\delta\varepsilon)\right) + 1-\delta\varepsilon,
    & \kink \in [\,1-\delta\varepsilon,\, 1\,].
    \end{cases}
\end{align*}
Under such pair of $(\prior, \cvx)$, taking $\varepsilon, \delta \rightarrow 0$ and taking the $\sup_{\kink\in[0, 1]}$, one can get the ratio in Eqn.~\eqref{sandwich: revenue ratio}.
Then taking $\max_{\quantileidx\in[\signalnum]}$ on RHS of Eqn.~\eqref{sandwich: revenue ratio} would give us the desired result in Part (ii). 
\end{proof}

\section{Robust Quality Partition Disclosure Policy}
\label{apx:quality partition}

In this section, we consider another class of simple disclosure policies that is closely related to {\KQuantilePartitionDisclosure}. Instead of partitioning based on quantile thresholds, this class uses quality thresholds.

A {\KQualityPartitionDisclosure} is parameterized by $(\signalnum+1)$ weakly increasing quality thresholds $0 = \quality_0 \leq \quality_1 \leq \dots \leq \quality_\signalnum = 1$ and $(\signalnum+1)$ partitional probabilities $\xi_0=0, \xi_1, \dots, \xi_{\signalnum-1} \in [0,1], \xi_\signalnum = 1$. Given a quality distribution $\prior \in \Delta([0,1])$, it outputs a monotone-partition signaling scheme that, for each $\quantileidx \in [\signalnum]$, pools into a single signal indexed by $\quantileidx$ with all qualities in the open interval $(\quality_{\quantileidx-1}, \quality_{\quantileidx})$,
a fraction $(1 - \xi_{\quantileidx-1})$ of the mass at quality $\quality_{\quantileidx-1}$, and
a fraction $\xi_{\quantileidx}$ of the mass at quality $\quality_{\quantileidx}$.
Consequently, the induced signaling scheme $\signalscheme$ has signal space $\signalSpace = [\signalnum]$. See \Cref{fig:quality} for an illustration.

\begin{figure}[H]
\centering

\tikzset{
  dashedCut/.style={dash pattern=on 2.2pt off 2.2pt, line width=0.6pt, black, opacity=0.8},
  axisLine/.style={line width=0.7pt, black},
  tick/.style={line width=0.6pt, black},
  arr/.style={
    -{Stealth[length=1.5mm,width=1.5mm]},
    line width=0.8pt,
    black,
    opacity=0.9
  }
}

\begin{minipage}[t]{0.49\textwidth}
\centering
\begin{tikzpicture}[font=\small]
\def\H{5.0}
\def\barW{0.65}
\def\xAxisStart{0.55}
\def\xPrior{2}
\def\xSig{4}
\def\xAxisEnd{5.10}

\def\qA{0.33333}\def\qB{0.66667}
\newcommand{\yq}[1]{#1*\H}

\colorlet{BlueLo}{blue!6}
\colorlet{BlueMid}{blue!45}
\colorlet{BlueHi}{blue!92}
\colorlet{SigD1}{BlueLo}
\colorlet{SigD2}{blue!30}
\colorlet{SigD3}{blue!65}
\colorlet{SigD4}{BlueHi}

\draw[axisLine] (\xAxisStart,0) -- (\xAxisEnd,0);
\draw[axisLine] (\xAxisStart,0) -- (\xAxisStart,\H);

\foreach \qq/\lab in {0/0,\qA/{$\tfrac13$},\qB/{$\tfrac23$},1/1}{
  \draw[tick] (\xAxisStart,\yq{\qq}) -- ++(-0.07,0);
  \node[anchor=east] at (\xAxisStart-0.09,\yq{\qq}) {\lab};
}

\draw[tick] (\xPrior,0) -- ++(0,-0.07);
\draw[tick] (\xSig,0) -- ++(0,-0.07);
\node[anchor=north] at (\xPrior,-0.12) {prior};
\node[anchor=north] at (\xSig,-0.12) {signals};

\fill[BlueLo]  (\xPrior-\barW/2, 0)      rectangle (\xPrior+\barW/2, \H/3);
\fill[BlueMid] (\xPrior-\barW/2, \H/3)   rectangle (\xPrior+\barW/2, 2*\H/3);
\fill[BlueHi]  (\xPrior-\barW/2, 2*\H/3) rectangle (\xPrior+\barW/2, \H);

\fill[BlueLo]  (\xSig-\barW/2, 0)      rectangle (\xSig+\barW/2, \H/3);
\fill[BlueMid] (\xSig-\barW/2, \H/3)   rectangle (\xSig+\barW/2, 2*\H/3);
\fill[BlueHi]  (\xSig-\barW/2, 2*\H/3) rectangle (\xSig+\barW/2, \H);

\def\gap{0.03}
\def\xL{\xPrior+\barW/2+\gap}
\def\xR{\xSig-\barW/2-\gap}
\def\mOne{1/6}\def\mTwo{1/2}\def\mThr{5/6}

\draw[arr] (\xL, \yq{\mOne})  -- (\xR, \yq{\mOne});
\draw[arr] (\xL, \yq{\mTwo}) -- (\xR, \yq{\mTwo});
\draw[arr] (\xL, \yq{\mThr}) -- (\xR, \yq{\mThr});

\draw[dashedCut] (\xAxisStart,\yq{\qA}) -- (\xAxisEnd,\yq{\qA});
\draw[dashedCut] (\xAxisStart,\yq{\qB}) -- (\xAxisEnd,\yq{\qB});
\end{tikzpicture}
\end{minipage}\hfill
\begin{minipage}[t]{0.49\textwidth}
\centering
\begin{tikzpicture}[font=\small]
\def\H{5.0}
\def\barW{0.65}
\def\xAxisStart{0.55}
\def\xPrior{2}
\def\xSig{4}
\def\xAxisEnd{5.10}
\def\qA{0.25}\def\qB{0.50}\def\qC{0.75}
\newcommand{\yq}[1]{#1*\H}

\colorlet{SigC1}{blue!15}
\colorlet{SigC2}{blue!35}
\colorlet{SigC3}{blue!60}
\colorlet{SigC4}{blue!85}

\draw[axisLine] (\xAxisStart,0) -- (\xAxisEnd,0);
\draw[axisLine] (\xAxisStart,0) -- (\xAxisStart,\H);

\foreach \qq/\lab in {0/0,\qA/0.25,\qB/0.5,\qC/0.75,1/1}{
  \draw[tick] (\xAxisStart,\yq{\qq}) -- ++(-0.07,0);
  \node[anchor=east] at (\xAxisStart-0.09,\yq{\qq}) {\lab};
}

\draw[tick] (\xPrior,0) -- ++(0,-0.07);
\draw[tick] (\xSig,0) -- ++(0,-0.07);
\node[anchor=north] at (\xPrior,-0.12) {prior};
\node[anchor=north] at (\xSig,-0.12) {signals};

\shade[bottom color=blue!6, top color=blue!92]
  (\xPrior-\barW/2,0) rectangle (\xPrior+\barW/2,\H);

\fill[SigC1] (\xSig-\barW/2, 0)        rectangle (\xSig+\barW/2, \yq{\qA});
\fill[SigC2] (\xSig-\barW/2, \yq{\qA}) rectangle (\xSig+\barW/2, \yq{\qB});
\fill[SigC3] (\xSig-\barW/2, \yq{\qB}) rectangle (\xSig+\barW/2, \yq{\qC});
\fill[SigC4] (\xSig-\barW/2, \yq{\qC}) rectangle (\xSig+\barW/2, \H);

\def\gap{0.03}
\def\xL{\xPrior+\barW/2+\gap}
\def\xR{\xSig-\barW/2-\gap}
\foreach \m in {0.125,0.375,0.625,0.875}{
  \draw[arr] (\xL, \yq{\m}) -- (\xR, \yq{\m});
}

\draw[dashedCut] (\xAxisStart,\yq{\qA}) -- (\xAxisEnd,\yq{\qA});
\draw[dashedCut] (\xAxisStart,\yq{\qB}) -- (\xAxisEnd,\yq{\qB});
\draw[dashedCut] (\xAxisStart,\yq{\qC}) -- (\xAxisEnd,\yq{\qC});
\end{tikzpicture}
\end{minipage}

\caption{An illustration of {\KQualityPartitionDisclosure} disclosure. Here $\signalnum = 4$ and $\quality_0 = 0, \quality_1 = 0.25, \quality_2 = 0.5, \quality_3 = 0.75, \quality_4 = 1, \xi_0 = 0, \xi_\quantileidx = 1$ for all $\quantileidx\in[4]$. The $y$-axis is cumulative probability mass. 
Left: uniform discrete prior with $\supp(\prior)=\{0,0.5,1\}$. Right: uniform continuous prior with $\supp(\prior)=[0,1]$.
With the same {\KQualityPartitionDisclosure}, the induced posterior-mean distributions can differ across priors. In particular, for the uniform discrete prior on the left, one induced signal is degenerate because there is no prior probability mass between the quality thresholds $\quality_2 = 0.5$ and $\quality_3 = 0.75$.}
\label{fig:quality}
\end{figure}
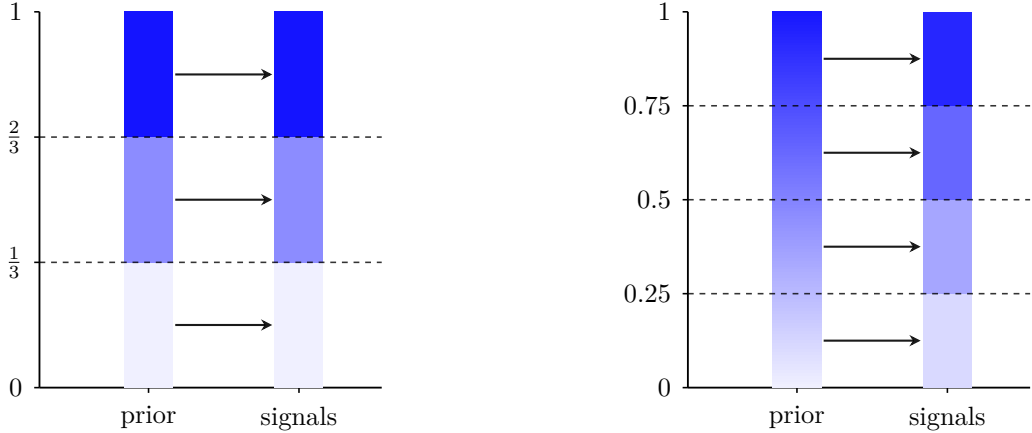

\begin{theorem}[Robust optimal quality partition]
\label{thm:opt quality partition}
Fix any integer $\signalnum \in \naturals$. Let $\disclosurefnClass_\signalnum\primed$ denote the class of all {\KQualityPartitionDisclosure} policies. Then the robust quality disclosure problem~$\RQDQuality$ satisfies the following: the optimal robust competitive ratio is $2$, and this value is achieved by every {\KQualityPartitionDisclosure} in $\disclosurefnClass_\signalnum\primed$.
\end{theorem}

In contrast to \Cref{thm:opt quantile partition:optimal policy} for {\KQuantilePartitionDisclosure}, \Cref{thm:opt quality partition} reveals that the robust quality disclosure problem under {\KQualityPartitionDisclosure} is degenerate: all such policies are equally effective (or ineffective), yielding a robust competitive ratio of exactly $2$ regardless of the number of partitions~$\signalnum$. 

By comparison, the optimal robust competitive ratio under {\KQuantilePartitionDisclosure} improves with $\signalnum$, converging to $1$ at a rate of $\Theta(1/\signalnum)$ (\Cref{prop:optimal robust CR monotonicity}). This stark contrast highlights that the {\KQuantilePartitionDisclosure} policies studied in the main text are not only theoretically richer but also practically more efficient for robust quality disclosure.

\begin{proof}[Proof of \Cref{thm:opt quality partition}]
Invoking \Cref{cor:full infor optimal no infor 2 approx}, since every monotone partition signaling scheme is 2-approximation under any consumer valuation function, type distribution and product quality distribution, the robust competitive ratio of any {\KQualityPartitionDisclosure} $\disclosurefn$ is at most 2. It remains to show a matching lower bound.

Let us fix any $\signalnum \in \naturals$, and an pair of $(\qualities, \partitionprobs)$, and consider a {\KQualityPartitionDisclosure} $\disclosurefn$ that is parameterized by $\qualities, \partitionprobs$. 
Fix any $\varepsilon \in (0, 1)$, and define the index 
\begin{align*}
    j \triangleq \min\{\quantileidx \in [\signalnum]: \quality_\quantileidx > 0\}
\end{align*}
Such $j$ exists because $\quality_\signalnum = 1 > 0$, so the open interval $(0, \quality_j)$ is nonempty.

We consider a quality distribution $\prior_\varepsilon$ supported on $\{\varepsilon^2 \cdot \frac{\quality_j}{2}, \frac{\quality_j}{2}\}$ with the mass 
$\frac{1}{1 + \varepsilon}, \frac{\varepsilon}{1 + \varepsilon}$, respectively.
By definition, the mean of this quality distribution is given by 
\begin{align*}
    \priormean 
    = \frac{1}{1 + \varepsilon} \varepsilon^2 \cdot \frac{\quality_j}{2} + \frac{\varepsilon}{1 + \varepsilon} \cdot \frac{\quality_j}{2} 
    = \varepsilon \cdot \frac{\quality_j}{2}~.
\end{align*}
Note that by construction, we have $\supp(\prior_\varepsilon) \subset (0, \quality_j)$.
Thus, the {\KQualityPartitionDisclosure} $\disclosurefn$ parameterized with $\qualities, \partitionprobs$ with the input $\prior_\varepsilon$ essentially outputs a no-information signaling scheme. 
Namely, we have $\disclosurefn(\prior_\varepsilon) = \delta_{(\priormean)}$.
Now consider a valuation distribution $\buyerUtility$ and the type distribution $\typeCDF_\varepsilon$ defined in \Cref{lem:stoploss-realize} with $\kink\gets \priormean$.
Then the induced optimal posted-price revenue function $\cvx_{\buyerUtility, \typeCDF_\varepsilon}(\pooledmean) = \max\{\pooledmean, \priormean\}$.
Thus, we have 
\begin{align*}
    \frac{\OPT_{\buyerUtility, \prior_\varepsilon, \typeCDF_\varepsilon}}{\Rev_{\buyerUtility, \prior_\varepsilon, \typeCDF_\varepsilon}(\disclosurefn(\prior_\varepsilon))} 
    = 
    \frac{\frac{1}{1+\varepsilon} \cdot \varepsilon\cdot \frac{\quality_j}{2} + \frac{\varepsilon}{1+\varepsilon}\cdot \frac{\quality_j}{2}}{\varepsilon\cdot \frac{\quality_j}{2}}
    = \frac{2}{1+\varepsilon}~.
\end{align*}
We thus finish the proof.
\end{proof}

\section{Missing Proofs}

\begin{lemma}[Monotonicity of $\recfn_\signalnum(\voa)$]
\label{lem:Tk-monotone}
Fix any $\signalnum\in\naturals$. Function $\recfn_\signalnum(\voa)$ defined in \Cref{thm:opt quantile partition:optimal policy} is continuous and strictly increasing in $\voa\in(1,\infty)$. Moreover,
\begin{align*}
    \recfn_\signalnum(1) = 0~,
    \quad\text{and}\quad
    \lim\limits_{\voa \to\infty}\recfn_\signalnum(\voa)=\infty~.
\end{align*}
\end{lemma}
\begin{proof}
Define auxiliary function $\varphi(z) \triangleq (1+\sqrt{z})^2$.
For each fixed $z\in[0,1]$, the mapping $\voa\mapsto \basefn_\voa(z)=z+(\voa - 1)\varphi(z)$ is continuous and strictly increasing on $(1,\infty)$ because $\varphi(z)\ge 1>0$. 
Thus, for each fixed $\voa$, the $\signalnum$-fold composition $\basefn_\voa^{\circ \signalnum}(0)$ is also continuous in $\voa$, and the strict monotonicity is preserved under composition: 
if $\voa_2>\voa_1>1$, then $\basefn_{\voa_2}(z)>\basefn_{\voa_1}(z)$ for all $z$, hence $\basefn_{\voa_2}^{\circ \signalnum}(0)>\basefn_{\voa_1}^{\circ \signalnum}(0)$, i.e., function $\voa \mapsto \recfn_\signalnum(\voa)$ is strictly increasing.
At $\voa = 1$, we have that $\basefn_1(z)=z$, so we have $\recfn_\signalnum(1)=0$.
For $\voa > 1$ and $z\ge 0$, we have that $\varphi(z)\ge 1$, hence
    \begin{align*}
        \basefn_\voa(z)=z+(\voa - 1)\varphi(z)\ge z+(\voa - 1)~.
    \end{align*}
    Iterating from $0$ yields $\recfn_\signalnum(\voa)=\basefn_\voa^{\circ \signalnum}(0)\ge \signalnum\cdot (\voa - 1)$, which approaches to $\infty$ as $\voa\to\infty$.
This completes the proof as desired.
\end{proof}

\begin{lemma}[Feasibility equivalence]
\label{lem:feasibility-equivalence}
Fix any $\voa > 1$. There exists a {\KQuantilePartitionDisclosure} $\disclosurefn$ parameterized by $\quants$ such that $\voa(\disclosurefn \mid \cvxspace)\le \voa$ if and only if $\recfn_\signalnum(\voa)\ge 1$.
\end{lemma}
\begin{proof}
Fix any $\voa > 1$. 
Suppose there exists {\KQuantilePartitionDisclosure} $\disclosurefn$ parameterized by $\quants$ such that $\voa(\disclosurefn \mid \cvxspace)\le \voa$. 
Define auxiliary function $\varphi(z) \triangleq (1+\sqrt{z})^2$.
From Part (ii) of \Cref{thm:rdd:optimal robust policy}, we know that $\voa(\disclosurefn \mid \cvxspace)\le \voa$ can be rewritten as
\begin{align*}
    \voa(\disclosurefn \mid \cvxspace)
    = 
    \max\nolimits_{ \quantileidx\in[\signalnum]}\left(1+\frac{\quant_\quantileidx-\quant_{\quantileidx-1}}{(1+\sqrt{1-\quant_\quantileidx})^2}\right) 
    & \le \voa \nonumber
\end{align*}
which is equivalent to 
\begin{align}
    1+\frac{\quant_\quantileidx-\quant_{\quantileidx-1}}{\varphi(1-\quant_\quantileidx)}
    & \le \voa \quad\text{for all } \quantileidx\in[\signalnum]~,
    \label{ineq:equivalence}
\end{align}

\xhdr{The \emph{$\Rightarrow$} direction}
We first prove the \emph{$\Rightarrow$} direction. 
From Eqn.~\eqref{ineq:equivalence}, we know
\begin{align*}
    1-\quant_{\quantileidx-1}\le 1-\quant_\quantileidx+(\voa - 1)\varphi(1-\quant_\quantileidx)=\basefn_\voa(1-\quant_\quantileidx), \quad \quantileidx\in[\signalnum]~.
\end{align*}
Note that function $\basefn_\voa (z) = z + (\voa - 1) \cdot (1+\sqrt{z})^2 = z + (\voa - 1)\cdot \varphi(z)$ is increasing,
iterating yields
\begin{align*}
    1 - \quant_0
    \le
    \basefn_\voa(1 - \quant_1)
    \le
    \basefn_\voa(\basefn_\voa(1-\quant_2))
    \le
    \basefn_\voa(\basefn_\voa(\basefn_\voa(1-\quant_3)))
    \le \cdots \le \recfn_\signalnum(\voa)~.
\end{align*}
and thus $\recfn_\signalnum(\voa) \geq 1$ as desired.

\xhdr{The \emph{$\Leftarrow$} direction}
Suppose $\recfn_\signalnum(\voa)\ge 1$. 
Define a sequence $(s_\quantileidx)_{\quantileidx\in[0:\signalnum]}$ by the tight recursion
\begin{equation}\label{eq:s-recursion}
    s_\signalnum\triangleq 0~,
    \quad 
    s_{\quantileidx-1}\triangleq \basefn_\voa(s_\quantileidx)=s_\quantileidx+(\voa - 1)\varphi(s_\quantileidx)~,
    \quad \quantileidx=\signalnum,\signalnum-1,\dots,1~.
\end{equation}
Then by definition, we have $s_0=\recfn_\signalnum(\voa)\ge 1$. 
Using the sequence $(s_\quantileidx)_{\quantileidx\in[0:\signalnum]}$, we next construct a valid quantile threshold profile $\quants = (\quant_0,\quant_1,\dots,\quant_\signalnum)$.
Let
\begin{equation}
    \label{eq:scaled-tr}
    \quant_\quantileidx
    \triangleq 1 -  
     \frac{s_\quantileidx}{s_0},\qquad \quantileidx\in[0:\signalnum]~.
\end{equation}
Then $\quant_0=0$ and $\quant_\signalnum=1-s_\signalnum/s_0=1$. 
Also we have $s_{\quantileidx-1}\ge s_\quantileidx$ because $\voa > 1$ implies $\basefn_\voa(z)>z$, 
hence we have $\quant_{\quantileidx-1}\le \quant_\quantileidx$.
Therefore, $\quants$ is a valid quantile threshold profile. 

Consider {\KQuantilePartitionDisclosure} parameterized by quantile threshold profile $\quants$ constructed above.
Now it remains to verify $\voa(\disclosurefn \mid \cvxspace)\le \voa$. By construction, for each $\quantileidx\in[\signalnum]$,
\begin{align*}
    s_{\quantileidx-1}-s_\quantileidx
     =(\voa - 1)\varphi(s_\quantileidx) 
\end{align*}
which is equivalent to
\begin{align*}
    \quant_{\quantileidx}-\quant_{\quantileidx-1}
    &=
    \frac{s_{\quantileidx-1}-s_\quantileidx}{s_0}
     =
     \frac{(\voa - 1)\varphi(s_\quantileidx)}{s_0}
     \le (\voa - 1)\varphi\left(\frac{s_\quantileidx}{s_0}\right) 
     = (\voa - 1)\varphi(1-\quant_\quantileidx)~,
\end{align*}
where we have used the observation that the function $\varphi(z)=(1+\sqrt{z})^2$ satisfies 
for every $\alpha\in[0,1]$ and every $z\in[0,1]$, we have $\varphi(\alpha\cdot z)\ \ge\ \alpha\cdot\varphi(z)$.

Thus all the inequalities in Eqn.~\eqref{ineq:equivalence} holds, and thus we have $\voa(\disclosurefn \mid \cvxspace)\le \voa$. This completes the proof of \Cref{lem:feasibility-equivalence} as desired.
\end{proof}

\end{document}